\documentclass[sigconf]{acmart}
\usepackage{subfig}

\usepackage{xcolor}       
\usepackage{colortbl}

\usepackage{makecell}
\usepackage{amsmath}

\usepackage{amsfonts}
\usepackage{amssymb}
\usepackage{mathtools}
\usepackage{amsthm}
\usepackage{url}
\usepackage{mathrsfs}
\usepackage{multirow}
\usepackage{pifont}

\theoremstyle{plain}
\newtheorem{theorem}{Theorem}[section]

\newtheorem{lemma}[theorem]{Lemma}
\newtheorem{corollary}[theorem]{Corollary}
\theoremstyle{definition}

\theoremstyle{remark}

\setcopyright{acmlicensed}
\copyrightyear{2025}
\acmYear{2025}
\setcopyright{acmlicensed}\acmConference[MM '25]{Proceedings of the 33rd ACM International Conference on Multimedia}{October 27--31, 2025}{Dublin, Ireland}
\acmBooktitle{Proceedings of the 33rd ACM International Conference on Multimedia (MM '25), October 27--31, 2025, Dublin, Ireland}
\acmDOI{10.1145/3746027.3754825}
\acmISBN{979-8-4007-2035-2/2025/10}

\begin{document}

\title{ELFATT: Efficient Linear Fast Attention for Vision Transformers}

\author{Chong Wu}
\authornote{Corresponding author.}
\orcid{0000-0003-3405-742X}
\affiliation{
  \institution{Department of Electrical Engineering and Centre for Intelligent Multidimensional Data Analysis, City University of Hong Kong}
  \city{Kowloon}
  \country{Hong Kong}
}
\email{imroxaswc@gmail.com}

\author{Maolin Che}
\orcid{0000-0001-9956-062X}
\authornote{These three authors contributed equally to this research.}
\affiliation{
\institution{School of Mathematics and Statistics and State Key Laboratory of Public Big Data, Guizhou University}
  \city{Guiyang}
  \country{China}}
\affiliation{
\institution{Department of Electrical Engineering and Centre for Intelligent Multidimensional Data Analysis, City University of Hong Kong}
\city{Kowloon}
  \country{Hong Kong}
}  
\email{mlche@gzu.edu.cn}

\author{Renjie Xu}
\orcid{0000-0002-6951-7273}
\authornotemark[2]
\affiliation{
  \institution{Department of Electrical Engineering and Centre for Intelligent Multidimensional Data Analysis, City University of Hong Kong}
  \city{Kowloon}
  \country{Hong Kong}
}
\email{harryxu950510@gmail.com}

\author{Zhuoheng Ran}
\orcid{0000-0002-0446-6276}
\authornotemark[2]
\affiliation{
  \institution{Department of Electrical Engineering and Centre for Intelligent Multidimensional Data Analysis, City University of Hong Kong}
  \city{Kowloon}
  \country{Hong Kong}}
\email{zhuoheran2-c@my.cityu.edu.hk}

\author{Hong Yan}
\orcid{0000-0001-9661-3095}
\affiliation{
  \institution{Department of Electrical Engineering and Centre for Intelligent Multidimensional Data Analysis, City University of Hong Kong}
  \city{Kowloon}
  \country{Hong Kong}}
\email{h.yan@cityu.edu.hk}

\renewcommand{\shortauthors}{Chong Wu, Maolin Che, Renjie Xu, Zhuoheng Ran, \& Hong Yan}

\begin{abstract}
The attention mechanism is the key to the success of transformers in different machine learning tasks. However, the quadratic complexity with respect to the sequence length of the vanilla softmax-based attention mechanism becomes the major bottleneck for the application of long sequence tasks, such as vision tasks. Although various efficient linear attention mechanisms have been proposed, they need to sacrifice performance to achieve high efficiency. What's more, memory-efficient methods, such as FlashAttention-1-3, still have quadratic computation complexity which can be further improved. In this paper, we propose a novel efficient linear fast attention (ELFATT) mechanism to achieve low memory input/output operations, linear computational complexity, and high performance at the same time. ELFATT offers 4-7x speedups over the vanilla softmax-based attention mechanism in high-resolution vision tasks without losing performance. ELFATT is FlashAttention friendly. Using FlashAttention-2 acceleration, ELFATT still offers 2-3x speedups over the vanilla softmax-based attention mechanism on high-resolution vision tasks without losing performance. Even in some non-vision tasks of long-range arena, ELFATT still achieves leading performance and offers 1.2-2.3x speedups over FlashAttention-2. Even on edge GPUs, ELFATT still offers 1.6x to 2.0x speedups compared to state-of-the-art attention mechanisms in various power modes from 5W to 60W. Furthermore, ELFATT can be used to enhance and accelerate diffusion tasks directly without training.
\end{abstract}



\maketitle

\section{Introduction}

Transformers have achieved great success in large language models (ChatGPT \cite{GPT} and Llama \cite{Llama}) and large vision models (SAM \cite{SAM} and Sora \cite{Sora}). The core technique of the transformer, the vanilla softmax-based attention (VaniATT) mechanism, is capable of capturing the relationship between any two tokens \cite{CURSA}. In complex tasks, the sequence length becomes longer and longer, usually much longer than the embedding dimension. The softmax operation after the multiplication of query and key matrices makes VaniATT quadratic computational complexity with respect to the sequence length. It has become a main obstacle to computational efficiency. 

Acceleration methods to speed up the attention computation can be classified into two categories: (1) memory-efficient methods; (2) computation-efficient methods. Memory-efficient methods focus on optimizing memory input/output (I/O) operations to achieve almost linear complexity \cite{Flash1,Flash2,Flash3,FlashSigmoid}. FlashAttention \cite{Flash1,Flash2,Flash3} and FlashSigmoid \cite{FlashSigmoid} are representatives of memory-efficient methods. Computation-efficient methods focus on minimizing the computation bound of the attention computation by using linear approximations based on different kernel methods \cite{flatten,PEF,ReLUT,TRNN,HATT,SOFT,RFA,COS,2SOFT,AGENTATT,fastvit2} or low-rank decomposition \cite{MPIPD,NTRANS,SOFT,CURSA}, and sparse computation \cite{Swin,CSWin,LIF,SPARSE1,SPARSE2,SPARSE3,SPARSE4,SPARSE5,SPARSE6}. However, memory-efficient methods still have quadratic computational complexity. And computation-efficient methods usually have lower performance compared to VaniATT \cite{Efficientvit2}. 

In this paper, we propose a novel efficient linear fast attention (ELFATT\footnote{The \textbf{ELFATT: Efficient Linear Fast Attention for Vision Transformers} methodology disclosed in this work is protected under intellectual property rights, including the following patent applications and granted patents: U.S. Non-Provisional Patent Application (No.~19/011,779 filed on 7 January 2025), Hong Kong Short-Term Patent (No.~HK30118021A2 granted on 30 May 2025), and Chinese Invention Patent Application (No.~202510108088.7 filed on 23 January 2025). Commercial inquiries should be directed to \href{mailto:kt@innocimda.com}{kt@innocimda.com}.}) mechanism that has low memory I/O operations and linear computational complexity and maintains noninferior performance compared to VaniATT. The core idea of ELFATT is the combination of sparse computation with a linear approximation. Each ELFATT block has two parallel attention heads. One head is used to compute sparse blockify attention to introduce inductive biases, and the other one head is used to compute global linear attention to capture long-range dependencies. Both heads have almost linear complexity, and the sparse blockify attention head can be further speeded up by using the FlashAttention \cite{Flash1,Flash2,Flash3} mechanisms to reduce memory I/O operations. ELFATT is evaluated on different vision tasks: image classification, semantic segmentation, object detection, and diffusion. Furthermore, ELFATT is also evaluated on edge GPUs and in the long-range arena. Compared to state-of-the-art memory-efficient acceleration methods and computation-efficient acceleration methods, ELFATT inherits advantages of both two kinds of methods: noninferior performance compared to VaniATT, low memory I/O operations, and linear computational complexity. In summary, this paper has the following contributions.
\begin{itemize}
  \item [(i)] 
  A novel efficient linear fast attention (ELFATT) mechanism is proposed that has low memory I/O operations and linear computational complexity and maintains noninferior performance compared to VaniATT.
  \item [(ii)] 
  ELFATT is the first sparse blockify attention + global linear attention design for attention acceleration which is significantly different from the existing hierarchical sparse attention mechanisms.
  \item [(iii)]
  The relationship between ELFATT and VaniATT is analyzed and given. An upper bound is given for the use of ELFATT to approximate VaniATT. To the best of our knowledge, we are the first to provide a theoretical bound analysis to sparse attention.
  \item [(iv)]
  ELFATT offers 4-7x speedups over VaniATT without using FlashAttention-2 and 2-3x speedups over VaniATT using FlashAttention-2 on high-resolution vision tasks.
  \item [(v)]
  ELFATT can be used to accelerate non-vision tasks with 1K-4K sequence lengths and offers 1.2-2.3x speedups over FlashAttention-2.
  \item [(vi)]
  ELFATT offers 1.6x to 2.0x speedups over state-of-the-art attention mechanisms in various power modes from 5W to 60W on edge GPUs.
  \item [(vii)]
  ELFATT can be used to enhance and accelerate diffusion tasks directly without training.
  \end{itemize}

\section{Related Work}
Since VaniATT has quadratic complexity in both time and memory, depending on whether the focus is on optimizing memory complexity or time complexity, attention acceleration methods can be categorized into the following two types: (A) memory-efficient attention acceleration methods; (B) computation-efficient attention acceleration methods.

\subsection{Memory-Efficient Attention Acceleration Methods}
Memory-efficient attention acceleration methods focus on optimizing quadratic memory complexity to almost linear complexity. One of the most representative methods is FlashAttention \cite{Flash1}, which minimizes memory I/O access for the softmax computation of the product of query and key matrices to improve the utilization rate of GPUs to achieve almost linear memory complexity \cite{Flash1,CURSA}. FlashAttention-2 further reduces the number of floating point operations (FLOPs) of non-matrix multiplication, parallelizes both forward and backward processes according to the sequence length, and reduces the I/O access of shared memory \cite{Flash2}. FlashAttention-3 is proposed to fully utilize the performance of new Hopper GPUs by introducing warp-specialization, interleave block-wise matrix multiplication and softmax operations, and the support of FP8 precision \cite{Flash3}. FlashSigmoid investigates the feasibility of using the sigmoid function to replace the softmax function in attention computation, proposes a new regularity method for sigmoid attention to stabilize training, and introduces a memory-efficient version based on FlashAttention-2 \cite{FlashSigmoid}. However, these memory-efficient attention acceleration methods are usually hard aware and can not support all kinds of GPUs and their computation complexity is still quadratic.

\subsection{Computation-Efficient Attention Acceleration Methods}
Different to hard aware memory-efficient attention acceleration methods, computation-efficient attention acceleration methods focus on optimizing computational complexity by linear approximations and usually can be categorized into two classes: (a) kernel methods; (b) sparse methods. 

Kernel methods usually change the order of nonlinear normalization and multiplication of the query matrix and key matrix to achieve linear computational complexity. FLatten \cite{flatten} introduces a cubic rectified linear unit (ReLU) \cite{ReLU} based feature map and uses it to perform nonlinear normalization on the query matrix and key matrix, respectively, before the attention computation. The linear transformer \cite{TRNN} introduces a feature map based on the exponential linear unit (ELU) \cite{elus} activation function. cosFormer \cite{COS} introduces a \emph{cos}-based nonlinear feature map to obtain a linear approximation of VaniATT. Efficient attention \cite{2SOFT} performs softmax normalization on the query matrix and the key matrix, respectively, to achieve linear complexity. However, this linearization introduces noise from two aspects: (1) Separated softmax normalization will reduce the similarity of elements in the query matrix and key matrix that are both negative at the same position; (2) Large negative values are suppressed to small positive values. The above two types of noise result in decreased discrimination between tokens and incorrect concentration of attention. Agent \cite{AGENTATT} introduces a small agent matrix, which is obtained by performing the pooling operation on the query matrix. This agent matrix is used as the auxiliary key matrix to be multiplied with the query matrix to reduce its dimenison before the softmax normalization. Similarly, it is also used to reduce the dimension of the key matrix before its softmax normalization. This agent matrix serves as a bridge between two independent softmax normalization processes to reduce noise. Another type of kernel methods introduce the low-rank decomposition \cite{xu1,xu2,xu3}. They usually take the softmax normalization of the product of query and key matrices as a whole for decomposition to find appropriate feature maps to reduce complexity \cite{CURSA}. Nystr{\" o}mformer \cite{NTRANS} introduces Nyström approximation to perform low-rank decomposition of the softmax normalization of the product of query and key matrices. For calculating an attention score matrix, it needs to perform an iterative inverse approximation and operate softmax normalization three times (each softmax normalization still involves a product of two matrices), which makes its acceleration effect not significant. Skeleton decomposition-based self-attention uses a simplified inverse approximation method based on the permuted diagonal matrix \cite{MPIPD}. Interactive multihead self-attention (iMHSA) introduces several linear layers to perform head interactions instead of performing an iterative inverse approximation \cite{Hilinear}. CUR decomposition based self-attention (CURSA) \cite{CURSA} introduces the CUR decomposition to replace the Nystr{\" o}m approximation and reduces the number of matrix multiplication. Mamba \cite{gu2024mamba} and its variants \cite{mamba2,Vmamba,vim} are also a kind of kernelized attention mechanism.

Sparse methods usually separate the original sequences into smaller blocks using sliding windows and perform attention computation within these blocks to reduce complexity. Each token in a sequence is only affected by several tokens, not all tokens; hence, this kind of sparse method is also called the local attention mechanism. Swin \cite{Swin} and Longformer \cite{SPARSE2} are pioneers in introducing sliding windows into attention computation for vision tasks and language tasks, respectively. To address the information loss introduced by the local windows, Swin introduces a shifted window mechanism to capture cross-block information, while Longformer selects some tokens as global tokens, which have effects on all tokens, to compensate for the global information loss. NesT \cite{NesT} further simplifies the shifted window mechanism through simple spatial operations. CSWin \cite{CSWin} introduces a cross-shaped window mechanism for 2D sequences. The cross-shaped window mechanism separates a 2D sequence into horizontal and vertical stripes, respectively, and performs parallel attention computation within these stripes in horizontal and vertical directions, which can simultaneously obtain local inductive biases and cross-block (global) information. In addition to using the deterministic global token selection mechanism of Longformer, Big Bird \cite{SPARSE5} also randomly selects some global tokens to enhance performance. Similarly, the sparse transformer \cite{SPARSE3} proposed by OpenAI also introduces several fixed-step tokens, which is similar to dilated convolution \cite{atrous} in convolutional neural networks (CNNs) to capture long-range information. Reformer \cite{SPARSE6} finds the nearest neighbors for each token to calculate local attention scores using the locality-sensitive-hashing (LSH) algorithm.

\section{Methods}
\subsection{Vanilla Softmax-Based Attention Mechanism}
For clarity, we assume that all vectors appear in this paper are row vectors. Scaling and normalization are omitted in this paper for convenience. For any two same sized tokens $\textit{\textbf{x}} \in \mathbb{R}^{c}$ and $\textit{\textbf{y}} \in \mathbb{R}^{c}$, their attention similarity score can be calculated as follows,
\begin{equation}
a = {\rm exp}\left(\textit{\textbf{x}}\textit{\textbf{y}}^{\top}\right),
\label{sa}
\end{equation}
where ${\rm exp}\left(\cdot\right)$ is an element-wise exponential function and is conducted after the inner product of two tokens. For given two $m$-tuples $\{\textit{\textbf{x}}_1,\dots,\textit{\textbf{x}}_m\}$ and $\{\textit{\textbf{y}}_1,\dots,\textit{\textbf{y}}_m\}$ with $\textit{\textbf{x}}_i,\textit{\textbf{y}}_i\in\mathbb{R}^c$ and $i=1,2,\dots,m$, we need to calculate $m^2$ pairs of attention similarity scores. Currently, a long sequence length $m$ is preferred in large models and $m \gg c$, therefore, the computational complexity of VaniATT is quadratic with respect to the sequence length $m$.

\subsection{General Attention Mechanism}
Eq. (\ref{sa}) can be rewritten as a more general form \cite{TRNN} as follows,
\begin{equation}
a = {\rm sim}\left(\textit{\textbf{x}}, \textit{\textbf{y}}\right),
\label{sage1}
\end{equation}
where ${\rm sim}(\cdot)$ is a non-negative similarity function and it satisfies the definition of a kernel function $G(\textit{\textbf{x}}, \textit{\textbf{y}}): \mathbb{R}^c \times \mathbb{R}^c \rightarrow\mathbb{R}_{+}$ with $\mathbb{R}_{+}=\{x\in\mathbb{R}:x\geq 0\}$. 

\subsection{Kernelized Attention Mechanism}
If we have such a kernel with a non-negative feature map $\phi$, Eq. (\ref{sage1}) can be written as follows,
\begin{equation}
a = \phi\left(\textit{\textbf{x}}\right)\phi\left(\textit{\textbf{y}}\right)^{\top}.
\label{sage2}
\end{equation}
Performer \cite{PEF} has proved that one of the best choices of the non-negative feature map $\phi$ for Eq. (\ref{sage2}) is $\rm exp\left(\cdot\right)$. Performer obtains an exact alternative of Eq. (\ref{sa}) using random feature maps as follows,
\begin{equation}
{\rm exp}\left(\textit{\textbf{x}}\textit{\textbf{y}}^{\top}\right) = \mathbb{E}_{\omega \sim \mathcal{N}(\textbf{0}_c, \textbf{I}_c)} \left\lbrack{\rm exp}\left(\omega\textit{\textbf{x}}^{\top}\right){\rm exp}\left(\omega\textit{\textbf{y}}^{\top}\right)e\left(\textit{\textbf{x}}\right)e\left(\textit{\textbf{y}}\right)\right\rbrack,
\label{pef}
\end{equation}
where $e\left(\textit{\textbf{x}}\right) = {\rm exp}\left(-\frac{\|\textit{\textbf{x}}\|^{2}}{2}\right)$, $e\left(\textit{\textbf{y}}\right) = {\rm exp}\left(-\frac{\|\textit{\textbf{y}}\|^{2}}{2}\right)$, $\textbf{0}_c \in \mathbb{R}^{c}$ is the zero vector, and $\textbf{I}_c \in \mathbb{R}^{c\times c}$ is the identity matrix. If the attention scores for all pairs of $\textit{\textbf{x}}$ and $\textit{\textbf{y}}_{i}$ ($i=1,2,\dots,m$) using Eq. (\ref{pef}) are obtained, after the following normalization $e\left(\textit{\textbf{x}}\right)$ will be canceled.
\begin{equation*}
\begin{split}
&\frac{{\rm exp}\left(\textit{\textbf{x}}\textit{\textbf{y}}_{i}^{\top}\right)}{\sum_{j=1}^{m}{{\rm exp}\left(\textit{\textbf{x}}\textit{\textbf{y}}_{j}^{\top}\right)}} \\
&= \frac{\mathbb{E}_{\omega \sim \mathcal{N}(\textbf{0}_c, \textbf{I}_c)} \left\lbrack{\rm exp}\left(\omega\textit{\textbf{x}}^{\top}\right){\rm exp}\left(\omega\textit{\textbf{y}}_{i}^{\top}\right)e\left(\textit{\textbf{x}}\right)e\left(\textit{\textbf{y}}_{i}\right)\right\rbrack}{\sum_{j=1}^{m}{\mathbb{E}_{\omega \sim \mathcal{N}(\textbf{0}_c, \textbf{I}_c)} \left\lbrack{\rm exp}\left(\omega\textit{\textbf{x}}^{\top}\right){\rm exp}\left(\omega\textit{\textbf{y}}_{j}^{\top}\right)e\left(\textit{\textbf{x}}\right)e\left(\textit{\textbf{y}}_{j}\right)\right\rbrack}} \\
&= \frac{\mathbb{E}_{\omega \sim \mathcal{N}(\textbf{0}_c, \textbf{I}_c)} \left\lbrack{\rm exp}\left(\omega\textit{\textbf{x}}^{\top}\right){\rm exp}\left(\omega\textit{\textbf{y}}_{i}^{\top}\right)e\left(\textit{\textbf{y}}_{i}\right)\right\rbrack}{\sum_{j=1}^{m}{\mathbb{E}_{\omega \sim \mathcal{N}(\textbf{0}_c, \textbf{I}_c)} \left\lbrack{\rm exp}\left(\omega\textit{\textbf{x}}^{\top}\right){\rm exp}\left(\omega\textit{\textbf{y}}_{j}^{\top}\right)e\left(\textit{\textbf{y}}_{j}\right)\right\rbrack}}.
\end{split}
\end{equation*}
Hence, $e\left(\textit{\textbf{x}}\right)$ has no effect on the final attention score. Eq. (\ref{pef}) can be written as follows,
\begin{equation}
{\rm exp}\left(\textit{\textbf{x}}\textit{\textbf{y}}^{\top}\right) \approx \mathbb{E}_{\omega \sim \mathcal{N}(\textbf{0}_c, \textbf{I}_c)} \left\lbrack{\rm exp}\left(\omega\textit{\textbf{x}}^{\top}\right){\rm exp}\left(\omega\textit{\textbf{y}}^{\top}\right)e\left(\textit{\textbf{y}}\right)\right\rbrack.
\label{pef1}
\end{equation}
After the normalization, the left and right-hand sides of Eq. (\ref{pef1}) will be equal. Eq. (\ref{pef}) only holds when taking the sum of an infinite number of random vectors $\omega$. To avoid performing summation of infinite terms, Performer samples $c\times \log(c)$ random vectors $\omega$ to ensure a low approximation error. If the change of $e\left(\textit{\textbf{y}}\right)$ is much smaller than ${\rm exp}\left(\omega\textit{\textbf{x}}^{\top}\right){\rm exp}\left(\omega\textit{\textbf{y}}^{\top}\right)$, Eq. (\ref{pef1}) can be further simplified and approximated as follows,
\begin{equation}
{\rm exp}\left(\textit{\textbf{x}}\textit{\textbf{y}}^{\top}\right) 
\approx \mathbb{E}_{\omega \sim \mathcal{N}(\textbf{0}_c, \textbf{I}_c)} \left\lbrack{\rm exp}\left(\omega\textit{\textbf{x}}^{\top}\right){\rm exp}\left(\omega\textit{\textbf{y}}^{\top}\right)\right\rbrack.
\label{pef2}
\end{equation}
Efficient attention (EFFATT) \cite{2SOFT} is a special case of Eq. (\ref{pef2}). It simplifies $\omega$ to a one-hot vector and only uses $c$ one-hot vectors in Eq. (\ref{pef2}) to obtain an approximation as follows,
\begin{equation}
{\rm exp}\left(\textit{\textbf{x}}\textit{\textbf{y}}^{\top}\right) 
\approx \frac{1}{c} {\rm exp}\left(\textit{\textbf{x}}\right){\rm exp}\left(\textit{\textbf{y}}\right)^{\top}.
\label{effatt}
\end{equation}
Its approximation error was studied and obtained in \cite{CURSA}. Eq. (\ref{effatt}) has a problem of concentration reduction of attention maps \cite{CURSA} and Performer also has this problem when the number of $\omega$ sampled is too small. When the number of $\omega$ sampled is too large, Performer may be slower than VaniATT. 

\subsection{Efficient Linear Fast Attention Mechanism}
To address the defects of Eq. (\ref{effatt}), in this paper, we propose a novel attention mechanism as follows,

\begin{small}
\begin{equation}
\begin{split}
{\rm exp}&\left(\textit{\textbf{Q}}\textit{\textbf{K}}^{\top}\right)\textit{\textbf{V}} \approx\left\lbrack{\rm exp}\left(\bar{\textit{\textbf{Q}}}\right){\rm exp}\left(\bar{\textit{\textbf{K}}}\right)^{\top}\bar{\textit{\textbf{V}}},  g\left({\rm exp}\left(f(\tilde{\textit{\textbf{Q}}}) f(\tilde{\textit{\textbf{K}}})^{\top}\right)f(\tilde{\textit{\textbf{V}}})\right)\right\rbrack,
\end{split}
\label{elfatt}
\end{equation}
\end{small}

\noindent where $\textit{\textbf{Q}} = \lbrack\bar{\textit{\textbf{Q}}},\tilde{\textit{\textbf{Q}}}\rbrack \in \mathbb{R}^{m \times c}$, $\textit{\textbf{K}} = \lbrack\bar{\textit{\textbf{K}}},\tilde{\textit{\textbf{K}}}\rbrack \in \mathbb{R}^{m \times c}$, $\textit{\textbf{V}} = \lbrack\bar{\textit{\textbf{V}}},\tilde{\textit{\textbf{V}}}\rbrack \in \mathbb{R}^{m \times c}$, $\bar{\textit{\textbf{Q}}}\in\mathbb{R}^{m \times c_1}$, $\tilde{\textit{\textbf{Q}}}\in\mathbb{R}^{m \times c_2}$, $\bar{\textit{\textbf{K}}}\in\mathbb{R}^{m \times c_1}$, $\tilde{\textit{\textbf{K}}}\in\mathbb{R}^{m \times c_2}$, $\bar{\textit{\textbf{V}}}\in\mathbb{R}^{m \times c_1}$, $\tilde{\textit{\textbf{V}}}\in\mathbb{R}^{m \times c_2}$, $c=c_1+c_2$, $f\left(\cdot\right)$ is a blockify function to separate a matrix with a size of $m\times c_2$ into $b$ blocks (each block has a size of $(m/b)\times c_2$), and $g\left(\cdot\right)$ is an unblockify function to unblock $b$ blocks to a single matrix with a size of $m\times c_2$. 

\noindent Eq. (\ref{elfatt}) denotes the single-head VaniATT mechanism approximated by using ELFATT. From the right-hand side of Eq. (\ref{elfatt}), each ELFATT attention process consists of two heads. Hence, ELFATT can directly become an approximation of the double-head VaniATT mechanism as follows,

\begin{small}
\begin{equation}
\begin{split}
\left\lbrack{\rm exp}\left(\bar{\textit{\textbf{Q}}}\bar{\textit{\textbf{K}}}^{\top}\right)\bar{\textit{\textbf{V}}}, {\rm exp}\left(\tilde{\textit{\textbf{Q}}}\tilde{\textit{\textbf{K}}}^{\top}\right)\tilde{\textit{\textbf{V}}}\right\rbrack \approx &\left\lbrack{\rm exp}\left(\bar{\textit{\textbf{Q}}}\right){\rm exp}\left(\bar{\textit{\textbf{K}}}\right)^{\top}\bar{\textit{\textbf{V}}}, \right.\\
&\left.  g\left({\rm exp}\left(f(\tilde{\textit{\textbf{Q}}}) f(\tilde{\textit{\textbf{K}}})^{\top}\right)f(\tilde{\textit{\textbf{V}}})\right)\right\rbrack.
\end{split}
\label{elfatt2}
\end{equation}
\end{small}

Since each ELFATT block corresponds to two parallel heads as Fig. \ref{cursaover} shows, $s$ ($s > 0$) ELFATT will be needed for the approximation of $2\times s$ heads of VaniATT according to Eq. (\ref{elfatt2}). For the approximation of $2\times s - 1$ heads of VaniATT, $2\times s - 1$ ELFATT will be needed according to Eq. (\ref{elfatt}). The approximation error bound analysis of Eqs. (\ref{elfatt}) and (\ref{elfatt2}) is available in Section \ref{AEB} and Fig. \ref{boundana} of Appendix.

\subsection{Positional Encoding}
Different positional encoding mechanisms such as absolute positional encoding \cite{ATA}, relative positional encoding \cite{RPE}, and conditional positional encoding \cite{CPE}, have been proposed to make use of the ordering information of sequences in vision transformers \cite{Swin,VRPE,CPE}. Among these positional encoding mechanisms, the locally enhanced positional encoding (LePE) \cite{CSWin} mechanism shows more powerful local positional information enhancement and brings a higher performance gain for vision transformers. Hence, LePE is selected as the positional encoding mechanism for ELFATT. After introducing LePE, Eq. (\ref{elfatt}) will become as follows,
\begin{equation}
\begin{split}
{\rm exp}\left(\textit{\textbf{Q}}\textit{\textbf{K}}^{\top}\right)\textit{\textbf{V}} + &L\left(\textit{\textbf{V}}\right) \approx \left\lbrack{\rm exp}\left(\bar{\textit{\textbf{Q}}}\right){\rm exp}\left(\bar{\textit{\textbf{K}}}\right)^{\top}\bar{\textit{\textbf{V}}}+L\left(\bar{\textit{\textbf{V}}}\right),\right.  \\
&\left.g\left({\rm exp}\left(f(\tilde{\textit{\textbf{Q}}}) f(\tilde{\textit{\textbf{K}}})^{\top}\right)f(\tilde{\textit{\textbf{V}}})+L\left(f(\tilde{\textit{\textbf{V}}})\right)\right)\right\rbrack,
\end{split}
\label{lepe1}
\end{equation}

\begin{figure}[htbp]
  \centering
    \subfloat[VaniATT]{\includegraphics[width=0.5694\linewidth]{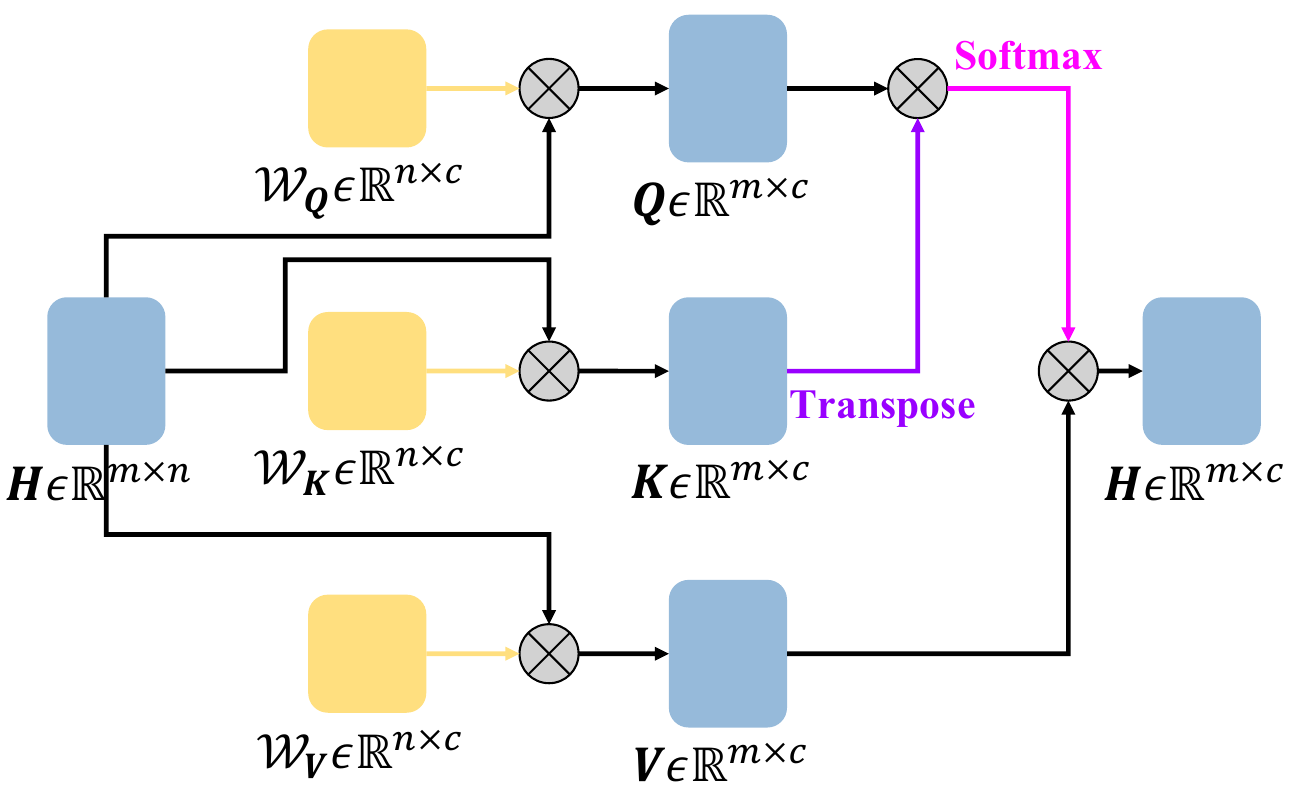}}
    \hfill
    \subfloat[ELFATT]{\includegraphics[width=1\linewidth]{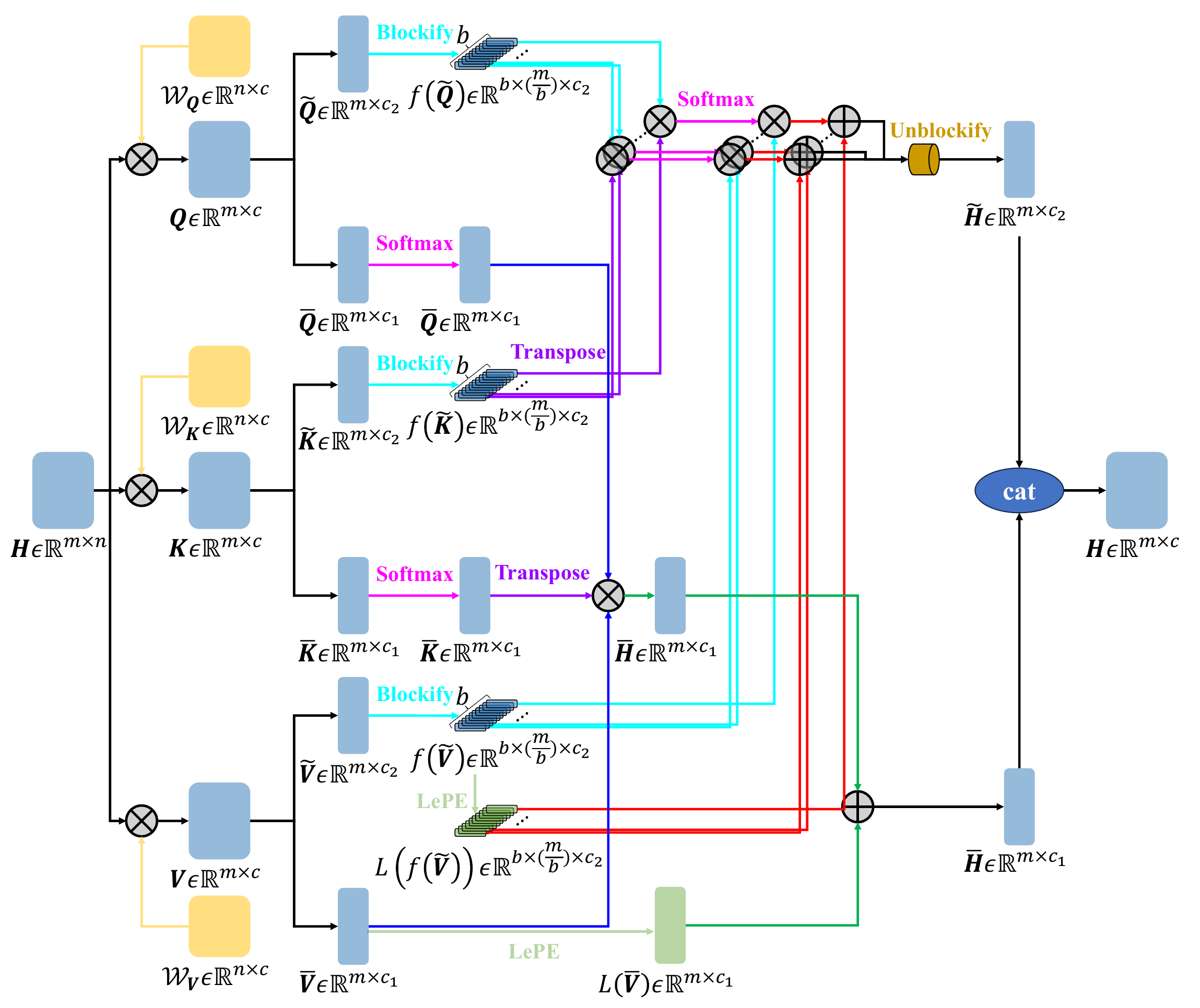}}
    \caption{The comparison of VaniATT and ELFATT. For an input embedding matrix $\textit{\textbf{H}}\in \mathbb{R}^{m \times n}$, each ELFATT block will process it in two parallel attention heads to obtain two new embedding matrices, $\bar{\textit{\textbf{H}}}\in\mathbb{R}^{m \times c_1}$ and $\tilde{\textit{\textbf{H}}}\in\mathbb{R}^{m \times c_2}$, respectively. After a concatenation operation ${\rm cat}$, the ELFATT block will output the updated embedding matrix $\textit{\textbf{H}} = \lbrack\bar{\textit{\textbf{H}}},\tilde{\textit{\textbf{H}}}\rbrack \in \mathbb{R}^{m \times c}$.}
    \label{cursaover}
\end{figure}

\noindent Eq. (\ref{elfatt2}) will become as follows,
\begin{equation}
\begin{split}
&\left\lbrack{\rm exp}\left(\bar{\textit{\textbf{Q}}}\bar{\textit{\textbf{K}}}^{\top}\right)\bar{\textit{\textbf{V}}}+L\left(\bar{\textit{\textbf{V}}}\right), {\rm exp}\left(\tilde{\textit{\textbf{Q}}}\tilde{\textit{\textbf{K}}}^{\top}\right)\tilde{\textit{\textbf{V}}}+L\left(\tilde{\textit{\textbf{V}}}\right)\right\rbrack\approx \\
&\left\lbrack{\rm exp}\left(\bar{\textit{\textbf{Q}}}\right){\rm exp}\left(\bar{\textit{\textbf{K}}}\right)^{\top}\bar{\textit{\textbf{V}}}+L\left(\bar{\textit{\textbf{V}}}\right), g\left({\rm exp}\left(f(\tilde{\textit{\textbf{Q}}}) f(\tilde{\textit{\textbf{K}}})^{\top}\right)f(\tilde{\textit{\textbf{V}}})+L\left(f(\tilde{\textit{\textbf{V}}})\right)\right)\right\rbrack,
\end{split}
\label{lepe2}
\end{equation}
where $L(\cdot)$ denotes a depthwise convolution operation with a kernel size of 3. Fig. \ref{cursaover} shows the comparison of VaniATT and ELFATT. $\mathcal{W}_{\textit{\textbf{Q}}}$, $\mathcal{W}_{\textit{\textbf{K}}}$, and $\mathcal{W}_{\textit{\textbf{V}}} \in \mathbb{R}^{n \times c}$ are three projection matrices. For an input embedding matrix $\textit{\textbf{H}}\in \mathbb{R}^{m \times n}$, each ELFATT block will process it in two parallel attention heads to obtain two new embedding matrices, $\bar{\textit{\textbf{H}}}\in\mathbb{R}^{m \times c_1}$ and $\tilde{\textit{\textbf{H}}}\in\mathbb{R}^{m \times c_2}$, respectively. After a concatenation operation, the ELFATT block will output the updated embedding matrix $\textit{\textbf{H}} = \lbrack\bar{\textit{\textbf{H}}},\tilde{\textit{\textbf{H}}}\rbrack \in \mathbb{R}^{m \times c}$.

\begin{table}[!t]
  \centering
  \caption{The comparison of top-1 test accuracy (Acc.), inference throughput (FPS), parameter numbers (\#), and number of floating point operations (FLOPs) of different methods on ImageNet-1K. Note: Inference throughput is obtained using a batch size of 512 for tiny models, and 256/32 for base models with a resolution of $224^2$/$384^2$ using mixed precision on a single NVIDIA H20 (96 GB) GPU. ``---'' denotes the corresponding method cannot be accelerated by FlashAttention-2. ``Res.'' denotes resolution, ``imgs'' denotes images, and ``nFA/FA'' denotes without/with using FlashAttention-2.}
    \tiny
    \begin{tabular}{lrrrrrr}
    \toprule
    Method & \multicolumn{1}{r}{Res.} & \multicolumn{1}{r}{Acc. (\%)}   & \multicolumn{1}{r}{FPS (nFA/FA)} & \multicolumn{1}{r}{\#} & \multicolumn{1}{r} {FLOPs (nFA/FA)} \\
    \midrule
    \multicolumn{6}{c}{CSWin-B}\\
    \midrule
    Agent                                        &$224^2$& 84.7                                              &    930/994 imgs/s         &  73M     &  14.92G/14.49G     \\ 
    EFFATT                                       &$224^2$& 84.4                                              &   985/1059 imgs/s         &  73M     &  14.98G/14.53G     \\ 
    \rowcolor{gray!25}
    ELFATT                                       &$224^2$& 84.7                                              &  1000/1187 imgs/s         &  73M     &  15.47G/14.46G     \\ 
    FLatten                                      &$224^2$& 84.5                                              &    814/864 imgs/s         &  75M     &  14.96G/14.52G     \\ 
    GLOBAL                                       &$224^2$& 84.7                                              &    478/879 imgs/s         &  73M     &  22.33G/14.39G     \\
    LOCAL                                        &$224^2$& 84.4                                              &   941/1037 imgs/s         &  73M     &  15.03G/14.39G     \\ 
    \midrule
    Agent                                        &$384^2$& 85.8                                              &    246/276 imgs/s         &  73M     &  46.33G/42.57G     \\ 
    EFFATT                                       &$384^2$& 85.7                                              &    294/331 imgs/s         &  73M     &  46.57G/42.69G     \\ 
    \rowcolor{gray!25}
    ELFATT                                       &$384^2$& 85.8                                              &    272/355 imgs/s         &  73M     &  51.21G/42.48G     \\ 
    FLatten                                      &$384^2$& 85.5                                              &    238/266 imgs/s         &  78M     &  46.43G/42.67G     \\ 
    GLOBAL                                       &$384^2$& 85.9                                              &     82/201 imgs/s         &  73M     & 110.89G/42.28G     \\ 
    LOCAL                                        &$384^2$& 85.5                                              &    276/323 imgs/s         &  73M     &  47.06G/42.28G     \\ 
    \midrule
    \multicolumn{6}{c}{CSWin-T}\\
    \midrule
    Agent                                        &$224^2$& 83.1                                              &   2297/2425 imgs/s         &  20M     &  4.31G/4.14G       \\
    EFFATT                                       &$224^2$& 82.6                                              &   2394/2526 imgs/s         &  20M     &  4.35G/4.17G       \\
    \rowcolor{gray!25}
    ELFATT                                       &$224^2$& 83.1                                              &   2603/2856 imgs/s         &  20M     &  4.44G/4.13G       \\
    FLatten                                      &$224^2$& 83.1                                              &   1934/2025 imgs/s         &  21M     &  4.34G/4.16G       \\
    GLOBAL                                       &$224^2$& 83.1                                              &   1303/2210 imgs/s         &  20M     &  7.60G/4.09G       \\
    LOCAL                                        &$224^2$& 82.7                                              &   2330/2519 imgs/s         &  20M     &  4.36G/4.09G       \\
    \midrule
    \multicolumn{6}{c}{Swin-B}\\
    \midrule
    Agent                                              &$224^2$& 84.0                                              &  1367/\ ---\ \ \ imgs/s    &  88M     &  15.44G/\,\, ---\;\ \ \ \ \\ 
    EFFATT                                             &$224^2$& 84.2                                              &  1439/1536 imgs/s          &  88M     &  15.69G/15.33G       \\ 
    \rowcolor{gray!25}    
    ELFATT                                             &$224^2$& 84.5                                              &  1314/1497 imgs/s          &  91M     &  16.46G/15.68G       \\ 
    FLatten                                            &$224^2$& 83.8                                              &  1226/\ ---\ \ \ imgs/s    &  89M     &  15.46G/\,\, ---\;\ \ \ \ \\ 
    GLOBAL                                             &$224^2$& 84.2                                              &  743/1325 imgs/s           &  88M     &  21.57G/15.19G       \\ 
    LOCAL                                              &$224^2$& 83.5                                              &  1351/\ ---\ \ \ imgs/s    &  88M     &  15.47G/\,\, ---\;\ \ \ \ \\ 
    \midrule
    Agent                                              &$384^2$& 84.9                                              &  372/\,--- imgs/s            &  88M     &  46.34G/\,\, ---\;\ \ \ \  \\      
    EFFATT                                             &$384^2$& 85.3                                              &  433/481 imgs/s             &  88M     &  48.18G/45.04G             \\ 
    \rowcolor{gray!25}    
    ELFATT                                             &$384^2$& 85.5                                              &  372/457 imgs/s             &  91M     &  52.79G/46.08G             \\ 
    FLatten                                            &$384^2$& 85.0                                              &  353/\,--- imgs/s            &  91M     &  46.49G/\,\, ---\;\ \ \ \  \\ 
    GLOBAL                                             &$384^2$& 85.3                                              &  134/313 imgs/s             &  88M     &  99.76G/44.64G             \\ 
    LOCAL                                              &$384^2$& 84.5                                              &  357/\,--- imgs/s            &  88M     &  47.19G/\,\, ---\;\ \ \ \  \\ 
    \midrule
    \multicolumn{6}{c}{Swin-T}\\
    \midrule    
    Agent                                              &$224^2$& 82.6                                              &   2847/\ ---\ \ \ imgs/s   &  29M     &  4.53G/\ ---\ \ \ \ \\
    EFFATT                                             &$224^2$& 82.1                                              &   3165/3282 imgs/s         &  28M     &  4.55G/4.45G        \\
    \rowcolor{gray!25}    
    ELFATT                                             &$224^2$& 82.7                                              &   2884/3159 imgs/s         &  30M     &  4.99G/4.67G        \\
    FLatten                                            &$224^2$& 82.1                                              &   2502/\ ---\ \ \ imgs/s   &  29M     &  4.50G/\ ---\ \ \ \ \\
    GLOBAL                                             &$224^2$& 82.4                                              &   1269/2571 imgs/s         &  28M     &  8.81G/4.38G        \\
    LOCAL                                              &$224^2$& 81.4                                              &   2943/\ ---\ \ \ imgs/s   &  28M     &  4.51G/\ ---\ \ \ \ \\
    \midrule
    \multicolumn{6}{c}{Others}\\
    \midrule
    ConvNeXt-T                                                &$224^2$& 82.1                                              &   3911/\ ---\ \ \ imgs/s   &  29M     &  4.47G/\ ---\ \ \ \ \\
    VMamba-T                                                  &$224^2$& 82.5                                              &   1837/\ ---\ \ \ imgs/s   &  30M     &  4.84G/\ ---\ \ \ \ \\
   \bottomrule
    \end{tabular}
  \label{tab1}
\end{table}

\section{Complexity Analysis}
\label{ComANA}
According to Eqs. (\ref{lepe1}) and (\ref{lepe2}), the complexity of ELFATT can be divided into two parts: (a) the global linear attention head ${\rm exp}\left(\bar{\textit{\textbf{Q}}}\right){\rm exp}\left(\bar{\textit{\textbf{K}}}\right)^{\top}\bar{\textit{\textbf{V}}}+L\left(\bar{\textit{\textbf{V}}}\right)$; (b) the sparse blockify attention head $g\left({\rm exp}\left(f(\tilde{\textit{\textbf{Q}}}) f(\tilde{\textit{\textbf{K}}})^{\top}\right)f(\tilde{\textit{\textbf{V}}})+L\left(f(\tilde{\textit{\textbf{V}}})\right)\right)$. The complexity of the global linear attention head is $O(m \times c_1^2)$ in the case of $m>c_1$, and $O(m^2 \times c_1)$ in the case of $m\leq c_1$. The complexity of the sparse blockify attention head is $O((m^2/b) \times c_2)$. If $m/b \ll m$, the sparse blockify attention head will also achieve almost linear complexity. 

\section{Experiments and Results}
We evaluated ELFATT in commonly used vision tasks: image classification (ImageNet-1K \cite{IMGNET}), semantic segmentation (ADE20K \cite{ADE20K}), and object detection (MS COCO 2017 \cite{COCO}). We compared ELFATT with VaniATT \cite{ATA}, the memory-efficient attention mechanism (FlashAttention-2 \cite{Flash2}), local window-based attention mechanisms (Swin \cite{Swin} and CSWin \cite{CSWin}), and kernel-based attention mechanisms (Agent \cite{AGENTATT}, EFFATT \cite{2SOFT}, and FLatten \cite{flatten}). The backbone ViT architectures used to evaluate the different attention mechanisms are: Swin-T \cite{Swin}, Swin-B \cite{Swin}, CSWin-T24181 (CSWin-T) \cite{flatten}, and CSWin-B36292 (CSWin-B) \cite{flatten}. The original Swin-T, Swin-B, CSWin-T, and CSWin-B are called Swin-T-LOCAL, Swin-B-LOCAL, CSWin-T-LOCAL, and CSWin-B-LOCAL, respectively. The backbones after replacing local window-based attention mechanisms with VaniATT \cite{ATA} are called Swin-T-GLOBAL, Swin-B-GLOBAL, CSWin-T-GLOBAL, and CSWin-B-GLOBAL, respectively. The backbones after replacing local window-based attention mechanisms with kernel-based attention mechanisms (Agent \cite{AGENTATT}, EFFATT \cite{2SOFT}, and FLatten \cite{flatten}) are called Swin-T-Agent, Swin-T-EFFATT, Swin-T-FLatten, Swin-B-Agent, Swin-B-EFFATT, Swin-B-FLatten, CSWin-T-Agent, CSWin-T-EFFATT, CSWin-T-FLatten, CSWin-B-Agent, CSWin-B-EFFATT, and CSWin-B-FLatten, respectively. The backbones after replacing local window-based attention mechanisms with ELFATT are called Swin-T-ELFATT, Swin-B-ELFATT, CSWin-T-ELFATT, and CSWin-B-ELFATT, respectively. For the ImageNet-1K image classification task, we used the same training settings and data augmentation methods from \cite{CSWin} to train all models from scratch using mixed precision. For the ADE20K semantic segmentation task \cite{ADE20K} and the MS COCO 2017 object detection task \cite{COCO}, we used the same training settings and data augmentation methods from \cite{Vmamba} to fine-tune the weights of all models obtained from the ImageNet-1K image classification task using mixed precision. We also compared ELFATT with ConvNeXt-T \cite{ConvNeXt} and VMamba-T \cite{Vmamba} on ImageNet-1K, ADE20K, and MS COCO 2017. The experiments of ImageNet-1K, ADE20K, and MS COCO 2017 were carried out on 8 NVIDIA vGPU (32 GB) GPUs. Inference throughput comparison experiments were carried out on 1 NVIDIA H20 (96 GB) GPU. Besides high-performance computing applications, we also compared ELFATT with state-of-the-art attention mechanisms on edge GPUs (NVIDIA Jetson AGX Orin/NVIDIA Jetson Nano). Furthermore, we investigated how ELFATT can be used to accelerate diffusion tasks. To validate the performance of ELFATT on some non-vision long sequence tasks, we introduced long-range arena (LRA) benchmark \cite{LRA} and compared ELFATT with VaniATT \cite{ATA} and efficient attention mechanisms: Linformer \cite{LIF}, Nystr{\" o}mformer \cite{NTRANS}, Primal \cite{SVDATT}, and Reformer \cite{SPARSE6}. We followed the training protocol and the settings of \cite{SVDATT} and carried out the experiments on 1 NVIDIA Tesla A100 (40GB) GPU. The details of ablation studies, experiment settings of edge GPUs, and experiment settings and ablation studies of diffusion acceleration are available in Sections \ref{APAB1}, \ref{APAB2}, \ref{APAB3}, \ref{APAB4}, \ref{APAB5}, \ref{APEGPU}, and \ref{APSD} of Appendix. We used FlashAttention-2 \cite{Flash2} to speed up all models that are compatible with FlashAttention. The PyTorch implementation of ELFATT, including the detailed architecture specifications of Swin-T-ELFATT, Swin-B-ELFATT, CSWin-T-ELFATT and CSWin-B-ELFATT, is available \href{https://github.com/Alicewithrabbit/ELFATT}{at this URL}.

\subsection{Image Classification Performance}
Table \ref{tab1} shows the performance comparison of different attention mechanisms on ImageNet-1K. From Table \ref{tab1}, VaniATT (Swin-T-GLOBAL and CSWin-T-GLOBAL) can outperform most linear attention mechanisms using the same architecture. Only Agent and ELFATT achieve the most close performance compared to VaniATT in this paper. As to the inference speed comparison, the proposed ELFATT achieves the highest inference throughput (frame per second, FPS) than all other attention mechanisms when using CSWin-T as the backbone. ELFATT achieves almost the same speed as EFFATT and is significantly faster than other attention mechanisms when using Swin-T as the backbone. ELFATT offers almost a 2x speedup over VaniATT without using FlashAttention-2. VaniATT using FlashAttention-2 is still 0.1-0.2x slower than ELFATT without using FlashAttention-2. With the use of FlashAttention-2 to optimize memory operations, ELFATT can be further accelerated and offers 1.2-1.3x speedups over VaniATT. Fig. \ref{visualcom2} in Appendix shows the visual comparison of different attention mechanisms using CSWin-T as the backbone. In Fig. \ref{visualcom2} (a), only Agent, ELFATT, and VaniATT (GLOBAL) focus mainly on the true object: American Staffordshire Terrier. In Fig. \ref{visualcom2} (b), only ELFATT and GLOBAL focus mainly on the true object: Hartebeest. In Fig. \ref{visualcom2} (c), only EFFATT, ELFATT, GLOBAL, and LOCAL focus mainly on the true object: Zebra. In Fig. \ref{visualcom2} (d), only ELFATT and GLOBAL focus mainly on the true object: Library. In Fig. \ref{visualcom2} (e), all methods focus accurately on the true object: Snail. In Fig. \ref{visualcom2} (f), all methods focus accurately on the true object: Black-Footed Ferret, ELFATT covers more areas of the true object. In Fig. \ref{visualcom2} (g), all methods also focus accurately on the true object: Scottish Deerhound, ELFATT also covers more areas of the true object. In Fig. \ref{visualcom2} (h), only ELFATT focuses mainly on the true object: Pineapple. To sum up, ELFATT shows a more focused attention map on the ground-truth object. For the larger backbone, CSWin-B, ELFATT still achieves state-of-the-art (SOTA) performance. ELFATT offers a 2.1/1.4x speedup over VaniATT with/without using FlashAttention-2 for an input resolution of $224^2$ and a 3.3/1.8x speedup over VaniATT with/without using FlashAttention-2 for an input resolution of $384^2$. In addition, ELFATT achieves significantly leading performance under the backbone of Swin-B, offering a 1.8/2.8x speedup over VaniATT without using FlashAttention-2 for an input resolution of $224^2$/$384^2$. After using FlashAttention-2, it matches the speed of EFFATT which is a real linear attention mechanism. Also, as shown in Table \ref{tab1}, it can be seen that although some methods have lower FLOPs than ELFATT, their actual speed is slower than ELFATT which is consistent with the observations and conclusions of \cite{FLOPsslow}. Fig. \ref{acceff} in Appendix shows accuracy-efficiency curves obtained by different attention mechanisms on ImageNet-1K using different bacbones. Compared to VaniATT, ELFATT offers 2.2-4.3x speedups for the backbone of CSWin and 2.5-3.4x speedups for the backbone of Swin to achieve similar accuracy. Compared to FlashAttention-2, ELFATT offers 1.3-1.8x speedups for the backbone of CSWin and 1.2-1.5x speedups for the backbone of Swin to achieve similar accuracy. ELFATT is significantly faster than all comparison methods in achieving similar accuracy using different backbones.

\begin{table}[!t]
  \centering
  \caption{The comparison of semantic segmentation performance of all methods on ADE20K. Note: `mAcc' denotes mean class accuracy and `mIoU' denotes mean intersection over union. FLOPs are calculated using an input size of $512\times2048$. `160k' denotes the 160k fine-tuning iterations. Inference throughput is obtained using a batch size of 1 with mixed precision on a single NVIDIA H20 (96 GB) GPU.}
  \tiny
    \begin{tabular}{lrrrrr}
    \toprule
           \multicolumn{6}{c}{UperNet \cite{upernet} 160k} \\
    \midrule
    Method & \multicolumn{1}{r}{mAcc}  & \multicolumn{1}{r}{mIoU} & \multicolumn{1}{r}{FPS (nFA/FA)} & \multicolumn{1}{r}{\#} & \multicolumn{1}{r}{FLOPs (nFA/FA)} \\
    \midrule
    \multicolumn{6}{c}{CSWin-T}\\
    \midrule
    Agent                               & 60.8                    & 48.5               &  16/17 imgs/s   &    50M&   953.60G/929.64G      \\
    EFFATT                              & 60.6                    & 48.8               &  28/33 imgs/s   &    50M&  1008.73G/930.34G      \\
    \rowcolor{gray!25}
    ELFATT                              & 61.2                    & 49.6               &  28/32 imgs/s   &    50M&  1014.26G/929.53G      \\
    FLatten                             & 61.4                    & 49.3               &  25/27 imgs/s   &    51M&   954.15G/930.19G      \\
    GLOBAL                              & 61.1                    & 48.8               & \ 6/14 imgs/s   &    50M&  2458.75G/928.67G      \\
    LOCAL                               & 61.1                    & 49.6               &  26/28 imgs/s   &    50M&   963.38G/928.68G      \\    
    \midrule
    \multicolumn{6}{c}{Swin-T}\\
    \midrule
    Agent                                     & 58.5                    & 46.7               &   4/--- imgs/s   &    61M&   957.50G/\ ---\ \ \ \ \ \ \  \\
    EFFATT                                    & 58.3                    & 46.7               &  35/39 imgs/s    &    60M&   981.22G/939.35G             \\
    \rowcolor{gray!25}
    ELFATT                                    & 59.3                    & 47.7               &  34/38 imgs/s    &    62M&   991.27G/943.94G             \\
    FLatten                                   & 57.0                    & 44.8               &  35/--- imgs/s   &    61M&   944.62G/\ ---\ \ \ \ \ \ \  \\
    GLOBAL                                    & 59.3                    & 47.8               &   5/14 imgs/s    &    60M&  2873.79G/937.84G             \\
    LOCAL                                     & 55.6                    & 44.5               &  38/--- imgs/s   &    60M&   945.66G/\ ---\ \ \ \ \ \ \  \\
    \midrule
    \multicolumn{6}{c}{Others}\\    
    \midrule
    ConvNeXt-T                                       & 58.3                    & 46.1               &  37/--- imgs/s   &    60M&  939.69G/\ ---\ \ \ \ \ \ \  \\
    VMamba-T                                         & 59.3                    & 47.9               &  34/--- imgs/s   &    62M&  948.78G/\ ---\ \ \ \ \ \ \  \\
    \bottomrule
    \end{tabular}
  \label{tabADE}
\end{table}

\begin{table}[htbp]
\scriptsize
  \centering
  \caption{The comparison of test accuracy of different attention mechanisms on LRA.}
    \begin{tabular}{l>{\columncolor{gray!25}}rrrrrrr}
    \toprule
    Dataset (Seqlen)   & \multicolumn{1}{r}{ELFATT} & \multicolumn{1}{r}{Linformer}& \multicolumn{1}{r}{Nystr{\" o}m} & \multicolumn{1}{r}{Primal}  & \multicolumn{1}{r}{Reformer} & \multicolumn{1}{r}{VaniATT}\\
    \midrule
    ListOps (2K)                 & \textbf{38.1} & 37.3  & 37.2            & 37.3 & 19.1 & 37.1 \\
    Text (4K)                    & 64.8          & 55.9  & \textbf{65.5}   & 61.2 & 64.9 & 65.0 \\
    Retrieval (4K)               & \textbf{81.3} & 79.4  & 79.6            & 77.8 & 78.6 & 79.4 \\
    Image (1K)                   & \textbf{45.6} & 37.8  & 41.6            & 43.0 & 43.3 & 38.2 \\
    Pathfinder (1K)              & 72.3          & 67.6  & 70.9            & 68.3 & 69.4 & \textbf{74.2}\\
    \midrule
    Avg Acc. (\%)            & \textbf{60.4} & 55.6  & 59.0            & 57.5 & 55.1 & 58.8\\
    \bottomrule
    \end{tabular}
  \label{LRA}
\end{table}

\subsection{Semantic Segmentation Performance}
Table \ref{tabADE} shows the comparison of semantic segmentation performance of all methods on ADE20K. VaniATT and ELFATT achieve close mean class accuracy (mAcc) and mean intersection over union (mIoU), and significantly outperform other attention mechanisms when using Swin-T as the backbone. ELFATT and FLatten achieve close mAcc when using CSWin-T as the backbone. However, ELFATT and the local window-based attention mechanism \cite{CSWin} significantly outperform other attention mechanisms in terms of mIoU. Without using FlashAttention-2, ELFATT offers a nearly 5x speedup over VaniATT using CSWin-T as the backbone and a 7x speedup over VaniATT using Swin-T as the backbone. Even using FlashAttention-2, under the CSWin-T backbone, VaniATT is still 2x slower than ELFATT without using FlashAttention-2 and 2.3x slower than ELFATT using FlashAttention-2. Under the Swin-T backbone, VaniATT is still 2.4x slower than ELFATT without using FlashAttention-2 and 2.7x slower than ELFATT using FlashAttention-2. The speed of ELFATT is almost the same as the speed of EFFATT which is a real linear attention mechanism. However, ELFATT achieves significantly higher mAcc and mIoU than those of EFFATT. ELFATT outperforms ConvNeXt-T and VMamba-T in terms of segmentation performance.

\begin{table}[htbp]
    \tiny
  \centering
  \caption{The comparison of running time (s) per 1K training steps and peak memory consumption of different attention mechanisms on LRA on a single NVIDIA H20 (96GB).}
    \begin{tabular}{l>{\columncolor{gray!25}}rrrrrrr}
    \toprule
    Dataset (Seqlen)  & \multicolumn{1}{c}{ELFATT} & \multicolumn{1}{c}{FA-2} & \multicolumn{1}{c}{Linformer}& \multicolumn{1}{c}{Nystr{\" o}m} & \multicolumn{1}{c}{Primal} & \multicolumn{1}{c}{Reformer} & \multicolumn{1}{c}{VaniATT} \\
    \midrule
    \multicolumn{8}{c}{Time (s) $\vert$ Memory (GB)}\\
    \midrule
    ListOps (2K)     & 	\textbf{16.2} $\vert$ \textbf{0.4} & 24.2 $\vert$ \textbf{0.4} & 19.6 $\vert$ 1.1   & 28.3 $\vert$ 0.6  & 21.2 $\vert$ 0.5 & 34.9 $\vert$ 1.4  &  71.2 $\vert$ 4.3 \\
    Text (4K)        & 	\textbf{28.8} $\vert$ \textbf{0.8} & 63.7 $\vert$ \textbf{0.8} & 35.1 $\vert$ 2.2   & 50.4 $\vert$ 1.2  & 33.9 $\vert$ 0.9 & 69.3 $\vert$ 2.8  & 263.2 $\vert$ 16.5\\
    Retrieval (4K)   &  \textbf{55.2} $\vert$ 1.5          & 126.3 $\vert$ \textbf{1.4}& 69.1 $\vert$ 4.1   & 96.0 $\vert$ 2.4  & 66.5 $\vert$ 1.9 & 137.8 $\vert$ 5.4 & 523.0 $\vert$ 17.2  \\
    Image (1K)       & 	\textbf{42.9} $\vert$ \textbf{0.8} & 48.9 $\vert$ \textbf{0.8} & 45.8 $\vert$ 2.2   & 69.7 $\vert$ 1.4  & 49.8 $\vert$ 1.0 & 82.9 $\vert$ 2.9  & 109.2 $\vert$ 4.5 \\
    Pathfinder (1K)  & 	\textbf{49.5} $\vert$ \textbf{0.8} & 63.8 $\vert$ \textbf{0.8} & 63.4 $\vert$ 2.2   & 91.0 $\vert$ 1.4  & 64.6 $\vert$ 1.0 & 113.5 $\vert$ 2.9 & 144.2 $\vert$ 4.5   \\
    \bottomrule
    \end{tabular}
  \label{LRA1}
\end{table}

\begin{table}[!t]
  \centering
  \caption{The comparison of object detection performance of all methods on MS COCO 2017. Note: FLOPs are calculated using an input size of $1280\times800$. `1$\times$' denotes the fine-tuning training schedule with 12 epochs and `3$\times$MS' represents fine-tuning using the multiscale training schedule with 36 epochs. AP$^{\rm b}$ denotes box average precision and AP$^{\rm m}$ denotes mask average precision. Inference throughput (FPS (nFA/FA), imgs/s) is obtained using a batch size of 1 with mixed precision on a single NVIDIA H20 (96 GB) GPU.}
    \tiny
    \begin{tabular}{lrrrrrrrrr}
    \toprule
    \multicolumn{10}{c}{Mask-RCNN \cite{maskrcnn} 1$\times$ schedule} \\
    \midrule
    Method & \multicolumn{1}{l}{AP$^{\rm b}$} & \multicolumn{1}{l}{AP$^{\rm b}_{50}$} & \multicolumn{1}{l}{AP$^{\rm b}_{75}$} & \multicolumn{1}{l}{AP$^{\rm m}$} & \multicolumn{1}{l}{AP$^{\rm m}_{50}$} & \multicolumn{1}{l}{AP$^{\rm m}_{75}$} & \multicolumn{1}{r}{FPS} & \multicolumn{1}{r}{\#}    & \multicolumn{1}{r}{FLOPs (nFA/FA)} \\
    \midrule
    \multicolumn{10}{c}{CSWin-T}\\
    \midrule
    Agent                  & 46.8           & 68.9           & 51.3           & 42.3           & 65.9           & 45.3           & 19/20 & 40M      &  273.87G/254.51G       \\
    EFFATT                 & 46.1           & 68.3           & 50.5           & 41.9           & 65.5           & 45.3           & 32/36 & 40M      &  329.95G/255.20G       \\
    \rowcolor{gray!25}
    ELFATT                 & 47.0           & 69.2           & 51.4           & 42.6           & 66.4           & 45.9           & 30/33 & 40M      &  334.86G/254.40G       \\
    FLatten                & 46.6           & 68.8           & 51.0           & 42.2           & 65.7           & 45.3           & 26/28 & 41M      &  274.41G/255.05G       \\
    GLOBAL                 & 47.0           & 69.1           & 51.9           & 42.6           & 66.1           & 45.9           &  7/18 & 40M      & 1712.76G/253.56G       \\
    LOCAL                  & 46.5           & 68.5           & 51.0           & 42.1           & 65.6           & 45.3           & 26/28 & 40M      &  281.45G/253.57G       \\
    \midrule
    \multicolumn{10}{c}{Swin-T}\\
    \midrule
    Agent                        & 44.6           & 67.5           & 48.7           & 40.7           & 64.4           & 43.4           & 5/--- & 48M      & 278.42G/\ ---\ \ \ \ \ \ \  \\
    EFFATT                       & 44.7           & 67.0           & 48.9           & 41.1           & 64.0           & 44.4           & 40/46 & 48M      & 301.89G/261.95G             \\
    \rowcolor{gray!25}
    ELFATT                       & 46.1           & 68.3           & 50.8           & 42.1           & 65.4           & 45.3           & 39/45 & 50M      & 311.39G/266.43G             \\
    FLatten                      & 44.2           & 67.3           & 48.5           & 40.2           & 63.8           & 43.0           &41/--- & 49M      & 266.43G/\ ---\ \ \ \ \ \ \  \\
    GLOBAL                       & 45.4           & 67.9           & 49.7           & 41.6           & 65.0           & 44.8           & 6/17  & 48M      & 2106.75G/260.48G            \\
    LOCAL                        & 42.7           & 65.2           & 46.8           & 39.3           & 62.2           & 42.2           &45/--- & 48M      & 267.01G/\ ---\ \ \ \ \ \ \  \\
    \midrule
    \multicolumn{10}{c}{Other tiny models}\\
    \midrule
    ConvNeXt                          & 44.2           & 66.6           & 48.3           & 40.1           & 63.3           & 42.8           &44/--- & 48M      & 262.29G/\ ---\ \ \ \ \ \ \  \\
    VMamba                            & 47.4           & 69.5           & 52.0           & 42.7           & 66.3           & 46.0           &35/--- & 50M      & 271.16G/\ ---\ \ \ \ \ \ \  \\
        \midrule
    \multicolumn{10}{c}{Mask-RCNN \cite{maskrcnn} 3$\times$MS schedule} \\
    \midrule
Method & \multicolumn{1}{l}{AP$^{\rm b}$} & \multicolumn{1}{l}{AP$^{\rm b}_{50}$} & \multicolumn{1}{l}{AP$^{\rm b}_{75}$} & \multicolumn{1}{l}{AP$^{\rm m}$} & \multicolumn{1}{l}{AP$^{\rm m}_{50}$} & \multicolumn{1}{l}{AP$^{\rm m}_{75}$} & \multicolumn{1}{r}{FPS} & \multicolumn{1}{r}{\#}    & \multicolumn{1}{r}{FLOPs (nFA/FA)} \\
    \midrule
    \multicolumn{10}{c}{CSWin-T}\\
    \midrule
    Agent                  & 49.3           & 70.8           & 53.9           & 43.9           & 67.9           & 47.3           & 19/20 & 40M   &  273.87G/254.51G       \\
    EFFATT                 & 48.5           & 70.0           & 53.2           & 43.4           & 67.3           & 46.9           & 32/36 & 40M   &  329.95G/255.20G       \\
    \rowcolor{gray!25}
    ELFATT                 & 49.4           & 70.9           & 54.4           & 44.0           & 68.0           & 47.5           & 30/33 & 40M   &  334.86G/254.40G       \\
    FLatten                & 48.9           & 70.8           & 53.5           & 43.9           & 67.9           & 47.3           & 26/28 & 41M   &  274.41G/255.05G       \\
    GLOBAL                 & 48.8           & 70.0           & 53.5           & 43.6           & 67.4           & 47.1           &  7/18 & 40M   & 1712.76G/253.56G       \\
    LOCAL                  & 49.3           & 70.8           & 54.3           & 44.0           & 67.8           & 47.5           & 26/28 & 40M   &  281.45G/253.57G       \\
    \midrule
    \multicolumn{10}{c}{Swin-T}\\
    \midrule
    Agent                        & 47.3           & 69.5           & 51.9           & 42.7           & 66.4           & 46.2           & 5/--- & 48M      & 278.42G/\ ---\ \ \ \ \ \ \  \\
    EFFATT                       & 47.6           & 69.4           & 52.6           & 42.7           & 65.9           & 46.1           & 40/46 & 48M      & 301.89G/261.95G             \\
    \rowcolor{gray!25}
    ELFATT                       & 48.5           & 70.4           & 53.4           & 43.6           & 67.3           & 47.3           & 39/45 & 50M      & 311.39G/266.43G             \\
    FLatten                      & 46.5           & 68.5           & 50.8           & 42.1           & 65.4           & 45.1           &41/--- & 49M      & 266.43G/\ ---\ \ \ \ \ \ \  \\
    GLOBAL                       & 48.0           & 70.0           & 52.7           & 43.3           & 67.0           & 46.8           & 6/17  & 48M      & 2106.75G/260.48G            \\
    LOCAL                        & 46.0           & 68.1           & 50.3           & 41.6           & 65.1           & 44.9           &45/--- & 48M      & 267.01G/\ ---\ \ \ \ \ \ \  \\
    \midrule
    \multicolumn{10}{c}{Other tiny models}\\
    \midrule
    ConvNeXt                          & 46.2           & 67.9           & 50.8           & 41.7           & 65.0           & 44.9           &44/--- & 48M      & 262.29G/\ ---\ \ \ \ \ \ \  \\
    VMamba                            & 48.9           & 70.6           & 53.6           & 43.7           & 67.7           & 46.8           &35/--- & 50M      & 271.16G/\ ---\ \ \ \ \ \ \  \\
    \bottomrule
    \end{tabular}
  \label{tabCOCO}
\end{table}

\subsection{Performance Comparison on LRA}
The block sizes of ELFATT are set to 64 as used in Nystr{\" o}mformer (Nystr{\" o}m) in all LRA tasks to maintain fairness. As Table \ref{LRA} shows, ELFATT significantly outperforms all efficient attention mechanisms in all tasks except Text (4K) and Pathfinder (1K) on LRA. All efficient attention mechanisms are inferior to VaniATT in Pathfinder (1K). ELFATT is only inferior to VaniATT in Pathfinder (1K). For Text (4K), ELFATT is slightly inferior to VaniATT. The overall performance of ELFATT is the best and significantly better than all comparison methods, including VaniATT. Table \ref{LRA1} shows the running time (s) per 1K training steps and peak memory consumption of all methods in the training of LRA. ELFATT is significantly faster than all comparison methods. Especially, ELFATT offers 1.2-2.3x speedups over FlashAttention-2. For most tasks, ELFATT and FlashAttention-2 consume less GPU memory than other comparison methods.

\subsection{Object Detection Performance}
Table \ref{tabCOCO} shows the object detection performance of all methods on MS COCO 2017. Using the Mask-RCNN \cite{maskrcnn} 1$\times$ schedule, VaniATT and ELFATT significantly outperform other attention mechanisms. ELFATT exhibits slightly better performance than VaniATT when using CSWin-T as the backbone and significantly outperforms VaniATT when using Swin-T as the backbone. Under the backbone of CSWin-T, ELFATT offers a 4.3x speedup over VaniATT without using FlashAttention-2 and is still 1.8x faster than VaniATT using FlashAttention-2. When using Swin-T as the backbone, ELFATT offers a 6.5x speedup over VaniATT without using FlashAttention-2 and is still 2.6x faster than VaniATT using FlashAttention-2. Using the Mask-RCNN \cite{maskrcnn} 3$\times$ multiscale training schedule, ELFATT still achieves the best performance. VaniATT still outperforms most linear attention mechanisms when using Swin-T as the backbone. In the 1$\times$ schedule, ELFATT using CSWin-T as the backbone achieves close performance compared to VMamba-T and significantly outperforms ConvNeXt-T. In the 3$\times$ multiscale training schedule, ELFATT using CSWin-T as the backbone outperforms ConvNeXt-T and VMamba-T in terms of object detection performance.

\begin{table}[htbp]
  \centering
  \caption{The comparison of inference throughput of different attention mechanisms obtained on ImageNet-1K.}
  \scriptsize  
    \begin{tabular}{lrrrrr}
    \toprule
    Method & \multicolumn{1}{r}{Res.} & \multicolumn{1}{r}{FPS (FP32)}  & \multicolumn{1}{r}{Speedup} & \multicolumn{1}{r}{FPS (Mixed)} & \multicolumn{1}{r}{Speedup} \\
    \midrule
    \multicolumn{6}{c}{Backbone: CSWin-T; Platform: NVIDIA Jetson AGX Orin; Batch Size: 128; Mode: 60W}\\
    \midrule   
    Agent                               &$224^2$& 158 imgs/s                  & 1.5x               &  295 imgs/s   &       1.3x      \\
    EFFATT                              &$224^2$& 156 imgs/s                  & 1.4x               &  305 imgs/s   &       1.3x      \\
    \rowcolor{gray!25}
    ELFATT                              &$224^2$& 174 imgs/s                  & 1.6x               &  355 imgs/s   &       1.6x      \\
    FLatten                             &$224^2$& 138 imgs/s                  & 1.3x               &  229 imgs/s   &       1.0x      \\
    GLOBAL                              &$224^2$& 108 imgs/s                  & 1.0x               &  298 imgs/s   &       1.3x      \\
    LOCAL                               &$224^2$& 167 imgs/s                  & 1.5x               &  309 imgs/s   &       1.3x      \\
    \midrule
    \multicolumn{6}{c}{Backbone: Others; Platform: NVIDIA Jetson AGX Orin; Batch Size: 128; Mode: 60W}\\
    \midrule
    EfficientViT-B2                     &$288^2$& 198 imgs/s                  & 1.8x               &  282 imgs/s   &       1.2x      \\
    \midrule
    \multicolumn{6}{c}{Backbone: CSWin-T; Platform: NVIDIA Jetson Nano; Batch Size: 1; Mode: 10W}\\
    \midrule    
    Agent                               &$224^2$& 4.9 imgs/s      & 1.5x         & 4.5 imgs/s      & 1.3x \\
    EFFATT                              &$224^2$& 5.8 imgs/s      & 1.8x         & 5.5 imgs/s      & 1.6x \\
    \rowcolor{gray!25}
    ELFATT                              &$224^2$& 6.2 imgs/s      & 2.0x         & 5.8 imgs/s      & 1.7x \\
    FLatten                             &$224^2$& 4.7 imgs/s      & 1.5x         & 4.0 imgs/s      & 1.1x \\
    GLOBAL                              &$224^2$& 3.2 imgs/s      & 1.0x         & 3.5 imgs/s      & 1.0x \\
    LOCAL                               &$224^2$& 5.0 imgs/s      & 1.6x         & 4.2 imgs/s      & 1.2x \\
    \midrule
    \multicolumn{6}{c}{Backbone: Others; Platform: NVIDIA Jetson Nano; Batch Size: 1; Mode: 10W}\\
    \midrule
    EfficientViT-B2                     &$288^2$& 5.5 imgs/s      & 1.7x         & 5.7 imgs/s      & 1.6x \\
    \bottomrule
    \end{tabular}
  \label{tabEDG}
\end{table}

\begin{table}[htbp]
  \centering
  \caption{Inference speed on ImageNet-1K using 1 Jetson AGX Orin (60W) with a batch size of 128 in mixed precision using TensorRT to speed up CSWin-T-ELFATT.}
    \scriptsize
    \begin{tabular}{lrr}
    \toprule
    Running speed & \multicolumn{1}{r}{TensorRT} & \multicolumn{1}{r}{Pytorch} \\
    \midrule
    Running Time (s/batch)                               & \textbf{0.23} & 0.36     \\
    Inference Throughput (imgs/s)	                     &  \textbf{558} & 352      \\
    Total Inference Time (s)                             &   \textbf{90} & 142      \\
    \bottomrule
    \end{tabular}
  \label{tabEDG3}
\end{table}

\subsection{Speed Comparison on Edge GPUs}
Table \ref{tabEDG} shows the comparison of inference speed of different attention mechanisms obtained on ImageNet-1K using an NVIDIA Jetson AGX Orin/NVIDIA Jetson Nano GPU. In both FP32 and mixed precision, under the backbone of CSWin-T, ELFATT achieves the highest FPS. In mixed precision, ELFATT is significantly faster than all other attention mechanisms, as shown in Table \ref{tabEDG}. Even compared to edge-optimized EfficientViT-B2 \cite{efficientvit}, to achieve similar accuracy, ELFATT is significantly faster than EfficientViT-B2 in mixed precision on edge GPUs. Fig. \ref{figedgespeed} in Appendix shows the speed comparison of different attention mechanisms across various power modes from 5W to 60W by using NVIDIA Jetson Nano or NVIDIA Jetson AGX Orin. ELFATT consistently achieves 1.6x to 2.0x speedups compared to other attention mechanisms in various power modes from 5W to 60W. Compared to edge-optimized EfficientViT-B2, in mixed precision, ELFATT performs on par with EfficientViT-B2 in low-power conditions, such as the 5W mode. As power increases, the speed growth rate of ELFATT surpasses that of EfficientViT-B2, especially in high-power modes like 50W and 60W, where it demonstrates a significant speed advantage. Furthermore, in mixed precision, ELFATT outperforms all other attention mechanisms in terms of speed across all power modes tested, establishing itself as one of the fastest options in both high-performance and power-constrained scenarios. Table \ref{tabEDG3} shows that ELFATT can be further accelerated to achieve a 1.6x speedup by using TensorRT on a Jetson AGX Orin.

\begin{table}[htbp]
\scriptsize
  \centering
  \caption{The comparison of stable diffusion (SD) v1.5 (Diffusers v0.32.1) using different acceleration methods with different merging ratios in FP16 precision. Memory consumption for each method indicates how much GPU memory increases when the batch size is raised by 1. All methods are accelerated by FlashAttention-2. Note: The processing time of each method is obtained using a batch size of 48 on 1 NVIDIA Tesla A100 (40 GB) GPU.}
    \begin{tabular}{lrrrr}
    \toprule
    Method & \multicolumn{1}{r}{Ratio} & \multicolumn{1}{r}{FID $\downarrow$}  & \multicolumn{1}{r}{Time (s/img)} & \multicolumn{1}{r}{Memory (GB/img)}  \\
    \midrule
    Agent-SD          &0.1& 30.2       &  0.794       &   1.00     \\
    Agent-SD          &0.2& 30.2       &  0.777       &   1.00     \\
    Agent-SD          &0.3& 30.3       &  0.760       &   1.00     \\
    Agent-SD          &0.4& 30.6       &  0.746       &   0.70     \\
    Agent-SD          &0.5& 30.7       &  0.724       &   0.70     \\
    \midrule
    \rowcolor{gray!25}
    ELFATT-SD         &0.1& 30.2       &  0.772       &   0.80     \\
    \rowcolor{gray!25}
    ELFATT-SD         &0.2& 30.2       &  0.756       &   0.75     \\
    \rowcolor{gray!25}
    ELFATT-SD         &0.3& 30.4       &  0.739       &   0.75     \\
    \rowcolor{gray!25}
    ELFATT-SD         &0.4& 30.2       &  0.725       &   0.70     \\
    \rowcolor{gray!25}
    ELFATT-SD         &0.5& 30.4       &  0.707       &   0.70     \\
    \midrule
    SD (Baseline)     &---& 31.0       &  0.768       &   0.80     \\
    \midrule    
    ToMe-SD           &0.1& 30.7       &  0.788       &   0.80     \\
    ToMe-SD           &0.2& 30.5       &  0.765       &   0.80     \\
    ToMe-SD           &0.3& 30.6       &  0.748       &   0.80     \\
    ToMe-SD           &0.4& 30.6       &  0.732       &   0.80     \\
    ToMe-SD           &0.5& 30.9       &  0.714       &   0.70     \\
    \bottomrule
    \end{tabular}
  \label{tabSD}
\end{table}

\subsection{Stable Diffusion Acceleration Using ELFATT}
Table \ref{tabSD} shows the performance comparison of stable diffusion (SD) \cite{SD} v1.5 (Diffusers v0.32.1) by using different acceleration methods. Compared to Agent (Agent-SD) and ToMe (ToMe-SD) \cite{ToMe}, ELFATT (ELFATT-SD) provides lower FID \cite{FID} scores, consumes less GPU memory, and shows faster computation speed in the same merging ratio. From Table \ref{tabSD}, it can be seen that SD using FlashAttention-2 is able to achieve a comparable speed and memory usage compared to ToMe and Agent when the merging ratio is small. Only ELFATT is faster and has less memory usage than SD using FlashAttention-2 in most merging ratios. Fig. \ref{SDvis} in Appendix shows the visual comparison of images generated by different methods. ELFATT can help SD generate richer details and reduce unreality. Taking the prompt ``Komodo dragon'' as an example, ELFATT helps SD generate a ``Komodo dragon'' with a more realistic body proportion, which is more like a real monitor.

\section{Conclusion}
A novel efficient linear fast attention (ELFATT) mechanism is proposed for ViTs to achieve low memory I/O operations, linear computational complexity, and high performance at the same time. ELFATT offers 4-7x speedups over VaniATT in high-resolution vision tasks without losing performance. ELFATT is compatible with FlashAttention. Using FlashAttention-2, ELFATT still offers 2-3x speedups over VaniATT in high-resolution vision tasks without losing performance. ELFATT without using FlashAttention-2 is even faster than VaniATT using FlashAttention-2 on both low-resolution and high-resolution vision datasets, which shows great potential of ELFATT. Even in some non-vision tasks of LRA, ELFATT still achieves leading performance and offers 1.2-2.3x speedups over FlashAttention-2. What's more, on edge GPUs, ELFATT still offers 1.6x to 2.0x speedups compared to state-of-the-art attention mechanisms in various power modes from 5W to 60W. In addition, ELFATT can be used to enhance and accelerate diffusion tasks directly without training.

\section*{Acknowledgement}
This work is supported by the Hong Kong Innovation and Technology Commission (InnoHK Project CIMDA), the Institute of Digital Medicine, City University of Hong Kong (Project 9229503), and the Natural Science Special Project Research Fund of Guizhou University (No. 2025-06).

\bibliographystyle{ACM-Reference-Format}
\balance
\bibliography{references}


\begin{thebibliography}{66}


\ifx \showCODEN    \undefined \def \showCODEN     #1{\unskip}     \fi
\ifx \showISBNx    \undefined \def \showISBNx     #1{\unskip}     \fi
\ifx \showISBNxiii \undefined \def \showISBNxiii  #1{\unskip}     \fi
\ifx \showISSN     \undefined \def \showISSN      #1{\unskip}     \fi
\ifx \showLCCN     \undefined \def \showLCCN      #1{\unskip}     \fi
\ifx \shownote     \undefined \def \shownote      #1{#1}          \fi
\ifx \showarticletitle \undefined \def \showarticletitle #1{#1}   \fi
\ifx \showURL      \undefined \def \showURL       {\relax}        \fi
\providecommand\bibfield[2]{#2}
\providecommand\bibinfo[2]{#2}
\providecommand\natexlab[1]{#1}
\providecommand\showeprint[2][]{arXiv:#2}

\bibitem[Achiam et~al\mbox{.}(2023)]%
        {GPT}
\bibfield{author}{\bibinfo{person}{Josh Achiam}, \bibinfo{person}{Steven Adler}, \bibinfo{person}{Sandhini Agarwal}, \bibinfo{person}{Lama Ahmad}, \bibinfo{person}{Ilge Akkaya}, \bibinfo{person}{Florencia~Leoni Aleman}, \bibinfo{person}{Diogo Almeida}, \bibinfo{person}{Janko Altenschmidt}, \bibinfo{person}{Sam Altman}, \bibinfo{person}{Shyamal Anadkat}, {et~al\mbox{.}}} \bibinfo{year}{2023}\natexlab{}.
\newblock \showarticletitle{G{PT}-4 technical report}.
\newblock \bibinfo{journal}{\emph{arXiv preprint arXiv:2303.08774}} (\bibinfo{year}{2023}).
\newblock


\bibitem[Beltagy et~al\mbox{.}(2020)]%
        {SPARSE2}
\bibfield{author}{\bibinfo{person}{Iz Beltagy}, \bibinfo{person}{Matthew~E. Peters}, {and} \bibinfo{person}{Arman Cohan}.} \bibinfo{year}{2020}\natexlab{}.
\newblock \showarticletitle{Longformer: {T}he long-document transformer}.
\newblock \bibinfo{journal}{\emph{CoRR}}  \bibinfo{volume}{abs/2004.05150} (\bibinfo{year}{2020}).
\newblock
\showeprint[arXiv]{2004.05150}
\urldef\tempurl%
\url{https://arxiv.org/abs/2004.05150}
\showURL{%
\tempurl}


\bibitem[Bolya et~al\mbox{.}(2023)]%
        {HATT}
\bibfield{author}{\bibinfo{person}{Daniel Bolya}, \bibinfo{person}{Cheng-Yang Fu}, \bibinfo{person}{Xiaoliang Dai}, \bibinfo{person}{Peizhao Zhang}, {and} \bibinfo{person}{Judy Hoffman}.} \bibinfo{year}{2023}\natexlab{}.
\newblock \showarticletitle{Hydra {A}ttention: {E}fficient attention with many heads}. In \bibinfo{booktitle}{\emph{Computer Vision -- ECCV 2022 Workshops}}, \bibfield{editor}{\bibinfo{person}{Leonid Karlinsky}, \bibinfo{person}{Tomer Michaeli}, {and} \bibinfo{person}{Ko~Nishino}} (Eds.). \bibinfo{publisher}{Springer Nature Switzerland}, \bibinfo{address}{Cham}, \bibinfo{pages}{35--49}.
\newblock
\showISBNx{978-3-031-25082-8}


\bibitem[Bolya and Hoffman(2023)]%
        {ToMe}
\bibfield{author}{\bibinfo{person}{Daniel Bolya} {and} \bibinfo{person}{Judy Hoffman}.} \bibinfo{year}{2023}\natexlab{}.
\newblock \showarticletitle{Token merging for fast stable diffusion}. In \bibinfo{booktitle}{\emph{Proceedings of the IEEE/CVF Conference on Computer Vision and Pattern Recognition (CVPR) Workshops}}. \bibinfo{pages}{4599--4603}.
\newblock


\bibitem[Brooks et~al\mbox{.}(2024)]%
        {Sora}
\bibfield{author}{\bibinfo{person}{Tim Brooks}, \bibinfo{person}{Bill Peebles}, \bibinfo{person}{Connor Holmes}, \bibinfo{person}{Will DePue}, \bibinfo{person}{Yufei Guo}, \bibinfo{person}{Li Jing}, \bibinfo{person}{David Schnurr}, \bibinfo{person}{Joe Taylor}, \bibinfo{person}{Troy Luhman}, \bibinfo{person}{Eric Luhman}, \bibinfo{person}{Clarence Ng}, \bibinfo{person}{Ricky Wang}, {and} \bibinfo{person}{Aditya Ramesh}.} \bibinfo{year}{2024}\natexlab{}.
\newblock \showarticletitle{Video generation models as world simulators}.
\newblock \bibinfo{journal}{\emph{OpenAI}} (\bibinfo{year}{2024}).
\newblock
\urldef\tempurl%
\url{https://openai.com/index/video-generation-models-as-world-simulators}
\showURL{%
\tempurl}


\bibitem[Cai et~al\mbox{.}(2023)]%
        {efficientvit}
\bibfield{author}{\bibinfo{person}{Han Cai}, \bibinfo{person}{Junyan Li}, \bibinfo{person}{Muyan Hu}, \bibinfo{person}{Chuang Gan}, {and} \bibinfo{person}{Song Han}.} \bibinfo{year}{2023}\natexlab{}.
\newblock \showarticletitle{Efficient{V}i{T}: {L}ightweight multi-scale attention for high-resolution dense prediction}. In \bibinfo{booktitle}{\emph{Proceedings of the IEEE/CVF International Conference on Computer Vision (ICCV)}}. \bibinfo{pages}{17302--17313}.
\newblock


\bibitem[Chen et~al\mbox{.}(2023a)]%
        {FLOPsslow}
\bibfield{author}{\bibinfo{person}{Jierun Chen}, \bibinfo{person}{Shiu-hong Kao}, \bibinfo{person}{Hao He}, \bibinfo{person}{Weipeng Zhuo}, \bibinfo{person}{Song Wen}, \bibinfo{person}{Chul-Ho Lee}, {and} \bibinfo{person}{S.-H.~Gary Chan}.} \bibinfo{year}{2023}\natexlab{a}.
\newblock \showarticletitle{Run, {D}on't {W}alk: {C}hasing higher {FLOPS} for faster neural networks}. In \bibinfo{booktitle}{\emph{Proceedings of the IEEE/CVF Conference on Computer Vision and Pattern Recognition (CVPR)}}. \bibinfo{pages}{12021--12031}.
\newblock


\bibitem[Chen et~al\mbox{.}(2023b)]%
        {SVDATT}
\bibfield{author}{\bibinfo{person}{Yingyi Chen}, \bibinfo{person}{Qinghua Tao}, \bibinfo{person}{Francesco Tonin}, {and} \bibinfo{person}{Johan Suykens}.} \bibinfo{year}{2023}\natexlab{b}.
\newblock \showarticletitle{Primal-{A}ttention: {S}elf-attention through asymmetric kernel {SVD} in primal representation}. In \bibinfo{booktitle}{\emph{Advances in Neural Information Processing Systems}}, \bibfield{editor}{\bibinfo{person}{A.~Oh}, \bibinfo{person}{T.~Naumann}, \bibinfo{person}{A.~Globerson}, \bibinfo{person}{K.~Saenko}, \bibinfo{person}{M.~Hardt}, {and} \bibinfo{person}{S.~Levine}} (Eds.), Vol.~\bibinfo{volume}{36}. \bibinfo{publisher}{Curran Associates, Inc.}, \bibinfo{pages}{65088--65101}.
\newblock


\bibitem[Child et~al\mbox{.}(2019)]%
        {SPARSE3}
\bibfield{author}{\bibinfo{person}{Rewon Child}, \bibinfo{person}{Scott Gray}, \bibinfo{person}{Alec Radford}, {and} \bibinfo{person}{Ilya Sutskever}.} \bibinfo{year}{2019}\natexlab{}.
\newblock \showarticletitle{Generating long sequences with sparse transformers}.
\newblock \bibinfo{journal}{\emph{CoRR}}  \bibinfo{volume}{abs/1904.10509} (\bibinfo{year}{2019}).
\newblock
\showeprint[arXiv]{1904.10509}
\urldef\tempurl%
\url{http://arxiv.org/abs/1904.10509}
\showURL{%
\tempurl}


\bibitem[Choromanski et~al\mbox{.}(2021)]%
        {PEF}
\bibfield{author}{\bibinfo{person}{Krzysztof~Marcin Choromanski}, \bibinfo{person}{Valerii Likhosherstov}, \bibinfo{person}{David Dohan}, \bibinfo{person}{Xingyou Song}, \bibinfo{person}{Andreea Gane}, \bibinfo{person}{Tamas Sarlos}, \bibinfo{person}{Peter Hawkins}, \bibinfo{person}{Jared~Quincy Davis}, \bibinfo{person}{Afroz Mohiuddin}, \bibinfo{person}{Lukasz Kaiser}, \bibinfo{person}{David~Benjamin Belanger}, \bibinfo{person}{Lucy~J Colwell}, {and} \bibinfo{person}{Adrian Weller}.} \bibinfo{year}{2021}\natexlab{}.
\newblock \showarticletitle{Rethinking attention with performers}. In \bibinfo{booktitle}{\emph{International Conference on Learning Representations}}.
\newblock
\urldef\tempurl%
\url{https://openreview.net/forum?id=Ua6zuk0WRH}
\showURL{%
\tempurl}


\bibitem[Chu et~al\mbox{.}(2023)]%
        {CPE}
\bibfield{author}{\bibinfo{person}{Xiangxiang Chu}, \bibinfo{person}{Zhi Tian}, \bibinfo{person}{Bo Zhang}, \bibinfo{person}{Xinlong Wang}, {and} \bibinfo{person}{Chunhua Shen}.} \bibinfo{year}{2023}\natexlab{}.
\newblock \showarticletitle{Conditional positional encodings for vision transformers}. In \bibinfo{booktitle}{\emph{International Conference on Learning Representations}}.
\newblock
\urldef\tempurl%
\url{https://openreview.net/forum?id=3KWnuT-R1bh}
\showURL{%
\tempurl}


\bibitem[Clevert(2015)]%
        {elus}
\bibfield{author}{\bibinfo{person}{Djork-Arn{\'e} Clevert}.} \bibinfo{year}{2015}\natexlab{}.
\newblock \showarticletitle{Fast and accurate deep network learning by exponential linear units ({ELU}s)}.
\newblock \bibinfo{journal}{\emph{arXiv preprint arXiv:1511.07289}} (\bibinfo{year}{2015}).
\newblock


\bibitem[Dao(2024)]%
        {Flash2}
\bibfield{author}{\bibinfo{person}{Tri Dao}.} \bibinfo{year}{2024}\natexlab{}.
\newblock \showarticletitle{{F}lash{A}ttention-2: {F}aster attention with better parallelism and work partitioning}. In \bibinfo{booktitle}{\emph{International Conference on Learning Representations}}.
\newblock
\urldef\tempurl%
\url{https://openreview.net/forum?id=mZn2Xyh9Ec}
\showURL{%
\tempurl}


\bibitem[Dao et~al\mbox{.}(2022)]%
        {Flash1}
\bibfield{author}{\bibinfo{person}{Tri Dao}, \bibinfo{person}{Dan Fu}, \bibinfo{person}{Stefano Ermon}, \bibinfo{person}{Atri Rudra}, {and} \bibinfo{person}{Christopher R\'{e}}.} \bibinfo{year}{2022}\natexlab{}.
\newblock \showarticletitle{Flash{A}ttention: {F}ast and memory-efficient exact attention with IO-awareness}. In \bibinfo{booktitle}{\emph{Advances in Neural Information Processing Systems}}, \bibfield{editor}{\bibinfo{person}{S.~Koyejo}, \bibinfo{person}{S.~Mohamed}, \bibinfo{person}{A.~Agarwal}, \bibinfo{person}{D.~Belgrave}, \bibinfo{person}{K.~Cho}, {and} \bibinfo{person}{A.~Oh}} (Eds.), Vol.~\bibinfo{volume}{35}. \bibinfo{publisher}{Curran Associates, Inc.}, \bibinfo{pages}{16344--16359}.
\newblock


\bibitem[Dao and Gu(2024)]%
        {mamba2}
\bibfield{author}{\bibinfo{person}{Tri Dao} {and} \bibinfo{person}{Albert Gu}.} \bibinfo{year}{2024}\natexlab{}.
\newblock \showarticletitle{Transformers are {SSM}s: {G}eneralized models and efficient algorithms through structured state space duality}. In \bibinfo{booktitle}{\emph{Proceedings of the 41st International Conference on Machine Learning}} \emph{(\bibinfo{series}{Proceedings of Machine Learning Research}, Vol.~\bibinfo{volume}{235})}, \bibfield{editor}{\bibinfo{person}{Ruslan Salakhutdinov}, \bibinfo{person}{Zico Kolter}, \bibinfo{person}{Katherine Heller}, \bibinfo{person}{Adrian Weller}, \bibinfo{person}{Nuria Oliver}, \bibinfo{person}{Jonathan Scarlett}, {and} \bibinfo{person}{Felix Berkenkamp}} (Eds.). \bibinfo{publisher}{PMLR}, \bibinfo{pages}{10041--10071}.
\newblock


\bibitem[Dhariwal and Nichol(2021)]%
        {CFG}
\bibfield{author}{\bibinfo{person}{Prafulla Dhariwal} {and} \bibinfo{person}{Alexander Nichol}.} \bibinfo{year}{2021}\natexlab{}.
\newblock \showarticletitle{Diffusion models beat {GAN}s on image synthesis}. In \bibinfo{booktitle}{\emph{Advances in Neural Information Processing Systems}}, \bibfield{editor}{\bibinfo{person}{M.~Ranzato}, \bibinfo{person}{A.~Beygelzimer}, \bibinfo{person}{Y.~Dauphin}, \bibinfo{person}{P.S. Liang}, {and} \bibinfo{person}{J.~Wortman Vaughan}} (Eds.), Vol.~\bibinfo{volume}{34}. \bibinfo{publisher}{Curran Associates, Inc.}, \bibinfo{pages}{8780--8794}.
\newblock


\bibitem[Dong et~al\mbox{.}(2022)]%
        {CSWin}
\bibfield{author}{\bibinfo{person}{Xiaoyi Dong}, \bibinfo{person}{Jianmin Bao}, \bibinfo{person}{Dongdong Chen}, \bibinfo{person}{Weiming Zhang}, \bibinfo{person}{Nenghai Yu}, \bibinfo{person}{Lu Yuan}, \bibinfo{person}{Dong Chen}, {and} \bibinfo{person}{Baining Guo}.} \bibinfo{year}{2022}\natexlab{}.
\newblock \showarticletitle{{CSW}in {T}ransformer: {A} general vision transformer backbone with cross-shaped windows}. In \bibinfo{booktitle}{\emph{Proceedings of the IEEE/CVF Conference on Computer Vision and Pattern Recognition (CVPR)}}. \bibinfo{pages}{12124--12134}.
\newblock


\bibitem[Golub and Van~Loan(2013)]%
        {golub2013matrix}
\bibfield{author}{\bibinfo{person}{Gene~H. Golub} {and} \bibinfo{person}{Charles~F. Van~Loan}.} \bibinfo{year}{2013}\natexlab{}.
\newblock \bibinfo{booktitle}{\emph{Matrix {C}omputations} (\bibinfo{edition}{fourth} ed.)}.
\newblock \bibinfo{publisher}{Johns Hopkins University Press, Baltimore, MD}.
\newblock
\showISBNx{978-1-4214-0794-4; 1-4214-0794-9; 978-1-4214-0859-0}


\bibitem[Grattafiori et~al\mbox{.}(2024)]%
        {Llama}
\bibfield{author}{\bibinfo{person}{Aaron Grattafiori}, \bibinfo{person}{Abhimanyu Dubey}, \bibinfo{person}{Abhinav Jauhri}, \bibinfo{person}{Abhinav Pandey}, \bibinfo{person}{Abhishek Kadian}, \bibinfo{person}{Ahmad Al-Dahle}, \bibinfo{person}{Aiesha Letman}, \bibinfo{person}{Akhil Mathur}, \bibinfo{person}{Alan Schelten}, \bibinfo{person}{Alex Vaughan}, {et~al\mbox{.}}} \bibinfo{year}{2024}\natexlab{}.
\newblock \showarticletitle{The llama 3 herd of models}.
\newblock \bibinfo{journal}{\emph{arXiv preprint arXiv:2407.21783}} (\bibinfo{year}{2024}).
\newblock


\bibitem[Gu and Dao(2024)]%
        {gu2024mamba}
\bibfield{author}{\bibinfo{person}{Albert Gu} {and} \bibinfo{person}{Tri Dao}.} \bibinfo{year}{2024}\natexlab{}.
\newblock \showarticletitle{Mamba: {L}inear-time sequence modeling with selective state spaces}. In \bibinfo{booktitle}{\emph{First Conference on Language Modeling}}.
\newblock
\urldef\tempurl%
\url{https://openreview.net/forum?id=tEYskw1VY2}
\showURL{%
\tempurl}


\bibitem[Han et~al\mbox{.}(2023)]%
        {flatten}
\bibfield{author}{\bibinfo{person}{Dongchen Han}, \bibinfo{person}{Xuran Pan}, \bibinfo{person}{Yizeng Han}, \bibinfo{person}{Shiji Song}, {and} \bibinfo{person}{Gao Huang}.} \bibinfo{year}{2023}\natexlab{}.
\newblock \showarticletitle{{FL}atten {T}ransformer: {V}ision transformer using focused linear attention}. In \bibinfo{booktitle}{\emph{Proceedings of the IEEE/CVF International Conference on Computer Vision (ICCV)}}. \bibinfo{pages}{5961--5971}.
\newblock


\bibitem[Han et~al\mbox{.}(2025)]%
        {AGENTATT}
\bibfield{author}{\bibinfo{person}{Dongchen Han}, \bibinfo{person}{Tianzhu Ye}, \bibinfo{person}{Yizeng Han}, \bibinfo{person}{Zhuofan Xia}, \bibinfo{person}{Siyuan Pan}, \bibinfo{person}{Pengfei Wan}, \bibinfo{person}{Shiji Song}, {and} \bibinfo{person}{Gao Huang}.} \bibinfo{year}{2025}\natexlab{}.
\newblock \showarticletitle{Agent {A}ttention: {O}n the integration of softmax and linear attention}. In \bibinfo{booktitle}{\emph{Computer Vision -- ECCV 2024}}, \bibfield{editor}{\bibinfo{person}{Ale{\v{s}} Leonardis}, \bibinfo{person}{Elisa Ricci}, \bibinfo{person}{Stefan Roth}, \bibinfo{person}{Olga Russakovsky}, \bibinfo{person}{Torsten Sattler}, {and} \bibinfo{person}{G{\"u}l Varol}} (Eds.). \bibinfo{publisher}{Springer Nature Switzerland}, \bibinfo{address}{Cham}, \bibinfo{pages}{124--140}.
\newblock


\bibitem[Han et~al\mbox{.}(2022)]%
        {MPIPD}
\bibfield{author}{\bibinfo{person}{Jiayi Han}, \bibinfo{person}{Longbin Zeng}, \bibinfo{person}{Liang Du}, \bibinfo{person}{Xiaoqing Ye}, \bibinfo{person}{Weiyang Ding}, {and} \bibinfo{person}{Jianfeng Feng}.} \bibinfo{year}{2022}\natexlab{}.
\newblock \showarticletitle{Modify self-attention via skeleton decomposition for effective point cloud transformer}.
\newblock \bibinfo{journal}{\emph{Proceedings of the AAAI Conference on Artificial Intelligence}} \bibinfo{volume}{36}, \bibinfo{number}{1} (\bibinfo{year}{2022}), \bibinfo{pages}{808--816}.
\newblock
\href{https://doi.org/10.1609/aaai.v36i1.19962}{doi:\nolinkurl{10.1609/aaai.v36i1.19962}}


\bibitem[He et~al\mbox{.}(2017)]%
        {maskrcnn}
\bibfield{author}{\bibinfo{person}{Kaiming He}, \bibinfo{person}{Georgia Gkioxari}, \bibinfo{person}{Piotr Dollar}, {and} \bibinfo{person}{Ross Girshick}.} \bibinfo{year}{2017}\natexlab{}.
\newblock \showarticletitle{Mask {R}-{CNN}}. In \bibinfo{booktitle}{\emph{2017 IEEE International Conference on Computer Vision (ICCV)}}. \bibinfo{publisher}{IEEE Computer Society}, \bibinfo{address}{Los Alamitos, CA, USA}, \bibinfo{pages}{2980--2988}.
\newblock
\href{https://doi.org/10.1109/ICCV.2017.322}{doi:\nolinkurl{10.1109/ICCV.2017.322}}


\bibitem[Heusel et~al\mbox{.}(2017)]%
        {FID}
\bibfield{author}{\bibinfo{person}{Martin Heusel}, \bibinfo{person}{Hubert Ramsauer}, \bibinfo{person}{Thomas Unterthiner}, \bibinfo{person}{Bernhard Nessler}, {and} \bibinfo{person}{Sepp Hochreiter}.} \bibinfo{year}{2017}\natexlab{}.
\newblock \showarticletitle{{GAN}s trained by a two time-scale update rule converge to a local {N}ash equilibrium}. In \bibinfo{booktitle}{\emph{Advances in Neural Information Processing Systems}}, \bibfield{editor}{\bibinfo{person}{I.~Guyon}, \bibinfo{person}{U.~Von Luxburg}, \bibinfo{person}{S.~Bengio}, \bibinfo{person}{H.~Wallach}, \bibinfo{person}{R.~Fergus}, \bibinfo{person}{S.~Vishwanathan}, {and} \bibinfo{person}{R.~Garnett}} (Eds.), Vol.~\bibinfo{volume}{30}. \bibinfo{publisher}{Curran Associates, Inc.}
\newblock


\bibitem[Horn and Johnson(2012)]%
        {hadamard}
\bibfield{author}{\bibinfo{person}{Roger~A. Horn} {and} \bibinfo{person}{Charles~R. Johnson}.} \bibinfo{year}{2012}\natexlab{}.
\newblock \bibinfo{booktitle}{\emph{Matrix Analysis} (\bibinfo{edition}{second} ed.)}.
\newblock \bibinfo{publisher}{Cambridge University Press}.
\newblock


\bibitem[Kang et~al\mbox{.}(2024)]%
        {Hilinear}
\bibfield{author}{\bibinfo{person}{Hankyul Kang}, \bibinfo{person}{Ming-Hsuan Yang}, {and} \bibinfo{person}{Jongbin Ryu}.} \bibinfo{year}{2024}\natexlab{}.
\newblock \showarticletitle{Interactive multi-head self-attention with linear complexity}.
\newblock \bibinfo{journal}{\emph{arXiv preprint arXiv:2402.17507}} (\bibinfo{year}{2024}).
\newblock


\bibitem[Katharopoulos et~al\mbox{.}(2020)]%
        {TRNN}
\bibfield{author}{\bibinfo{person}{Angelos Katharopoulos}, \bibinfo{person}{Apoorv Vyas}, \bibinfo{person}{Nikolaos Pappas}, {and} \bibinfo{person}{Fran{\c{c}}ois Fleuret}.} \bibinfo{year}{2020}\natexlab{}.
\newblock \showarticletitle{Transformers are {RNN}s: {F}ast autoregressive transformers with linear attention}. In \bibinfo{booktitle}{\emph{Proceedings of the 37th International Conference on Machine Learning}} \emph{(\bibinfo{series}{Proceedings of Machine Learning Research}, Vol.~\bibinfo{volume}{119})}, \bibfield{editor}{\bibinfo{person}{Hal~Daumé III} {and} \bibinfo{person}{Aarti Singh}} (Eds.). \bibinfo{publisher}{PMLR}, \bibinfo{pages}{5156--5165}.
\newblock


\bibitem[Kirillov et~al\mbox{.}(2023)]%
        {SAM}
\bibfield{author}{\bibinfo{person}{Alexander Kirillov}, \bibinfo{person}{Eric Mintun}, \bibinfo{person}{Nikhila Ravi}, \bibinfo{person}{Hanzi Mao}, \bibinfo{person}{Chloe Rolland}, \bibinfo{person}{Laura Gustafson}, \bibinfo{person}{Tete Xiao}, \bibinfo{person}{Spencer Whitehead}, \bibinfo{person}{Alexander~C. Berg}, \bibinfo{person}{Wan-Yen Lo}, \bibinfo{person}{Piotr Dollar}, {and} \bibinfo{person}{Ross Girshick}.} \bibinfo{year}{2023}\natexlab{}.
\newblock \showarticletitle{Segment anything}. In \bibinfo{booktitle}{\emph{Proceedings of the IEEE/CVF International Conference on Computer Vision (ICCV)}}. \bibinfo{pages}{4015--4026}.
\newblock


\bibitem[Kitaev et~al\mbox{.}(2020)]%
        {SPARSE6}
\bibfield{author}{\bibinfo{person}{Nikita Kitaev}, \bibinfo{person}{Lukasz Kaiser}, {and} \bibinfo{person}{Anselm Levskaya}.} \bibinfo{year}{2020}\natexlab{}.
\newblock \showarticletitle{{R}eformer: {T}he efficient transformer}. In \bibinfo{booktitle}{\emph{International Conference on Learning Representations}}.
\newblock
\urldef\tempurl%
\url{https://openreview.net/forum?id=rkgNKkHtvB}
\showURL{%
\tempurl}


\bibitem[Lin et~al\mbox{.}(2015)]%
        {COCO}
\bibfield{author}{\bibinfo{person}{Tsung-Yi Lin}, \bibinfo{person}{Michael Maire}, \bibinfo{person}{Serge Belongie}, \bibinfo{person}{Lubomir Bourdev}, \bibinfo{person}{Ross Girshick}, \bibinfo{person}{James Hays}, \bibinfo{person}{Pietro Perona}, \bibinfo{person}{Deva Ramanan}, \bibinfo{person}{C.~Lawrence Zitnick}, {and} \bibinfo{person}{Piotr Dollár}.} \bibinfo{year}{2015}\natexlab{}.
\newblock \showarticletitle{Microsoft {COCO}: {C}ommon objects in context}.
\newblock \bibinfo{journal}{\emph{arXiv preprint arXiv:1405.0312}} (\bibinfo{year}{2015}).
\newblock


\bibitem[Liu et~al\mbox{.}(2022b)]%
        {INSTEPS}
\bibfield{author}{\bibinfo{person}{Luping Liu}, \bibinfo{person}{Yi Ren}, \bibinfo{person}{Zhijie Lin}, {and} \bibinfo{person}{Zhou Zhao}.} \bibinfo{year}{2022}\natexlab{b}.
\newblock \showarticletitle{Pseudo numerical methods for diffusion models on manifolds}. In \bibinfo{booktitle}{\emph{International Conference on Learning Representations}}.
\newblock
\urldef\tempurl%
\url{https://openreview.net/forum?id=PlKWVd2yBkY}
\showURL{%
\tempurl}


\bibitem[Liu et~al\mbox{.}(2023)]%
        {Efficientvit2}
\bibfield{author}{\bibinfo{person}{Xinyu Liu}, \bibinfo{person}{Houwen Peng}, \bibinfo{person}{Ningxin Zheng}, \bibinfo{person}{Yuqing Yang}, \bibinfo{person}{Han Hu}, {and} \bibinfo{person}{Yixuan Yuan}.} \bibinfo{year}{2023}\natexlab{}.
\newblock \showarticletitle{{Efficient{V}i{T}: {M}emory efficient vision transformer with cascaded group attention}}. In \bibinfo{booktitle}{\emph{2023 IEEE/CVF Conference on Computer Vision and Pattern Recognition (CVPR)}}. \bibinfo{pages}{14420--14430}.
\newblock


\bibitem[Liu et~al\mbox{.}(2024)]%
        {Vmamba}
\bibfield{author}{\bibinfo{person}{Yue Liu}, \bibinfo{person}{Yunjie Tian}, \bibinfo{person}{Yuzhong Zhao}, \bibinfo{person}{Hongtian Yu}, \bibinfo{person}{Lingxi Xie}, \bibinfo{person}{Yaowei Wang}, \bibinfo{person}{Qixiang Ye}, \bibinfo{person}{Jianbin Jiao}, {and} \bibinfo{person}{Yunfan Liu}.} \bibinfo{year}{2024}\natexlab{}.
\newblock \showarticletitle{{VM}amba: {V}isual state space model}. In \bibinfo{booktitle}{\emph{Advances in Neural Information Processing Systems}}, \bibfield{editor}{\bibinfo{person}{A.~Globerson}, \bibinfo{person}{L.~Mackey}, \bibinfo{person}{D.~Belgrave}, \bibinfo{person}{A.~Fan}, \bibinfo{person}{U.~Paquet}, \bibinfo{person}{J.~Tomczak}, {and} \bibinfo{person}{C.~Zhang}} (Eds.), Vol.~\bibinfo{volume}{37}. \bibinfo{publisher}{Curran Associates, Inc.}, \bibinfo{pages}{103031--103063}.
\newblock


\bibitem[Liu et~al\mbox{.}(2021)]%
        {Swin}
\bibfield{author}{\bibinfo{person}{Ze Liu}, \bibinfo{person}{Yutong Lin}, \bibinfo{person}{Yue Cao}, \bibinfo{person}{Han Hu}, \bibinfo{person}{Yixuan Wei}, \bibinfo{person}{Zheng Zhang}, \bibinfo{person}{Stephen Lin}, {and} \bibinfo{person}{Baining Guo}.} \bibinfo{year}{2021}\natexlab{}.
\newblock \showarticletitle{Swin {T}ransformer: {H}ierarchical vision transformer using shifted windows}. In \bibinfo{booktitle}{\emph{Proceedings of the IEEE/CVF International Conference on Computer Vision (ICCV)}}. \bibinfo{pages}{10012--10022}.
\newblock


\bibitem[Liu et~al\mbox{.}(2022a)]%
        {ConvNeXt}
\bibfield{author}{\bibinfo{person}{Zhuang Liu}, \bibinfo{person}{Hanzi Mao}, \bibinfo{person}{Chao-Yuan Wu}, \bibinfo{person}{Christoph Feichtenhofer}, \bibinfo{person}{Trevor Darrell}, {and} \bibinfo{person}{Saining Xie}.} \bibinfo{year}{2022}\natexlab{a}.
\newblock \showarticletitle{A {C}onv{N}et for the 2020s}. In \bibinfo{booktitle}{\emph{Proceedings of the IEEE/CVF Conference on Computer Vision and Pattern Recognition (CVPR)}}. \bibinfo{pages}{11976--11986}.
\newblock


\bibitem[Lu et~al\mbox{.}(2021)]%
        {SOFT}
\bibfield{author}{\bibinfo{person}{Jiachen Lu}, \bibinfo{person}{Jinghan Yao}, \bibinfo{person}{Junge Zhang}, \bibinfo{person}{Xiatian Zhu}, \bibinfo{person}{Hang Xu}, \bibinfo{person}{Weiguo Gao}, \bibinfo{person}{Chunjing Xu}, \bibinfo{person}{Tao Xiang}, {and} \bibinfo{person}{Li Zhang}.} \bibinfo{year}{2021}\natexlab{}.
\newblock \showarticletitle{{SOFT}: {S}oftmax-free transformer with linear complexity}. In \bibinfo{booktitle}{\emph{Advances in Neural Information Processing Systems}}, \bibfield{editor}{\bibinfo{person}{M.~Ranzato}, \bibinfo{person}{A.~Beygelzimer}, \bibinfo{person}{Y.~Dauphin}, \bibinfo{person}{P.S. Liang}, {and} \bibinfo{person}{J.~Wortman Vaughan}} (Eds.), Vol.~\bibinfo{volume}{34}. \bibinfo{publisher}{Curran Associates, Inc.}, \bibinfo{pages}{21297--21309}.
\newblock


\bibitem[Nair and Hinton(2010)]%
        {ReLU}
\bibfield{author}{\bibinfo{person}{Vinod Nair} {and} \bibinfo{person}{Geoffrey~E. Hinton}.} \bibinfo{year}{2010}\natexlab{}.
\newblock \showarticletitle{Rectified linear units improve restricted boltzmann machines}. In \bibinfo{booktitle}{\emph{Proceedings of the 27th International Conference on Machine Learning}} (Haifa, Israel) \emph{(\bibinfo{series}{ICML'10})}. \bibinfo{publisher}{Omnipress}, \bibinfo{address}{Madison, WI, USA}, \bibinfo{pages}{807–814}.
\newblock
\showISBNx{9781605589077}


\bibitem[Peng et~al\mbox{.}(2021)]%
        {RFA}
\bibfield{author}{\bibinfo{person}{Hao Peng}, \bibinfo{person}{Nikolaos Pappas}, \bibinfo{person}{Dani Yogatama}, \bibinfo{person}{Roy Schwartz}, \bibinfo{person}{Noah Smith}, {and} \bibinfo{person}{Lingpeng Kong}.} \bibinfo{year}{2021}\natexlab{}.
\newblock \showarticletitle{Random feature attention}. In \bibinfo{booktitle}{\emph{International Conference on Learning Representations}}.
\newblock
\urldef\tempurl%
\url{https://openreview.net/forum?id=QtTKTdVrFBB}
\showURL{%
\tempurl}


\bibitem[Qin et~al\mbox{.}(2022)]%
        {COS}
\bibfield{author}{\bibinfo{person}{Zhen Qin}, \bibinfo{person}{Weixuan Sun}, \bibinfo{person}{Hui Deng}, \bibinfo{person}{Dongxu Li}, \bibinfo{person}{Yunshen Wei}, \bibinfo{person}{Baohong Lv}, \bibinfo{person}{Junjie Yan}, \bibinfo{person}{Lingpeng Kong}, {and} \bibinfo{person}{Yiran Zhong}.} \bibinfo{year}{2022}\natexlab{}.
\newblock \showarticletitle{cos{F}ormer: {R}ethinking softmax in attention}. In \bibinfo{booktitle}{\emph{International Conference on Learning Representations}}.
\newblock
\urldef\tempurl%
\url{https://openreview.net/forum?id=Bl8CQrx2Up4}
\showURL{%
\tempurl}


\bibitem[Ramapuram et~al\mbox{.}(2025)]%
        {FlashSigmoid}
\bibfield{author}{\bibinfo{person}{Jason Ramapuram}, \bibinfo{person}{Federico Danieli}, \bibinfo{person}{Eeshan~Gunesh Dhekane}, \bibinfo{person}{Floris Weers}, \bibinfo{person}{Dan Busbridge}, \bibinfo{person}{Pierre Ablin}, \bibinfo{person}{Tatiana Likhomanenko}, \bibinfo{person}{Jagrit Digani}, \bibinfo{person}{Zijin Gu}, \bibinfo{person}{Amitis Shidani}, {and} \bibinfo{person}{Russell Webb}.} \bibinfo{year}{2025}\natexlab{}.
\newblock \showarticletitle{Theory, analysis, and best practices for sigmoid self-attention}. In \bibinfo{booktitle}{\emph{International Conference on Learning Representations}}.
\newblock
\urldef\tempurl%
\url{https://openreview.net/forum?id=Zhdhg6n2OG}
\showURL{%
\tempurl}


\bibitem[Ran et~al\mbox{.}(2025)]%
        {fastvit2}
\bibfield{author}{\bibinfo{person}{Zhuoheng Ran}, \bibinfo{person}{Zewen Ye}, \bibinfo{person}{Chong Wu}, \bibinfo{person}{{Ray C.C.} Cheung}, {and} \bibinfo{person}{Hong Yan}.} \bibinfo{year}{2025}\natexlab{}.
\newblock \showarticletitle{Fast{V}i{T}: {R}eal-time linear attention accelerator for dense predictions of vision transformer ({V}i{T})}. In \bibinfo{booktitle}{\emph{2025 IEEE International Symposium on Circuits and Systems (ISCAS)}}. \bibinfo{publisher}{IEEE}, \bibinfo{pages}{1--5}.
\newblock
\href{https://doi.org/10.1109/ISCAS56072.2025.11043624}{doi:\nolinkurl{10.1109/ISCAS56072.2025.11043624}}


\bibitem[Rombach et~al\mbox{.}(2022)]%
        {SD}
\bibfield{author}{\bibinfo{person}{Robin Rombach}, \bibinfo{person}{Andreas Blattmann}, \bibinfo{person}{Dominik Lorenz}, \bibinfo{person}{Patrick Esser}, {and} \bibinfo{person}{Bj\"orn Ommer}.} \bibinfo{year}{2022}\natexlab{}.
\newblock \showarticletitle{High-resolution image synthesis with latent diffusion models}. In \bibinfo{booktitle}{\emph{Proceedings of the IEEE/CVF Conference on Computer Vision and Pattern Recognition (CVPR)}}. \bibinfo{pages}{10684--10695}.
\newblock


\bibitem[Russakovsky et~al\mbox{.}(2015)]%
        {IMGNET}
\bibfield{author}{\bibinfo{person}{Olga Russakovsky}, \bibinfo{person}{Jia Deng}, \bibinfo{person}{Hao Su}, \bibinfo{person}{Jonathan Krause}, \bibinfo{person}{Sanjeev Satheesh}, \bibinfo{person}{Sean Ma}, \bibinfo{person}{Zhiheng Huang}, \bibinfo{person}{Andrej Karpathy}, \bibinfo{person}{Aditya Khosla}, \bibinfo{person}{Michael Bernstein}, \bibinfo{person}{Alexander~C. Berg}, {and} \bibinfo{person}{Li Fei-Fei}.} \bibinfo{year}{2015}\natexlab{}.
\newblock \showarticletitle{Image{N}et large scale visual recognition challenge}.
\newblock \bibinfo{journal}{\emph{International Journal of Computer Vision (IJCV)}} \bibinfo{volume}{115}, \bibinfo{number}{3} (\bibinfo{year}{2015}), \bibinfo{pages}{211--252}.
\newblock
\href{https://doi.org/10.1007/s11263-015-0816-y}{doi:\nolinkurl{10.1007/s11263-015-0816-y}}


\bibitem[Shah et~al\mbox{.}(2024)]%
        {Flash3}
\bibfield{author}{\bibinfo{person}{Jay Shah}, \bibinfo{person}{Ganesh Bikshandi}, \bibinfo{person}{Ying Zhang}, \bibinfo{person}{Vijay Thakkar}, \bibinfo{person}{Pradeep Ramani}, {and} \bibinfo{person}{Tri Dao}.} \bibinfo{year}{2024}\natexlab{}.
\newblock \showarticletitle{Flash{A}ttention-3: {F}ast and accurate attention with asynchrony and low-precision}. In \bibinfo{booktitle}{\emph{Advances in Neural Information Processing Systems}}, \bibfield{editor}{\bibinfo{person}{A.~Globerson}, \bibinfo{person}{L.~Mackey}, \bibinfo{person}{D.~Belgrave}, \bibinfo{person}{A.~Fan}, \bibinfo{person}{U.~Paquet}, \bibinfo{person}{J.~Tomczak}, {and} \bibinfo{person}{C.~Zhang}} (Eds.), Vol.~\bibinfo{volume}{37}. \bibinfo{publisher}{Curran Associates, Inc.}, \bibinfo{pages}{68658--68685}.
\newblock


\bibitem[Shaw et~al\mbox{.}(2018)]%
        {RPE}
\bibfield{author}{\bibinfo{person}{Peter Shaw}, \bibinfo{person}{Jakob Uszkoreit}, {and} \bibinfo{person}{Ashish Vaswani}.} \bibinfo{year}{2018}\natexlab{}.
\newblock \showarticletitle{Self-attention with relative position representations}. In \bibinfo{booktitle}{\emph{Proceedings of the 2018 Conference of the North {A}merican Chapter of the Association for Computational Linguistics: Human Language Technologies, Volume 2 (Short Papers)}}, \bibfield{editor}{\bibinfo{person}{Marilyn Walker}, \bibinfo{person}{Heng Ji}, {and} \bibinfo{person}{Amanda Stent}} (Eds.). \bibinfo{publisher}{Association for Computational Linguistics}, \bibinfo{address}{New Orleans, Louisiana}, \bibinfo{pages}{464--468}.
\newblock
\href{https://doi.org/10.18653/v1/N18-2074}{doi:\nolinkurl{10.18653/v1/N18-2074}}


\bibitem[Shen et~al\mbox{.}(2021)]%
        {2SOFT}
\bibfield{author}{\bibinfo{person}{Zhuoran Shen}, \bibinfo{person}{Mingyuan Zhang}, \bibinfo{person}{Haiyu Zhao}, \bibinfo{person}{Shuai Yi}, {and} \bibinfo{person}{Hongsheng Li}.} \bibinfo{year}{2021}\natexlab{}.
\newblock \showarticletitle{{E}fficient {A}ttention: {A}ttention with linear complexities}. In \bibinfo{booktitle}{\emph{Proceedings of the IEEE/CVF Winter Conference on Applications of Computer Vision (WACV)}}. \bibinfo{pages}{3531--3539}.
\newblock


\bibitem[Tay et~al\mbox{.}(2020)]%
        {SPARSE4}
\bibfield{author}{\bibinfo{person}{Yi Tay}, \bibinfo{person}{Dara Bahri}, \bibinfo{person}{Liu Yang}, \bibinfo{person}{Donald Metzler}, {and} \bibinfo{person}{Da-Cheng Juan}.} \bibinfo{year}{2020}\natexlab{}.
\newblock \showarticletitle{Sparse sinkhorn attention}. In \bibinfo{booktitle}{\emph{Proceedings of the 37th International Conference on Machine Learning}} \emph{(\bibinfo{series}{Proceedings of Machine Learning Research}, Vol.~\bibinfo{volume}{119})}, \bibfield{editor}{\bibinfo{person}{Hal~Daumé III} {and} \bibinfo{person}{Aarti Singh}} (Eds.). \bibinfo{publisher}{PMLR}, \bibinfo{pages}{9438--9447}.
\newblock


\bibitem[Tay et~al\mbox{.}(2021)]%
        {LRA}
\bibfield{author}{\bibinfo{person}{Yi Tay}, \bibinfo{person}{Mostafa Dehghani}, \bibinfo{person}{Samira Abnar}, \bibinfo{person}{Yikang Shen}, \bibinfo{person}{Dara Bahri}, \bibinfo{person}{Philip Pham}, \bibinfo{person}{Jinfeng Rao}, \bibinfo{person}{Liu Yang}, \bibinfo{person}{Sebastian Ruder}, {and} \bibinfo{person}{Donald Metzler}.} \bibinfo{year}{2021}\natexlab{}.
\newblock \showarticletitle{Long {R}ange {A}rena : {A} benchmark for efficient transformers}. In \bibinfo{booktitle}{\emph{International Conference on Learning Representations}}.
\newblock
\urldef\tempurl%
\url{https://openreview.net/forum?id=qVyeW-grC2k}
\showURL{%
\tempurl}


\bibitem[Vaswani et~al\mbox{.}(2017)]%
        {ATA}
\bibfield{author}{\bibinfo{person}{Ashish Vaswani}, \bibinfo{person}{Noam Shazeer}, \bibinfo{person}{Niki Parmar}, \bibinfo{person}{Jakob Uszkoreit}, \bibinfo{person}{Llion Jones}, \bibinfo{person}{Aidan~N Gomez}, \bibinfo{person}{\L~ukasz Kaiser}, {and} \bibinfo{person}{Illia Polosukhin}.} \bibinfo{year}{2017}\natexlab{}.
\newblock \showarticletitle{Attention is all you need}. In \bibinfo{booktitle}{\emph{Advances in Neural Information Processing Systems}}, \bibfield{editor}{\bibinfo{person}{I.~Guyon}, \bibinfo{person}{U.~Von Luxburg}, \bibinfo{person}{S.~Bengio}, \bibinfo{person}{H.~Wallach}, \bibinfo{person}{R.~Fergus}, \bibinfo{person}{S.~Vishwanathan}, {and} \bibinfo{person}{R.~Garnett}} (Eds.), Vol.~\bibinfo{volume}{30}. \bibinfo{publisher}{Curran Associates, Inc.}
\newblock


\bibitem[Wang et~al\mbox{.}(2020b)]%
        {SCCAM}
\bibfield{author}{\bibinfo{person}{Haofan Wang}, \bibinfo{person}{Zifan Wang}, \bibinfo{person}{Mengnan Du}, \bibinfo{person}{Fan Yang}, \bibinfo{person}{Zijian Zhang}, \bibinfo{person}{Sirui Ding}, \bibinfo{person}{Piotr Mardziel}, {and} \bibinfo{person}{Xia Hu}.} \bibinfo{year}{2020}\natexlab{b}.
\newblock \showarticletitle{Score-{CAM}: {S}core-weighted visual explanations for convolutional neural networks}. In \bibinfo{booktitle}{\emph{Proceedings of the IEEE/CVF Conference on Computer Vision and Pattern Recognition (CVPR) Workshops}}.
\newblock


\bibitem[Wang et~al\mbox{.}(2020a)]%
        {LIF}
\bibfield{author}{\bibinfo{person}{Sinong Wang}, \bibinfo{person}{Belinda~Z Li}, \bibinfo{person}{Madian Khabsa}, \bibinfo{person}{Han Fang}, {and} \bibinfo{person}{Hao Ma}.} \bibinfo{year}{2020}\natexlab{a}.
\newblock \showarticletitle{Linformer: {S}elf-attention with linear complexity}.
\newblock \bibinfo{journal}{\emph{arXiv preprint arXiv:2006.04768}} (\bibinfo{year}{2020}).
\newblock


\bibitem[Wortsman et~al\mbox{.}(2023)]%
        {ReLUT}
\bibfield{author}{\bibinfo{person}{Mitchell Wortsman}, \bibinfo{person}{Jaehoon Lee}, \bibinfo{person}{Justin Gilmer}, {and} \bibinfo{person}{Simon Kornblith}.} \bibinfo{year}{2023}\natexlab{}.
\newblock \showarticletitle{Replacing softmax with {R}e{LU} in vision transformers}.
\newblock \bibinfo{journal}{\emph{arXiv preprint arXiv:2309.08586}} (\bibinfo{year}{2023}).
\newblock


\bibitem[Wu et~al\mbox{.}(2024)]%
        {CURSA}
\bibfield{author}{\bibinfo{person}{Chong Wu}, \bibinfo{person}{Maolin Che}, {and} \bibinfo{person}{Hong Yan}.} \bibinfo{year}{2024}\natexlab{}.
\newblock \showarticletitle{The {CUR} decomposition of self-attention matrices in vision transformers}.
\newblock \bibinfo{journal}{\emph{TechRxiv}} (\bibinfo{year}{2024}).
\newblock
\href{https://doi.org/10.36227/techrxiv.171392846.60982484/v3}{doi:\nolinkurl{10.36227/techrxiv.171392846.60982484/v3}}


\bibitem[Wu et~al\mbox{.}(2021)]%
        {VRPE}
\bibfield{author}{\bibinfo{person}{Kan Wu}, \bibinfo{person}{Houwen Peng}, \bibinfo{person}{Minghao Chen}, \bibinfo{person}{Jianlong Fu}, {and} \bibinfo{person}{Hongyang Chao}.} \bibinfo{year}{2021}\natexlab{}.
\newblock \showarticletitle{Rethinking and improving relative position encoding for vision transformer}. In \bibinfo{booktitle}{\emph{Proceedings of the IEEE/CVF International Conference on Computer Vision (ICCV)}}. \bibinfo{pages}{10033--10041}.
\newblock


\bibitem[Xiao et~al\mbox{.}(2018)]%
        {upernet}
\bibfield{author}{\bibinfo{person}{Tete Xiao}, \bibinfo{person}{Yingcheng Liu}, \bibinfo{person}{Bolei Zhou}, \bibinfo{person}{Yuning Jiang}, {and} \bibinfo{person}{Jian Sun}.} \bibinfo{year}{2018}\natexlab{}.
\newblock \showarticletitle{Unified perceptual parsing for scene understanding}. In \bibinfo{booktitle}{\emph{Computer Vision -- ECCV 2018}}, \bibfield{editor}{\bibinfo{person}{Vittorio Ferrari}, \bibinfo{person}{Martial Hebert}, \bibinfo{person}{Cristian Sminchisescu}, {and} \bibinfo{person}{Yair Weiss}} (Eds.). \bibinfo{publisher}{Springer International Publishing}, \bibinfo{address}{Cham}, \bibinfo{pages}{432--448}.
\newblock


\bibitem[Xiong et~al\mbox{.}(2021)]%
        {NTRANS}
\bibfield{author}{\bibinfo{person}{Yunyang Xiong}, \bibinfo{person}{Zhanpeng Zeng}, \bibinfo{person}{Rudrasis Chakraborty}, \bibinfo{person}{Mingxing Tan}, \bibinfo{person}{Glenn Fung}, \bibinfo{person}{Yin Li}, {and} \bibinfo{person}{Vikas Singh}.} \bibinfo{year}{2021}\natexlab{}.
\newblock \showarticletitle{Nyströmformer: {A} {N}yström-based algorithm for approximating self-attention}.
\newblock \bibinfo{journal}{\emph{Proceedings of the AAAI Conference on Artificial Intelligence}} \bibinfo{volume}{35}, \bibinfo{number}{16} (\bibinfo{year}{2021}), \bibinfo{pages}{14138--14148}.
\newblock
\href{https://doi.org/10.1609/aaai.v35i16.17664}{doi:\nolinkurl{10.1609/aaai.v35i16.17664}}


\bibitem[Xu et~al\mbox{.}(2024a)]%
        {xu2}
\bibfield{author}{\bibinfo{person}{Renjie Xu}, \bibinfo{person}{Shenghao Feng}, \bibinfo{person}{Yimin Wei}, {and} \bibinfo{person}{Hong Yan}.} \bibinfo{year}{2024}\natexlab{a}.
\newblock \showarticletitle{{CUR} and generalized {CUR} decompositions of quaternion matrices and their applications}.
\newblock \bibinfo{journal}{\emph{Numerical Functional Analysis and Optimization}} \bibinfo{volume}{45}, \bibinfo{number}{3} (\bibinfo{year}{2024}), \bibinfo{pages}{234--258}.
\newblock


\bibitem[Xu et~al\mbox{.}(2024b)]%
        {xu1}
\bibfield{author}{\bibinfo{person}{Renjie Xu}, \bibinfo{person}{Tong Wei}, \bibinfo{person}{Yimin Wei}, {and} \bibinfo{person}{Hong Yan}.} \bibinfo{year}{2024}\natexlab{b}.
\newblock \showarticletitle{{UTV} decomposition of dual matrices and its applications}.
\newblock \bibinfo{journal}{\emph{Computational and Applied Mathematics}} \bibinfo{volume}{43}, \bibinfo{number}{1} (\bibinfo{year}{2024}), \bibinfo{pages}{41}.
\newblock


\bibitem[Xu and Wei(2024)]%
        {xu3}
\bibfield{author}{\bibinfo{person}{Renjie Xu} {and} \bibinfo{person}{Yimin Wei}.} \bibinfo{year}{2024}\natexlab{}.
\newblock \showarticletitle{Randomized quaternion matrix {UTV} decomposition and its applications in quaternion matrix optimization}.
\newblock \bibinfo{journal}{\emph{Pacific Journal of Optimization}} \bibinfo{volume}{20}, \bibinfo{number}{2} (\bibinfo{year}{2024}), \bibinfo{pages}{185--211}.
\newblock


\bibitem[Yu(2015)]%
        {atrous}
\bibfield{author}{\bibinfo{person}{F Yu}.} \bibinfo{year}{2015}\natexlab{}.
\newblock \showarticletitle{Multi-scale context aggregation by dilated convolutions}.
\newblock \bibinfo{journal}{\emph{arXiv preprint arXiv:1511.07122}} (\bibinfo{year}{2015}).
\newblock


\bibitem[Zaheer et~al\mbox{.}(2020)]%
        {SPARSE5}
\bibfield{author}{\bibinfo{person}{Manzil Zaheer}, \bibinfo{person}{Guru Guruganesh}, \bibinfo{person}{Kumar~Avinava Dubey}, \bibinfo{person}{Joshua Ainslie}, \bibinfo{person}{Chris Alberti}, \bibinfo{person}{Santiago Ontanon}, \bibinfo{person}{Philip Pham}, \bibinfo{person}{Anirudh Ravula}, \bibinfo{person}{Qifan Wang}, \bibinfo{person}{Li Yang}, {and} \bibinfo{person}{Amr Ahmed}.} \bibinfo{year}{2020}\natexlab{}.
\newblock \showarticletitle{{B}ig {B}ird: {T}ransformers for longer sequences}. In \bibinfo{booktitle}{\emph{Advances in Neural Information Processing Systems}}, \bibfield{editor}{\bibinfo{person}{H.~Larochelle}, \bibinfo{person}{M.~Ranzato}, \bibinfo{person}{R.~Hadsell}, \bibinfo{person}{M.F. Balcan}, {and} \bibinfo{person}{H.~Lin}} (Eds.), Vol.~\bibinfo{volume}{33}. \bibinfo{publisher}{Curran Associates, Inc.}, \bibinfo{pages}{17283--17297}.
\newblock


\bibitem[Zhang et~al\mbox{.}(2022)]%
        {NesT}
\bibfield{author}{\bibinfo{person}{Zizhao Zhang}, \bibinfo{person}{Han Zhang}, \bibinfo{person}{Long Zhao}, \bibinfo{person}{Ting Chen}, \bibinfo{person}{Sercan~{\"O} Arik}, {and} \bibinfo{person}{Tomas Pfister}.} \bibinfo{year}{2022}\natexlab{}.
\newblock \showarticletitle{Nested {H}ierarchical {T}ransformer: {T}owards accurate, data-efficient and interpretable visual understanding}.
\newblock \bibinfo{journal}{\emph{Proceedings of the AAAI Conference on Artificial Intelligence}} \bibinfo{volume}{36}, \bibinfo{number}{3} (\bibinfo{year}{2022}), \bibinfo{pages}{3417--3425}.
\newblock
\href{https://doi.org/10.1609/aaai.v36i3.20252}{doi:\nolinkurl{10.1609/aaai.v36i3.20252}}


\bibitem[Zhao et~al\mbox{.}(2019)]%
        {SPARSE1}
\bibfield{author}{\bibinfo{person}{Guangxiang Zhao}, \bibinfo{person}{Junyang Lin}, \bibinfo{person}{Zhiyuan Zhang}, \bibinfo{person}{Xuancheng Ren}, \bibinfo{person}{Qi Su}, {and} \bibinfo{person}{Xu Sun}.} \bibinfo{year}{2019}\natexlab{}.
\newblock \showarticletitle{Explicit {S}parse {T}ransformer: {C}oncentrated attention through explicit selection}.
\newblock \bibinfo{journal}{\emph{CoRR}}  \bibinfo{volume}{abs/1912.11637} (\bibinfo{year}{2019}).
\newblock
\showeprint[arXiv]{1912.11637}
\urldef\tempurl%
\url{http://arxiv.org/abs/1912.11637}
\showURL{%
\tempurl}


\bibitem[Zhou et~al\mbox{.}(2017)]%
        {ADE20K}
\bibfield{author}{\bibinfo{person}{Bolei Zhou}, \bibinfo{person}{Hang Zhao}, \bibinfo{person}{Xavier Puig}, \bibinfo{person}{Sanja Fidler}, \bibinfo{person}{Adela Barriuso}, {and} \bibinfo{person}{Antonio Torralba}.} \bibinfo{year}{2017}\natexlab{}.
\newblock \showarticletitle{Scene parsing through {ADE20K} dataset}. In \bibinfo{booktitle}{\emph{2017 IEEE Conference on Computer Vision and Pattern Recognition (CVPR)}}. \bibinfo{pages}{5122--5130}.
\newblock
\href{https://doi.org/10.1109/CVPR.2017.544}{doi:\nolinkurl{10.1109/CVPR.2017.544}}


\bibitem[Zhu et~al\mbox{.}(2024)]%
        {vim}
\bibfield{author}{\bibinfo{person}{Lianghui Zhu}, \bibinfo{person}{Bencheng Liao}, \bibinfo{person}{Qian Zhang}, \bibinfo{person}{Xinlong Wang}, \bibinfo{person}{Wenyu Liu}, {and} \bibinfo{person}{Xinggang Wang}.} \bibinfo{year}{2024}\natexlab{}.
\newblock \showarticletitle{Vision {M}amba: {E}fficient visual representation learning with bidirectional state space model}. In \bibinfo{booktitle}{\emph{Proceedings of the 41st International Conference on Machine Learning}} \emph{(\bibinfo{series}{Proceedings of Machine Learning Research}, Vol.~\bibinfo{volume}{235})}, \bibfield{editor}{\bibinfo{person}{Ruslan Salakhutdinov}, \bibinfo{person}{Zico Kolter}, \bibinfo{person}{Katherine Heller}, \bibinfo{person}{Adrian Weller}, \bibinfo{person}{Nuria Oliver}, \bibinfo{person}{Jonathan Scarlett}, {and} \bibinfo{person}{Felix Berkenkamp}} (Eds.). \bibinfo{publisher}{PMLR}, \bibinfo{pages}{62429--62442}.
\newblock


\end{thebibliography}

\newpage
\appendix

\section{Approximation Error Bound Analysis}
\label{AEB}
Eq. (\ref{elfatt}) of the main body of the paper can be derived as follows,

\begin{small}
\begin{equation*}
\begin{split}
&{\rm exp}\left(\textit{\textbf{Q}}\textit{\textbf{K}}^{\top}\right)\textit{\textbf{V}} = \left({\rm exp}\left(\bar{\textit{\textbf{Q}}}\bar{\textit{\textbf{K}}}^{\top}\right)\odot{\rm exp}\left(\tilde{\textit{\textbf{Q}}}\tilde{\textit{\textbf{K}}}^{\top}\right)\right)\textit{\textbf{V}}\\ 
&= {\rm exp}\left(\bar{\textit{\textbf{Q}}}\bar{\textit{\textbf{K}}}^{\top}\right)\odot{\rm exp}\left(\tilde{\textit{\textbf{Q}}}\tilde{\textit{\textbf{K}}}^{\top}\right)\left\lbrack\bar{\textit{\textbf{V}}}, \tilde{\textit{\textbf{V}}}\right\rbrack\\
& = \left\lbrack\left({\rm exp}\left(\bar{\textit{\textbf{Q}}}\bar{\textit{\textbf{K}}}^{\top}\right)\odot{\rm exp}\left(\tilde{\textit{\textbf{Q}}}\tilde{\textit{\textbf{K}}}^{\top}\right)\right)\bar{\textit{\textbf{V}}}, \quad\left({\rm exp}\left(\bar{\textit{\textbf{Q}}}\bar{\textit{\textbf{K}}}^{\top}\right)\odot{\rm exp}\left(\tilde{\textit{\textbf{Q}}}\tilde{\textit{\textbf{K}}}^{\top}\right)\right)\tilde{\textit{\textbf{V}}}\right\rbrack\\
& \approx \left\lbrack{\rm exp}\left(\bar{\textit{\textbf{Q}}}\bar{\textit{\textbf{K}}}^{\top}\right)\bar{\textit{\textbf{V}}}, {\rm exp}\left(\tilde{\textit{\textbf{Q}}}\tilde{\textit{\textbf{K}}}^{\top}\right)\tilde{\textit{\textbf{V}}}\right\rbrack \\
&\approx \left\lbrack{\rm exp}\left(\bar{\textit{\textbf{Q}}}\right){\rm exp}\left(\bar{\textit{\textbf{K}}}\right)^{\top}\bar{\textit{\textbf{V}}}, \left({\rm exp}\left(\tilde{\textit{\textbf{Q}}}\tilde{\textit{\textbf{K}}}^{\top}\right)\odot\textit{\textbf{Z}}\right)\tilde{\textit{\textbf{V}}}\right\rbrack,\\
\end{split}
\end{equation*}
\end{small}

\noindent where $\odot$ denotes the Hadamard product \cite{hadamard}, ${\rm exp}\left(\tilde{\textit{\textbf{Q}}}\tilde{\textit{\textbf{K}}}^{\top}\right)\odot\textit{\textbf{Z}}\tilde{\textit{\textbf{V}}}$ is equivalent to $g\left({\rm exp}\left(f(\tilde{\textit{\textbf{Q}}}) f(\tilde{\textit{\textbf{K}}})^{\top}\right)f(\tilde{\textit{\textbf{V}}})\right)$, and $\textit{\textbf{Z}} \in \mathbb{R}^{m\times m}$ is a matrix as follows,

\begin{small}
\begin{equation*}
\textit{\textbf{Z}} = \textbf{I}_b\otimes\textit{\textbf{U}}_{(m/b)},
\end{equation*}
\end{small}

\noindent with $\otimes$ denoting the Kronecker product \cite{hadamard} and $\textit{\textbf{U}}_{(m/b)} \in \mathbb{R}^{(m/b)\times (m/b)}$ being the all-ones matrix. It is obvious to see that

\begin{small}
\begin{equation*}
\begin{split}
    &\left\lbrack{\rm exp}\left(\bar{\textit{\textbf{Q}}}\right){\rm exp}\left(\bar{\textit{\textbf{K}}}\right)^{\top}\bar{\textit{\textbf{V}}}, \left({\rm exp}\left(\tilde{\textit{\textbf{Q}}}\tilde{\textit{\textbf{K}}}^{\top}\right)\odot\textit{\textbf{Z}}\right)\tilde{\textit{\textbf{V}}}\right\rbrack-{\rm exp}\left(\textit{\textbf{Q}}\textit{\textbf{K}}^{\top}\right)\textit{\textbf{V}}\\
    &=\left[\left({\rm exp}\left(\bar{\textit{\textbf{Q}}}\right){\rm exp}\left(\bar{\textit{\textbf{K}}}\right)^{\top} - {\rm exp}\left(\bar{\textit{\textbf{Q}}}\bar{\textit{\textbf{K}}}^{\top}\right)\odot{\rm exp}\left(\tilde{\textit{\textbf{Q}}}\tilde{\textit{\textbf{K}}}^{\top}\right)\right)\bar{\textit{\textbf{V}}},\right.\\
    &\quad\left.\left({\rm exp}\left(\tilde{\textit{\textbf{Q}}}\tilde{\textit{\textbf{K}}}^{\top}\right)\odot\textit{\textbf{Z}} - {\rm exp}\left(\bar{\textit{\textbf{Q}}}\bar{\textit{\textbf{K}}}^{\top}\right)\odot{\rm exp}\left(\tilde{\textit{\textbf{Q}}}\tilde{\textit{\textbf{K}}}^{\top}\right)\right)\tilde{\textit{\textbf{V}}}\right],
\end{split}
\end{equation*}
\end{small}

\noindent which implies that

\begin{small}
\begin{equation}
\begin{split}
    &\left\|\left\lbrack{\rm exp}\left(\bar{\textit{\textbf{Q}}}\right){\rm exp}\left(\bar{\textit{\textbf{K}}}\right)^{\top}\bar{\textit{\textbf{V}}}, \left({\rm exp}\left(\tilde{\textit{\textbf{Q}}}\tilde{\textit{\textbf{K}}}^{\top}\right)\odot\textit{\textbf{Z}}\right)\tilde{\textit{\textbf{V}}}\right\rbrack - {\rm exp}\left(\textit{\textbf{Q}}\textit{\textbf{K}}^{\top}\right)\textit{\textbf{V}}\right\|_\xi\\
    &\leq\left\|{\rm exp}\left(\bar{\textit{\textbf{Q}}}\right){\rm exp}\left(\bar{\textit{\textbf{K}}}\right)^{\top} - {\rm exp}\left(\bar{\textit{\textbf{Q}}}\bar{\textit{\textbf{K}}}^{\top}\right)\odot{\rm exp}\left(\tilde{\textit{\textbf{Q}}}\tilde{\textit{\textbf{K}}}^{\top}\right)\right\|_\xi\left\|\bar{\textit{\textbf{V}}}\right\|_\xi\\
    &+\left\|{\rm exp}\left(\tilde{\textit{\textbf{Q}}}\tilde{\textit{\textbf{K}}}^{\top}\right)\odot\textit{\textbf{Z}}-{\rm exp}\left(\bar{\textit{\textbf{Q}}}\bar{\textit{\textbf{K}}}^{\top}\right)\odot{\rm exp}\left(\tilde{\textit{\textbf{Q}}}\tilde{\textit{\textbf{K}}}^{\top}\right)\right\|_\xi\left\|\tilde{\textit{\textbf{V}}}\right\|_\xi,
\end{split}
\label{elfatterr2}
\end{equation}
\end{small}

\noindent where $\xi=2$ denotes the spectral norm and $\xi=F$ denotes the Frobenius norm. For given two matrices $\textit{\textbf{A}}\in\mathbb{R}^{m\times n}$ and $\textit{\textbf{B}}\in\mathbb{R}^{n\times p}$, it follows from \cite{golub2013matrix} [Section 2.3.1] that $\|\textit{\textbf{A}}\textit{\textbf{B}}\|_\xi\leq \|\textit{\textbf{A}}\|_\xi\|\textit{\textbf{B}}\|_\xi$.

For the second term in the right-hand side of Inequality (\ref{elfatterr2}), it is easy to see that

\begin{small}
\begin{equation*}
\begin{split}
    &\left\|{\rm exp}\left(\tilde{\textit{\textbf{Q}}}\tilde{\textit{\textbf{K}}}^{\top}\right)\odot\textit{\textbf{Z}} - {\rm exp}\left(\bar{\textit{\textbf{Q}}}\bar{\textit{\textbf{K}}}^{\top}\right)\odot{\rm exp}\left(\tilde{\textit{\textbf{Q}}}\tilde{\textit{\textbf{K}}}^{\top}\right)\right\|_\xi\left\|\tilde{\textit{\textbf{V}}}\right\|_\xi\\
    &\leq\left\|\textit{\textbf{Z}} - {\rm exp}\left(\bar{\textit{\textbf{Q}}}\bar{\textit{\textbf{K}}}^{\top}\right)\right\|_\xi\left\|{\rm exp}\left(\tilde{\textit{\textbf{Q}}}\tilde{\textit{\textbf{K}}}^{\top}\right)\right\|_\xi\left\|\tilde{\textit{\textbf{V}}}\right\|_\xi.
\end{split}
\end{equation*}
\end{small}

For the first term in the right-hand side of Inequality (\ref{elfatterr2}), we have the following theorem.

\begin{theorem}
    Let $\emph{\textit{\textbf{U}}}_m \in \mathbb{R}^{m \times m}$ be an all-ones matrix. For any two vectors $\bar{\textit{\textbf{q}}}\in\mathbb{R}^{c_1}$ from $\bar{\textit{\textbf{Q}}} \in \mathbb{R}^{m \times c_1}$ and $\bar{\textit{\textbf{k}}}\in\mathbb{R}^{c_1}$ from $\bar{\textit{\textbf{K}}} \in \mathbb{R}^{m \times c_1}$, let $\mathfrak{M}>0$ and $\mathscr{M}>0$ be the maximum and minimum of ${\rm exp}\left(\bar{\textit{\textbf{q}}}\bar{\textit{\textbf{k}}}^{\top}+0.5-\left(\bar{q}_{i}+\bar{k}_{i}\right)\right)$, respectively, and $\bar{q}_{i}$ and $\bar{k}_{i}$ are the elements at position $i$ of vectors $\bar{\textit{\textbf{q}}}$ and $\bar{\textit{\textbf{k}}}$, respectively. If $\left|\frac{c_1}{\mathscr{M}{\rm exp}\left(-0.5\right)} - 1\right|\geq \left|\frac{c_1}{\mathfrak{M}{\rm exp}\left(-0.5\right)} - 1\right|$, the following inequality holds,
    
    \begin{small}
    \begin{equation*}
       \begin{split}
           &\left\|{\rm exp}\left(\bar{\textit{\textbf{Q}}}\right){\rm exp}\left(\bar{\textit{\textbf{K}}}\right)^{\top} - {\rm exp}\left(\bar{\textit{\textbf{Q}}}\bar{\textit{\textbf{K}}}^{\top}\right)\odot{\rm exp}\left(\tilde{\textit{\textbf{Q}}}\tilde{\textit{\textbf{K}}}^{\top}\right)\right\|_\xi \\
           &\leq \frac{c_1}{\mathscr{M}{\rm exp}\left(-0.5\right)}\left\|{\rm exp}\left(\bar{\textit{\textbf{Q}}}\bar{\textit{\textbf{K}}}^{\top}\right)\right\|_\xi\left\|\emph{\textit{\textbf{U}}}_m - {\rm exp}\left(\tilde{\textit{\textbf{Q}}}\tilde{\textit{\textbf{K}}}^{\top}\right)\right\|_\xi  \\
           &+\left|\frac{c_1}{\mathscr{M}{\rm exp}\left(-0.5\right)}-1\right|\left\|{\rm exp}\left(\bar{\textit{\textbf{Q}}}\bar{\textit{\textbf{K}}}^{\top}\right)\right\|_\xi\left\|{\rm exp}\left(\tilde{\textit{\textbf{Q}}}\tilde{\textit{\textbf{K}}}^{\top}\right)\right\|_\xi.\\
       \end{split}
    \end{equation*}
    \end{small}
    
    If $\left|\frac{c_1}{\mathscr{M}{\rm exp}\left(-0.5\right)} - 1\right| < \left|\frac{c_1}{\mathfrak{M}{\rm exp}\left(-0.5\right)} - 1\right|$, the following inequality holds,
    
    \begin{small}
    \begin{equation*}
       \begin{split}
           &\left\|{\rm exp}\left(\bar{\textit{\textbf{Q}}}\right){\rm exp}\left(\bar{\textit{\textbf{K}}}\right)^{\top} - {\rm exp}\left(\bar{\textit{\textbf{Q}}}\bar{\textit{\textbf{K}}}^{\top}\right)\odot{\rm exp}\left(\tilde{\textit{\textbf{Q}}}\tilde{\textit{\textbf{K}}}^{\top}\right)\right\|_\xi \\
           &\leq 
           \frac{c_1}{\mathscr{M}{\rm exp}\left(-0.5\right)}\left\|{\rm exp}\left(\bar{\textit{\textbf{Q}}}\bar{\textit{\textbf{K}}}^{\top}\right)\right\|_\xi\left\|\emph{\textit{\textbf{U}}}_m - {\rm exp}\left(\tilde{\textit{\textbf{Q}}}\tilde{\textit{\textbf{K}}}^{\top}\right)\right\|_\xi  \\
           &+\left|\frac{c_1}{\mathfrak{M}{\rm exp}\left(-0.5\right)}-1\right|\left\|{\rm exp}\left(\bar{\textit{\textbf{Q}}}\bar{\textit{\textbf{K}}}^{\top}\right)\right\|_\xi\left\|{\rm exp}\left(\tilde{\textit{\textbf{Q}}}\tilde{\textit{\textbf{K}}}^{\top}\right)\right\|_\xi.\\
       \end{split}
    \end{equation*}
    \end{small}
    \label{Theory1}
\end{theorem}

Before proving Theorem \ref{Theory1}, we introduce the following lemma (see \cite{CURSA}).
\begin{lemma}
    For any two vectors $\bar{\textit{\textbf{q}}}\in\mathbb{R}^{c_1}$ from $\bar{\textit{\textbf{Q}}}\in\mathbb{R}^{m\times c_1}$ and $\bar{\textit{\textbf{k}}}\in\mathbb{R}^{c_1}$ from $\bar{\textit{\textbf{K}}}\in\mathbb{R}^{m\times c_1}$, let 
     
     \begin{small}
    \begin{equation*}
        \begin{split}
        D_{\bar{\textit{\textbf{q}}},\bar{\textit{\textbf{k}}}}&=\max_{i=1,2,\dots,c_1}{\rm exp}\left(\bar{\textit{\textbf{q}}}\bar{\textit{\textbf{k}}}^{\top}+0.5-\left(\bar{q}_{i}+\bar{k}_{i}\right)\right)>0,\\
        d_{\bar{\textit{\textbf{q}}},\bar{\textit{\textbf{k}}}}&=\min_{i=1,2,\dots,c_1}{\rm exp}\left(\bar{\textit{\textbf{q}}}\bar{\textit{\textbf{k}}}^{\top}+0.5-\left(\bar{q}_{i}+\bar{k}_{i}\right)\right)>0,
        \end{split}
    \end{equation*}
    \end{small}
    
    \noindent where $\bar{q}_{i}$ and $\bar{k}_{i}$ are the elements at position $i$ of vectors $\bar{\textit{\textbf{q}}}$ and $\bar{\textit{\textbf{k}}}$, respectively. The following inequalities hold,
    
    \begin{small}
    \begin{equation*}
        \begin{split}
            \frac{{\rm exp}\left(\bar{\textit{\textbf{q}}}\right){\rm exp}\left(\bar{\textit{\textbf{k}}}\right)^{\top}}{{\rm exp}\left( \bar{\textit{\textbf{q}}}\bar{\textit{\textbf{k}}}^{\top}\right)} &\leq  \frac{c_1}{d_{\bar{\textit{\textbf{q}}},\bar{\textit{\textbf{k}}}}{\rm exp}\left(-0.5\right)},\\
            \frac{{\rm exp}\left(\bar{\textit{\textbf{q}}}\right){\rm exp}\left(\bar{\textit{\textbf{k}}}\right)^{\top}}{{\rm exp}\left( \bar{\textit{\textbf{q}}}\bar{\textit{\textbf{k}}}^{\top}\right)} &\geq \frac{c_1}{D_{\bar{\textit{\textbf{q}}},\bar{\textit{\textbf{k}}}}{\rm exp}\left(-0.5\right)}.
        \end{split}
    \end{equation*}
    \end{small}
    \label{Lemma1}
\end{lemma}
The following corollary is easily obtained from Lemma \ref{Lemma1}.
\begin{corollary}
    For two matrices $\bar{\textit{\textbf{Q}}}\in\mathbb{R}^{m\times c_1}$ and $\bar{\textit{\textbf{K}}}\in\mathbb{R}^{m\times c_1}$, let 

    \begin{small}
    \begin{equation*}
        \begin{split}
        \mathfrak{M}&=\max_{i_1,i_2=1,2,\dots,m\atop j=1,2,\dots,c_1}{\rm exp}\left(\bar{\textit{\textbf{q}}}_{i_1}\bar{\textit{\textbf{k}}}_{i_2}^{\top}+0.5-\left(\bar{q}_{i_1j}+\bar{k}_{i_2j}\right)\right)>0,\\
        \mathscr{M}&=\min_{i_1,i_2=1,2,\dots,m\atop j=1,2,\dots,c_1}{\rm exp}\left(\bar{\textit{\textbf{q}}}_{i_1}\bar{\textit{\textbf{k}}}_{i_2}^{\top}+0.5-\left(\bar{q}_{i_1j}+\bar{k}_{i_2j}\right)\right)>0,
        \end{split}
    \end{equation*}
    \end{small}
    
    \noindent where $\bar{q}_{i_1j}$ and $\bar{k}_{i_2j}$ are the elements at position $j$ of the vector $\bar{\textit{\textbf{q}}}_{i_1}=\bar{\textit{\textbf{Q}}}(i_1,:)$ and the vector $\bar{\textit{\textbf{k}}}_{i_2}=\bar{\textit{\textbf{K}}}(i_2,:)$, respectively. The following inequalities hold,
    
    \begin{small}
    \begin{equation*}
        \begin{split}
            \frac{\left\|{\rm exp}\left(\bar{\textit{\textbf{Q}}}\right){\rm exp}\left(\bar{\textit{\textbf{K}}}\right)^{\top}\right\|_\xi}{\left\|{\rm exp}\left(\bar{\textit{\textbf{Q}}}\bar{\textit{\textbf{K}}}^{\top}\right)\right\|_\xi}&\leq\frac{c_1}{\mathscr{M}{\rm exp}\left(-0.5\right)},\\
            \frac{\left\|{\rm exp}\left(\bar{\textit{\textbf{Q}}}\right){\rm exp}\left(\bar{\textit{\textbf{K}}}\right)^{\top}\right\|_\xi}{\left\|{\rm exp}\left(\bar{\textit{\textbf{Q}}}\bar{\textit{\textbf{K}}}^{\top}\right)\right\|_\xi}&\geq\frac{c_1}{\mathfrak{M}{\rm exp}\left(-0.5\right)}.
        \end{split}
    \end{equation*}
    \end{small}
    \label{Corollary1}
\end{corollary}

\begin{proof}
For any matrices $\bar{\textit{\textbf{Q}}}\in\mathbb{R}^{m \times c_1}$, $\tilde{\textit{\textbf{Q}}}\in\mathbb{R}^{m \times c_2}$, $\bar{\textit{\textbf{K}}}\in\mathbb{R}^{m \times c_1}$, and $\tilde{\textit{\textbf{K}}}\in\mathbb{R}^{m \times c_2}$, we have

\begin{small}
\begin{equation*}
\begin{split}
&\left\|{\rm exp}\left(\bar{\textit{\textbf{Q}}}\right){\rm exp}\left(\bar{\textit{\textbf{K}}}\right)^{\top} - {\rm exp}\left(\bar{\textit{\textbf{Q}}}\bar{\textit{\textbf{K}}}^{\top}\right)\odot{\rm exp}\left(\tilde{\textit{\textbf{Q}}}\tilde{\textit{\textbf{K}}}^{\top}\right)\right\|_\xi \\
&= \left\|{\rm exp}\left(\bar{\textit{\textbf{Q}}}\right)
{\rm exp}\left(\bar{\textit{\textbf{K}}}\right)^{\top} - {\rm exp}\left(\bar{\textit{\textbf{Q}}}\right){\rm exp}\left(\bar{\textit{\textbf{K}}}\right)^{\top}\odot{\rm exp}\left(\tilde{\textit{\textbf{Q}}}\tilde{\textit{\textbf{K}}}^{\top}\right) +\right.\\
& \left.{\rm exp}\left(\bar{\textit{\textbf{Q}}}\right)
{\rm exp}\left(\bar{\textit{\textbf{K}}}\right)^{\top}\odot{\rm exp}\left(\tilde{\textit{\textbf{Q}}}\tilde{\textit{\textbf{K}}}^{\top}\right) - {\rm exp}\left(\bar{\textit{\textbf{Q}}}\bar{\textit{\textbf{K}}}^{\top}\right)\odot{\rm exp}\left(\tilde{\textit{\textbf{Q}}}\tilde{\textit{\textbf{K}}}^{\top}\right)\right\|_\xi \\ 
&\leq\left\|{\rm exp}\left(\bar{\textit{\textbf{Q}}}\right){\rm exp}\left(\bar{\textit{\textbf{K}}}\right)^{\top} - {\rm exp}\left(\bar{\textit{\textbf{Q}}}\right){\rm exp}\left(\bar{\textit{\textbf{K}}}\right)^{\top}\odot{\rm exp}\left(\tilde{\textit{\textbf{Q}}}\tilde{\textit{\textbf{K}}}^{\top}\right)\right\|_\xi +\\
&\left\|{\rm exp}\left(\bar{\textit{\textbf{Q}}}\right){\rm exp}\left(\bar{\textit{\textbf{K}}}\right)^{\top}\odot{\rm exp}\left(\tilde{\textit{\textbf{Q}}}\tilde{\textit{\textbf{K}}}^{\top}\right) - {\rm exp}\left(\bar{\textit{\textbf{Q}}}\bar{\textit{\textbf{K}}}^{\top}\right)\odot{\rm exp}\left(\tilde{\textit{\textbf{Q}}}\tilde{\textit{\textbf{K}}}^{\top}\right)\right\|_\xi\\
& \leq\left\|{\rm exp}\left(\bar{\textit{\textbf{Q}}}\right){\rm exp}\left(\bar{\textit{\textbf{K}}}\right)^{\top}\right\|_\xi\left\|\textit{\textbf{U}}_m - {\rm exp}\left(\tilde{\textit{\textbf{Q}}}\tilde{\textit{\textbf{K}}}^{\top}\right)\right\|_\xi \\
& +\left\|{\rm exp}\left(\bar{\textit{\textbf{Q}}}\right){\rm exp}\left(\bar{\textit{\textbf{K}}}\right)^{\top}- {\rm exp}\left(\bar{\textit{\textbf{Q}}}\bar{\textit{\textbf{K}}}^{\top}\right)\right\|_\xi\left\|{\rm exp}\left(\tilde{\textit{\textbf{Q}}}\tilde{\textit{\textbf{K}}}^{\top}\right)\right\|_\xi.\\
\end{split}
\end{equation*}
\end{small}

Following from Corollary \ref{Corollary1}, if $\left|\frac{c_1}{\mathscr{M}{\rm exp}\left(-0.5\right)} - 1\right|\geq \left|\frac{c_1}{\mathfrak{M}{\rm exp}\left(-0.5\right)} - 1\right|$, we have

\begin{small}
\begin{equation*}
\begin{split}
&\left\|{\rm exp}\left(\bar{\textit{\textbf{Q}}}\right){\rm exp}\left(\bar{\textit{\textbf{K}}}\right)^{\top} - {\rm exp}\left(\bar{\textit{\textbf{Q}}}\bar{\textit{\textbf{K}}}^{\top}\right)\odot{\rm exp}\left(\tilde{\textit{\textbf{Q}}}\tilde{\textit{\textbf{K}}}^{\top}\right)\right\|_\xi \\
           &\leq \frac{c_1}{\mathscr{M}{\rm exp}\left(-0.5\right)}\left\|{\rm exp}\left(\bar{\textit{\textbf{Q}}}\bar{\textit{\textbf{K}}}^{\top}\right)\right\|_\xi\left\|\emph{\textit{\textbf{U}}}_m - {\rm exp}\left(\tilde{\textit{\textbf{Q}}}\tilde{\textit{\textbf{K}}}^{\top}\right)\right\|_\xi  \\
           &+\left|\frac{c_1}{\mathscr{M}{\rm exp}\left(-0.5\right)}-1\right|\left\|{\rm exp}\left(\bar{\textit{\textbf{Q}}}\bar{\textit{\textbf{K}}}^{\top}\right)\right\|_\xi\left\|{\rm exp}\left(\tilde{\textit{\textbf{Q}}}\tilde{\textit{\textbf{K}}}^{\top}\right)\right\|_\xi.\\
\end{split}
\end{equation*}
\end{small}

If $\left|\frac{c_1}{\mathscr{M}{\rm exp}\left(-0.5\right)} - 1\right| < \left|\frac{c_1}{\mathfrak{M}{\rm exp}\left(-0.5\right)} - 1\right|$, one has

\begin{small}
\begin{equation*}
\begin{split}
           &\left\|{\rm exp}\left(\bar{\textit{\textbf{Q}}}\right){\rm exp}\left(\bar{\textit{\textbf{K}}}\right)^{\top} - {\rm exp}\left(\bar{\textit{\textbf{Q}}}\bar{\textit{\textbf{K}}}^{\top}\right)\odot{\rm exp}\left(\tilde{\textit{\textbf{Q}}}\tilde{\textit{\textbf{K}}}^{\top}\right)\right\|_\xi \\
           &\leq 
           \frac{c_1}{\mathscr{M}{\rm exp}\left(-0.5\right)}\left\|{\rm exp}\left(\bar{\textit{\textbf{Q}}}\bar{\textit{\textbf{K}}}^{\top}\right)\right\|_\xi\left\|\emph{\textit{\textbf{U}}}_m - {\rm exp}\left(\tilde{\textit{\textbf{Q}}}\tilde{\textit{\textbf{K}}}^{\top}\right)\right\|_\xi  \\
           &+\left|\frac{c_1}{\mathfrak{M}{\rm exp}\left(-0.5\right)}-1\right|\left\|{\rm exp}\left(\bar{\textit{\textbf{Q}}}\bar{\textit{\textbf{K}}}^{\top}\right)\right\|_\xi\left\|{\rm exp}\left(\tilde{\textit{\textbf{Q}}}\tilde{\textit{\textbf{K}}}^{\top}\right)\right\|_\xi.\\
\end{split}
\end{equation*}
\end{small}

The proof is completed. 
\end{proof}

We now consider the upper bound for the right-hand side of Inequality (\ref{elfatterr2}). If $\left|\frac{c_1}{\mathscr{M}{\rm exp}\left(-0.5\right)} - 1\right|\geq \left|\frac{c_1}{\mathfrak{M}{\rm exp}\left(-0.5\right)} - 1\right|$, the total approximation error is bounded as follows,

\begin{small}
\begin{equation}
\begin{split}
&\left\|{\rm exp}\left(\bar{\textit{\textbf{Q}}}\right){\rm exp}\left(\bar{\textit{\textbf{K}}}\right)^{\top} - {\rm exp}\left(\bar{\textit{\textbf{Q}}}\bar{\textit{\textbf{K}}}^{\top}\right)\odot{\rm exp}\left(\tilde{\textit{\textbf{Q}}}\tilde{\textit{\textbf{K}}}^{\top}\right)\right\|_\xi\left\|\bar{\textit{\textbf{V}}}\right\|_\xi\\
&+\left\|{\rm exp}\left(\tilde{\textit{\textbf{Q}}}\tilde{\textit{\textbf{K}}}^{\top}\right)\odot\textit{\textbf{Z}} - {\rm exp}\left(\bar{\textit{\textbf{Q}}}\bar{\textit{\textbf{K}}}^{\top}\right)\odot{\rm exp}\left(\tilde{\textit{\textbf{Q}}}\tilde{\textit{\textbf{K}}}^{\top}\right)\right\|_\xi\left\|\tilde{\textit{\textbf{V}}}\right\|_\xi\\
& \leq\left(\frac{c_1}{\mathscr{M}{\rm exp}\left(-0.5\right)}\left\|\textit{\textbf{U}}_m - {\rm exp}\left(\tilde{\textit{\textbf{Q}}}\tilde{\textit{\textbf{K}}}^{\top}\right)\right\|_\xi + \left|\frac{c_1}{\mathscr{M}{\rm exp}\left(-0.5\right)}-1\right|\left\|{\rm exp}\left(\tilde{\textit{\textbf{Q}}}\tilde{\textit{\textbf{K}}}^{\top}\right)\right\|_\xi\right)\\&\left\|{\rm exp}\left(\bar{\textit{\textbf{Q}}}\bar{\textit{\textbf{K}}}^{\top}\right)\right\|_\xi\left\|\bar{\textit{\textbf{V}}}\right\|_\xi + \left\|\textit{\textbf{Z}} - {\rm exp}\left(\bar{\textit{\textbf{Q}}}\bar{\textit{\textbf{K}}}^{\top}\right)\right\|_\xi\left\|{\rm exp}\left(\tilde{\textit{\textbf{Q}}}\tilde{\textit{\textbf{K}}}^{\top}\right)\right\|_\xi\left\|\tilde{\textit{\textbf{V}}}\right\|_\xi.
\end{split}
\label{elfatterr3-1}
\end{equation}
\end{small}

\noindent If $\left|\frac{c_1}{\mathscr{M}{\rm exp}\left(-0.5\right)} - 1\right| < \left|\frac{c_1}{\mathfrak{M}{\rm exp}\left(-0.5\right)} - 1\right|$, the total approximation error is bounded as follows,

\begin{small}
\begin{equation}
\begin{split}
&\left\|{\rm exp}\left(\bar{\textit{\textbf{Q}}}\right){\rm exp}\left(\bar{\textit{\textbf{K}}}\right)^{\top} - {\rm exp}\left(\bar{\textit{\textbf{Q}}}\bar{\textit{\textbf{K}}}^{\top}\right)\odot{\rm exp}\left(\tilde{\textit{\textbf{Q}}}\tilde{\textit{\textbf{K}}}^{\top}\right)\right\|_\xi\left\|\bar{\textit{\textbf{V}}}\right\|_\xi \\
&+\left\|{\rm exp}\left(\tilde{\textit{\textbf{Q}}}\tilde{\textit{\textbf{K}}}^{\top}\right)\odot\textit{\textbf{Z}} - {\rm exp}\left(\bar{\textit{\textbf{Q}}}\bar{\textit{\textbf{K}}}^{\top}\right)\odot{\rm exp}\left(\tilde{\textit{\textbf{Q}}}\tilde{\textit{\textbf{K}}}^{\top}\right)\right\|_\xi\left\|\tilde{\textit{\textbf{V}}}\right\|_\xi \\
& \leq\left(\frac{c_1}{\mathscr{M}{\rm exp}\left(-0.5\right)}\left\|\textit{\textbf{U}}_m - {\rm exp}\left(\tilde{\textit{\textbf{Q}}}\tilde{\textit{\textbf{K}}}^{\top}\right)\right\|_\xi + \left|\frac{c_1}{\mathfrak{M}{\rm exp}\left(-0.5\right)}-1\right|\left\|{\rm exp}\left(\tilde{\textit{\textbf{Q}}}\tilde{\textit{\textbf{K}}}^{\top}\right)\right\|_\xi\right)\\
&\left\|{\rm exp}\left(\bar{\textit{\textbf{Q}}}\bar{\textit{\textbf{K}}}^{\top}\right)\right\|_\xi\left\|\bar{\textit{\textbf{V}}}\right\|_\xi + \left\|\textit{\textbf{Z}} - {\rm exp}\left(\bar{\textit{\textbf{Q}}}\bar{\textit{\textbf{K}}}^{\top}\right)\right\|_\xi\left\|{\rm exp}\left(\tilde{\textit{\textbf{Q}}}\tilde{\textit{\textbf{K}}}^{\top}\right)\right\|_\xi\left\|\tilde{\textit{\textbf{V}}}\right\|_\xi.
\end{split}
\label{elfatterr3-2}
\end{equation}
\end{small}

\noindent For the approximation error of Eq. (\ref{elfatt2}) of the main body of the paper, it is obvious to see that

\begin{small}
\begin{equation*}
\begin{split}
    &\left\lbrack{\rm exp}\left(\bar{\textit{\textbf{Q}}}\right){\rm exp}\left(\bar{\textit{\textbf{K}}}\right)^{\top}\bar{\textit{\textbf{V}}}, \left({\rm exp}\left(\tilde{\textit{\textbf{Q}}}\tilde{\textit{\textbf{K}}}^{\top}\right)\odot\textit{\textbf{Z}}\right)\tilde{\textit{\textbf{V}}}\right\rbrack - \left\lbrack{\rm exp}\left(\bar{\textit{\textbf{Q}}}\bar{\textit{\textbf{K}}}^{\top}\right)\bar{\textit{\textbf{V}}}, {\rm exp}\left(\tilde{\textit{\textbf{Q}}}\tilde{\textit{\textbf{K}}}^{\top}\right)\tilde{\textit{\textbf{V}}}\right\rbrack\\
    &=\left\lbrack\left({\rm exp}\left(\bar{\textit{\textbf{Q}}}\right){\rm exp}\left(\bar{\textit{\textbf{K}}}\right)^{\top} - {\rm exp}\left(\bar{\textit{\textbf{Q}}}\bar{\textit{\textbf{K}}}^{\top}\right)\right)\bar{\textit{\textbf{V}}}, \left({\rm exp}\left(\tilde{\textit{\textbf{Q}}}\tilde{\textit{\textbf{K}}}^{\top}\right)\odot\textit{\textbf{Z}} - {\rm exp}\left(\tilde{\textit{\textbf{Q}}}\tilde{\textit{\textbf{K}}}^{\top}\right)\right)\tilde{\textit{\textbf{V}}}\right\rbrack.
\end{split}
\end{equation*}
\end{small}

\noindent Hence, the corresponding approximation error is bounded as follows,

\begin{small}
\begin{equation}
\begin{split}
    &\left\|\left\lbrack{\rm exp}\left(\bar{\textit{\textbf{Q}}}\right){\rm exp}\left(\bar{\textit{\textbf{K}}}\right)^{\top}\bar{\textit{\textbf{V}}}, \left({\rm exp}\left(\tilde{\textit{\textbf{Q}}}\tilde{\textit{\textbf{K}}}^{\top}\right)\odot\textit{\textbf{Z}}\right)\tilde{\textit{\textbf{V}}}\right\rbrack - \left\lbrack{\rm exp}\left(\bar{\textit{\textbf{Q}}}\bar{\textit{\textbf{K}}}^{\top}\right)\bar{\textit{\textbf{V}}}, {\rm exp}\left(\tilde{\textit{\textbf{Q}}}\tilde{\textit{\textbf{K}}}^{\top}\right)\tilde{\textit{\textbf{V}}}\right\rbrack\right\|_\xi \\
    &\leq\left\|{\rm exp}\left(\bar{\textit{\textbf{Q}}}\right){\rm exp}\left(\bar{\textit{\textbf{K}}}\right)^{\top} - {\rm exp}\left(\bar{\textit{\textbf{Q}}}\bar{\textit{\textbf{K}}}^{\top}\right)\right\|_\xi\left\|\bar{\textit{\textbf{V}}}\right\|_\xi\\
    &+\left\|{\rm exp}\left(\tilde{\textit{\textbf{Q}}}\tilde{\textit{\textbf{K}}}^{\top}\right)\odot\textit{\textbf{Z}} - {\rm exp}\left(\tilde{\textit{\textbf{Q}}}\tilde{\textit{\textbf{K}}}^{\top}\right)\right\|_\xi\left\|\tilde{\textit{\textbf{V}}}\right\|_\xi.
\end{split}
\label{elfatterr4}
\end{equation}
\end{small}

\noindent If $\left|\frac{c_1}{\mathscr{M}{\rm exp}\left(-0.5\right)} - 1\right|\geq \left|\frac{c_1}{\mathfrak{M}{\rm exp}\left(-0.5\right)} - 1\right|$, we have

\begin{small}
\begin{equation}
\begin{split}
    &\left\|\left\lbrack{\rm exp}\left(\bar{\textit{\textbf{Q}}}\right){\rm exp}\left(\bar{\textit{\textbf{K}}}\right)^{\top}\bar{\textit{\textbf{V}}}, \left({\rm exp}\left(\tilde{\textit{\textbf{Q}}}\tilde{\textit{\textbf{K}}}^{\top}\right)\odot\textit{\textbf{Z}}\right)\tilde{\textit{\textbf{V}}}\right\rbrack - \left\lbrack{\rm exp}\left(\bar{\textit{\textbf{Q}}}\bar{\textit{\textbf{K}}}^{\top}\right)\bar{\textit{\textbf{V}}}, {\rm exp}\left(\tilde{\textit{\textbf{Q}}}\tilde{\textit{\textbf{K}}}^{\top}\right)\tilde{\textit{\textbf{V}}}\right\rbrack\right\|_\xi \\
    &\leq\left|\frac{c_1}{\mathscr{M}{\rm exp}\left(-0.5\right)}-1\right|\left\|{\rm exp}\left(\bar{\textit{\textbf{Q}}}\bar{\textit{\textbf{K}}}^{\top}\right)\right\|_\xi\left\|\bar{\textit{\textbf{V}}}\right\|_\xi +\left\|\textit{\textbf{Z}} - \textit{\textbf{U}}_m\right\|_\xi\left\|{\rm exp}\left(\tilde{\textit{\textbf{Q}}}\tilde{\textit{\textbf{K}}}^{\top}\right)\right\|_\xi\left\|\tilde{\textit{\textbf{V}}}\right\|_\xi.
\end{split}
\label{elfatterr4-1}
\end{equation}
\end{small}

\noindent If $\left|\frac{c_1}{\mathscr{M}{\rm exp}\left(-0.5\right)} - 1\right| < \left|\frac{c_1}{\mathfrak{M}{\rm exp}\left(-0.5\right)} - 1\right|$, one has

\begin{small}
\begin{equation}
\begin{split}
    &\left\|\left\lbrack{\rm exp}\left(\bar{\textit{\textbf{Q}}}\right){\rm exp}\left(\bar{\textit{\textbf{K}}}\right)^{\top}\bar{\textit{\textbf{V}}}, \left({\rm exp}\left(\tilde{\textit{\textbf{Q}}}\tilde{\textit{\textbf{K}}}^{\top}\right)\odot\textit{\textbf{Z}}\right)\tilde{\textit{\textbf{V}}}\right\rbrack - \left\lbrack{\rm exp}\left(\bar{\textit{\textbf{Q}}}\bar{\textit{\textbf{K}}}^{\top}\right)\bar{\textit{\textbf{V}}}, {\rm exp}\left(\tilde{\textit{\textbf{Q}}}\tilde{\textit{\textbf{K}}}^{\top}\right)\tilde{\textit{\textbf{V}}}\right\rbrack\right\|_\xi \\
    &\leq\left|\frac{c_1}{\mathfrak{M}{\rm exp}\left(-0.5\right)}-1\right|\left\|{\rm exp}\left(\bar{\textit{\textbf{Q}}}\bar{\textit{\textbf{K}}}^{\top}\right)\right\|_\xi\left\|\bar{\textit{\textbf{V}}}\right\|_\xi +\left\|\textit{\textbf{Z}} - \textit{\textbf{U}}_m\right\|_\xi\left\|{\rm exp}\left(\tilde{\textit{\textbf{Q}}}\tilde{\textit{\textbf{K}}}^{\top}\right)\right\|_\xi\left\|\tilde{\textit{\textbf{V}}}\right\|_\xi.
\end{split}
\label{elfatterr4-2}
\end{equation}
\end{small}

\noindent Inequalities (\ref{elfatterr4}-\ref{elfatterr4-2}) give a tighter bound than Inequalities (\ref{elfatterr2}-\ref{elfatterr3-2}). Fig. \ref{boundana} (a) shows the upper bounds of different attention mechanisms (EFFATT \cite{2SOFT}, ELFATT, and LOCAL \cite{CSWin}) for approximating attention matrices of vanilla attention during the training of ImageNet-1K. The backbone used is CSWin-T \cite{flatten}. The upper bound for ELFATT is obtained according to 

\begin{small}
\begin{equation*}
\begin{split}
&\left\lVert\left[ \mathrm{exp}(\bar{\textit{\textbf{Q}}})\mathrm{exp}(\bar{\textit{\textbf{K}}})^{\top},\ \left( \mathrm{exp}(\tilde{\textit{\textbf{Q}}} \tilde{\textit{\textbf{K}}}^{\top})\odot\textit{\textbf{Z}}\right)\right] - \left[ \mathrm{exp}(\bar{\textit{\textbf{Q}}}\bar{\textit{\textbf{K}}}^{\top}),\ \mathrm{exp}(\tilde{\textit{\textbf{Q}}}\tilde{\textit{\textbf{K}}}^{\top})\right] \right\rVert_{\xi}\\
&\leq
\left\lVert \mathrm{exp}(\bar{\textit{\textbf{Q}}})\mathrm{exp}(\bar{\textit{\textbf{K}}})^{\top}-\mathrm{exp}(\bar{\textit{\textbf{Q}}} \bar{\textit{\textbf{K}}}^{\top}) \right\rVert_{\xi}
+
\left\lVert\mathrm{exp}(\tilde{\textit{\textbf{Q}}}\tilde{\textit{\textbf{K}}}^{\top})\odot\textit{\textbf{Z}}-\mathrm{exp}(\tilde{\textit{\textbf{Q}}} \tilde{\textit{\textbf{K}}}^{\top})\right\rVert_{\xi},
\end{split}
\end{equation*}
\end{small}

\noindent for EFFATT is as 

\begin{small}
\begin{equation*}
\begin{split}
&\left\lVert\left[ \mathrm{exp}(\bar{\textit{\textbf{Q}}})\mathrm{exp}(\bar{\textit{\textbf{K}}})^{\top},\ \mathrm{exp}(\tilde{\textit{\textbf{Q}}}) \mathrm{exp}(\tilde{\textit{\textbf{K}}})^{\top}\right] - \left[ \mathrm{exp}(\bar{\textit{\textbf{Q}}}\bar{\textit{\textbf{K}}}^{\top}),\ \mathrm{exp}(\tilde{\textit{\textbf{Q}}}\tilde{\textit{\textbf{K}}}^{\top})\right] \right\rVert_{\xi}\\
&\leq
\left\lVert \mathrm{exp}(\bar{\textit{\textbf{Q}}})\mathrm{exp}(\bar{\textit{\textbf{K}}})^{\top}-\mathrm{exp}(\bar{\textit{\textbf{Q}}} \bar{\textit{\textbf{K}}}^{\top}) \right\rVert_{\xi}
+
\left\lVert\mathrm{exp}(\tilde{\textit{\textbf{Q}}})\mathrm{exp}(\tilde{\textit{\textbf{K}}})^{\top}-\mathrm{exp}(\tilde{\textit{\textbf{Q}}} \tilde{\textit{\textbf{K}}}^{\top})\right\rVert_{\xi},
\end{split}
\end{equation*}
\end{small}

\noindent and for LOCAL is as 

\begin{small}
\begin{equation*}
\begin{split}
&\left\lVert\left[\left(\mathrm{exp}(\bar{\textit{\textbf{Q}}} \bar{\textit{\textbf{K}}}^{\top})\odot\textit{\textbf{Z}}_1\right),\ \left(\mathrm{exp}(\tilde{\textit{\textbf{Q}}} \tilde{\textit{\textbf{K}}}^{\top})\odot\textit{\textbf{Z}}_2\right)\right] - \left[\mathrm{exp}(\bar{\textit{\textbf{Q}}}\bar{\textit{\textbf{K}}}^{\top}),\ \mathrm{exp}(\tilde{\textit{\textbf{Q}}}\tilde{\textit{\textbf{K}}}^{\top})\right] \right\rVert_{\xi}\\
&\leq
\left\lVert\mathrm{exp}(\bar{\textit{\textbf{Q}}} \bar{\textit{\textbf{K}}}^{\top})\odot\textit{\textbf{Z}}_1-\mathrm{exp}(\bar{\textit{\textbf{Q}}} \bar{\textit{\textbf{K}}}^{\top}) \right\rVert_{\xi}
+
\left\lVert\mathrm{exp}(\tilde{\textit{\textbf{Q}}} \tilde{\textit{\textbf{K}}}^{\top})\odot\textit{\textbf{Z}}_2-\mathrm{exp}(\tilde{\textit{\textbf{Q}}} \tilde{\textit{\textbf{K}}}^{\top})\right\rVert_{\xi},
\end{split}
\end{equation*}
\end{small}

\noindent where $\xi=F$ denotes the Frobenius norm, and $\textit{\textbf{Z}}_1$ and $\textit{\textbf{Z}}_2$ are different to $\textit{\textbf{Z}}$ of ELFATT, because LOCAL uses a more complex cross-shaped blockify method for different heads. All upper bounds were obtained from the final attention layer of the second level of the backbone. The upper bounds of different methods show a decreasing trend as the training process progresses. ELFATT further reduces the upper bound of the LOCAL attention mechanism, although it still has a higher bound than EFFATT, which is caused by its sparse blockify attention heads. Figs. \ref{boundana} (b) and (c) show the comparison of relative attention matrix approximation error ($\left\lVert\textit{\textbf{A}}-\textit{\textbf{A}}^{'}\right\rVert_{\xi}/\left\lVert\textit{\textbf{A}}\right\rVert_{\xi}$, where $\textit{\textbf{A}}$ denotes the attention matrix of VaniATT and $\textit{\textbf{A}}'$ denotes the attention matrix of linear attention) of global linear attention heads of ELFATT and the corresponding heads in EFFATT and relative attention matrix approximation error of sparse blockify attention heads of ELFATT and the corresponding heads in LOCAL. The backbone used is CSWin-T \cite{flatten}. Since ELFATT uses a hybrid head architecture, its half heads consist of global linear attention heads, and the remaining heads consist of sparse blockify attention heads. We compared the relative attention matrix approximation error of global linear attention heads of ELFATT and the heads in the same position of EFFATT. We compared the relative attention matrix approximation error of sparse blockify attention heads of ELFATT and the heads in the same position of LOCAL. We found that the global linear attention heads of ELFATT have a lower approximation error than the heads in the same position of EFFATT, and the sparse blockify attention heads of ELFATT have a lower approximation error than the heads in the same position of LOCAL. The hybrid head architecture has complementary effects on approximation error reduction.

\begin{figure}[htbp]
  \centering
    \subfloat[Upper Bounds]{\includegraphics[width=0.8\linewidth]{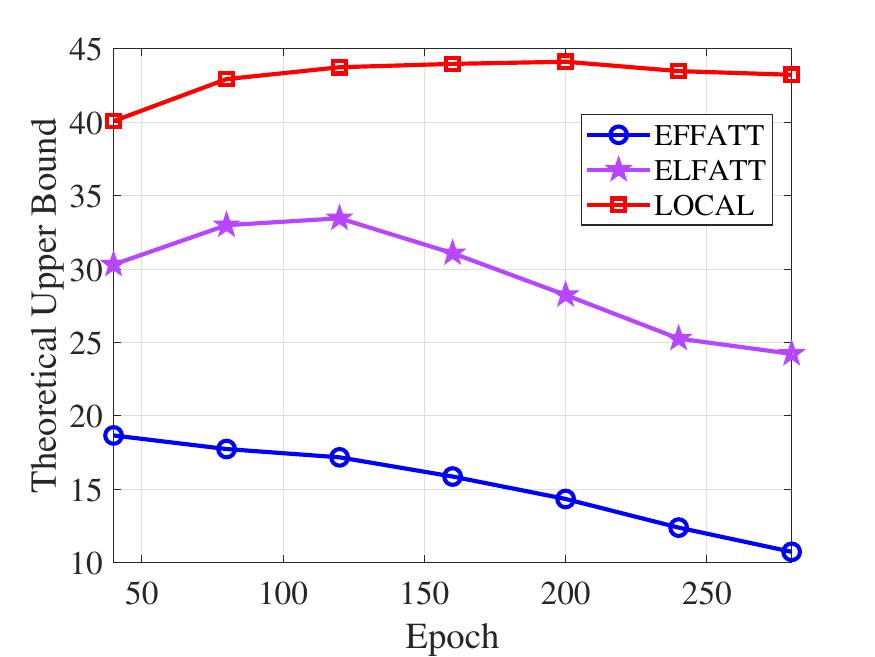}}
    \hfill
    \subfloat[Global Linear Attention Heads]{\includegraphics[width=0.8\linewidth]{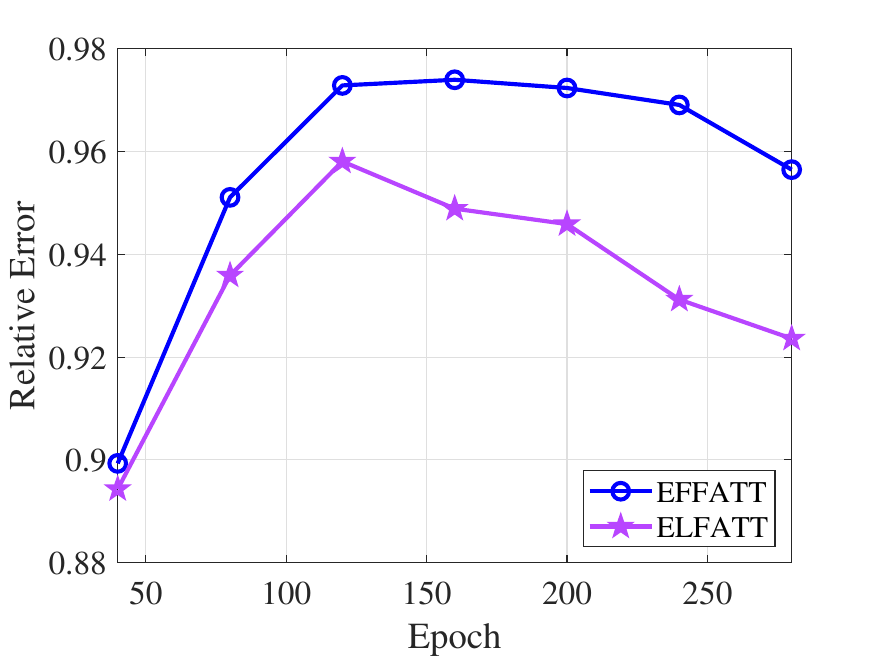}}
    \hfill
    \subfloat[Sparse Blockify Attention Heads]{\includegraphics[width=0.8\linewidth]{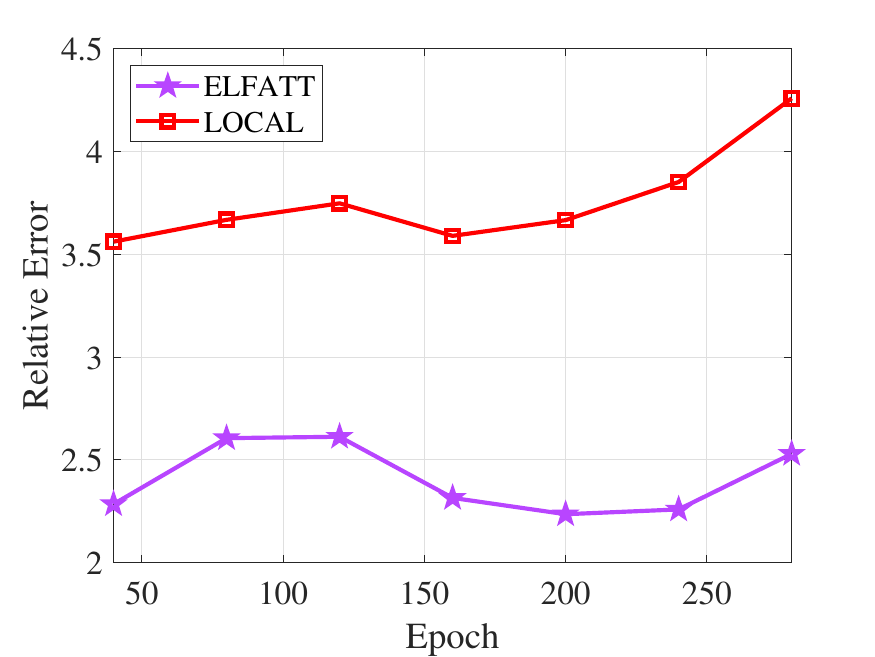}}
    \caption{(a) The upper bound of different attention mechanisms (EFFATT \cite{2SOFT}, ELFATT, and LOCAL \cite{CSWin}) for approximating attention matrices of vanilla attention during the training of ImageNet-1K. (b) The comparison of relative attention matrix approximation error of global linear attention heads of ELFATT and the corresponding heads in EFFATT. (c) The comparison of relative attention matrix approximation error of sparse blockify attention heads of ELFATT and the corresponding heads in LOCAL. The backbone used is CSWin-T.}
    \label{boundana}
\end{figure}

\begin{figure*}[htbp]
  \centering
    \subfloat[True Label: American Staffordshire Terrier]{\includegraphics[width=0.5\linewidth]{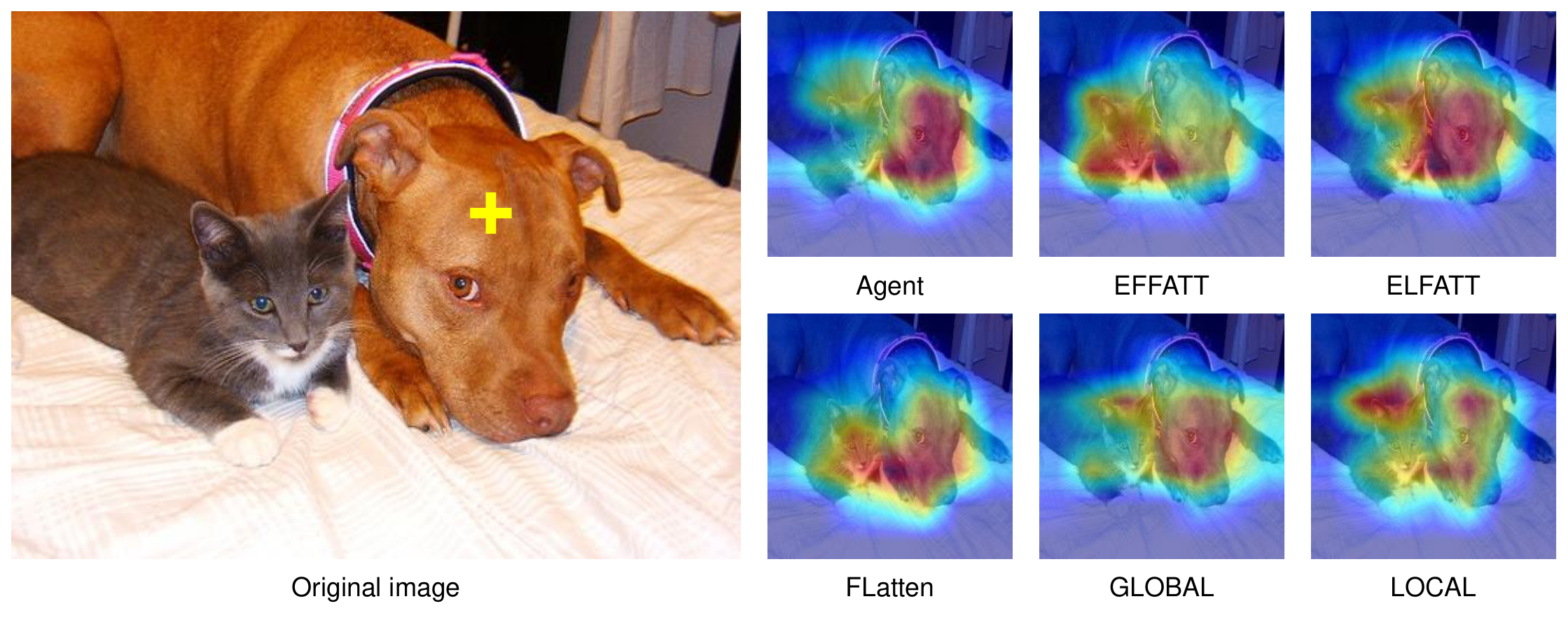}}
    \hfill
    \subfloat[True Label: Hartebeest]{\includegraphics[width=0.5\linewidth]{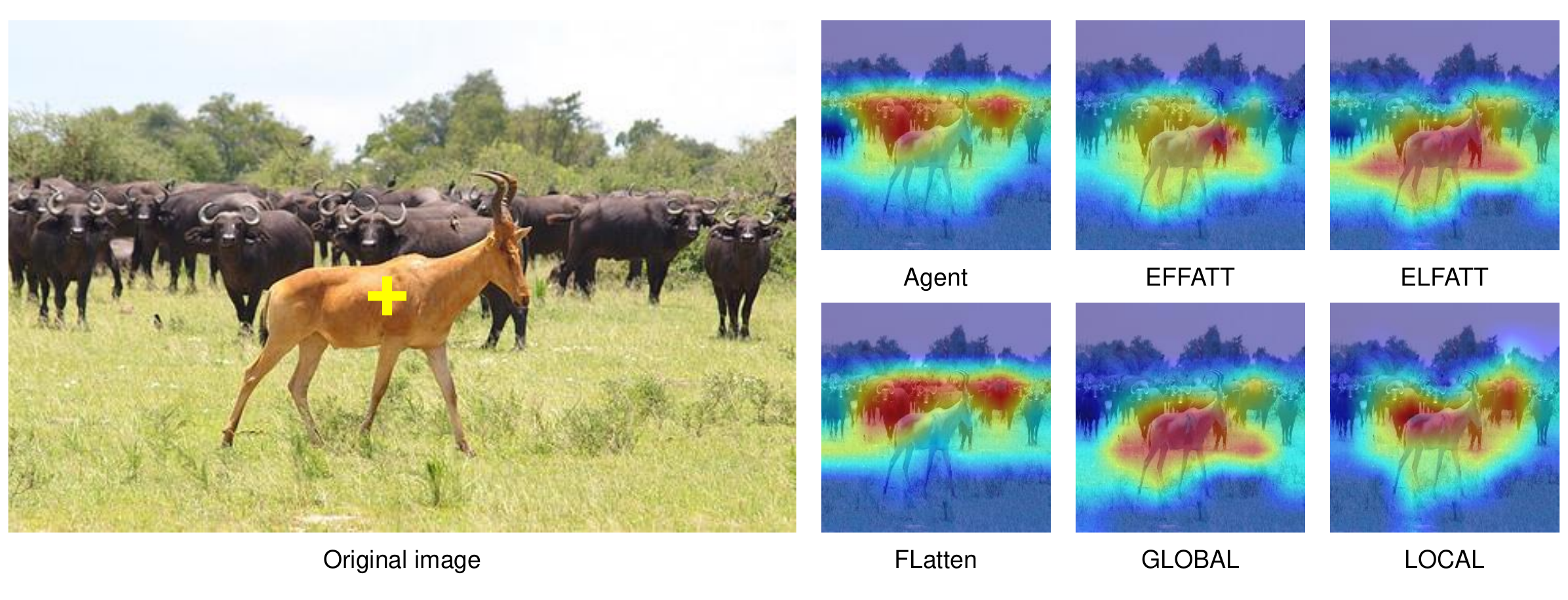}}
    \hfill
    \subfloat[True Label: Zebra]{\includegraphics[width=0.5\linewidth]{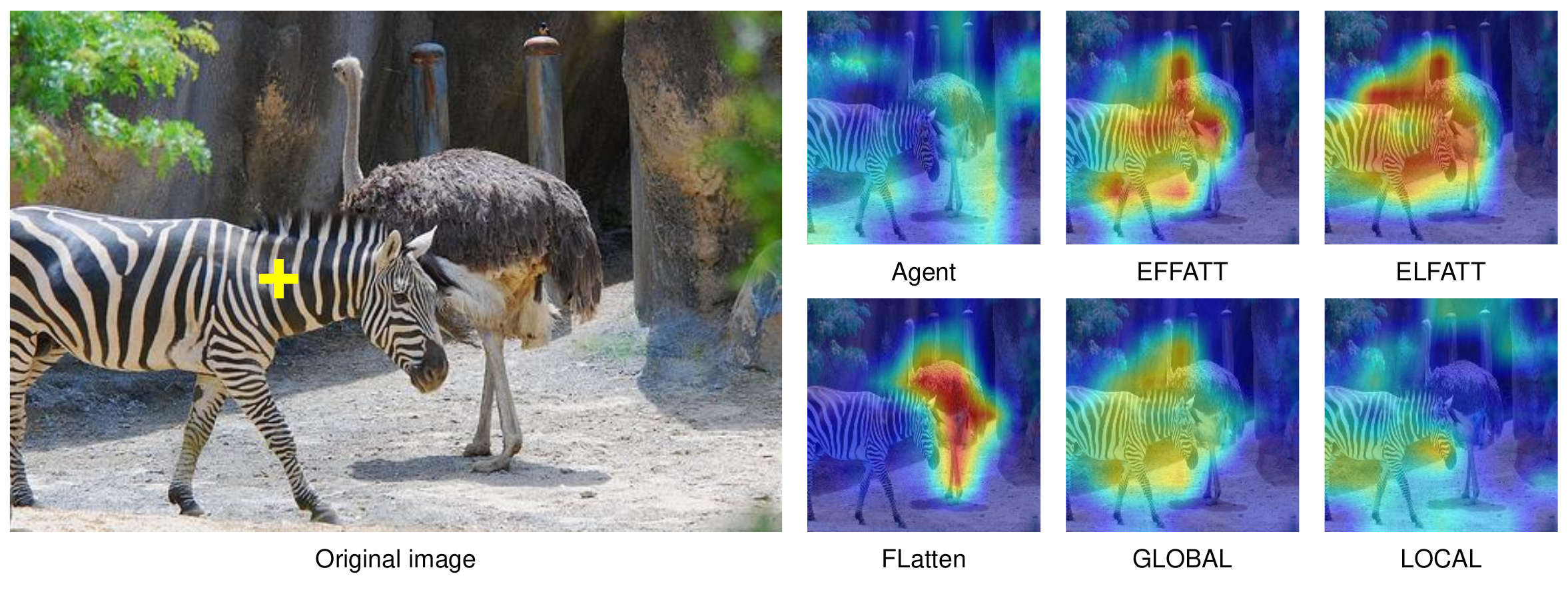}}
    \hfill
    \subfloat[True Label: Library]{\includegraphics[width=0.5\linewidth]{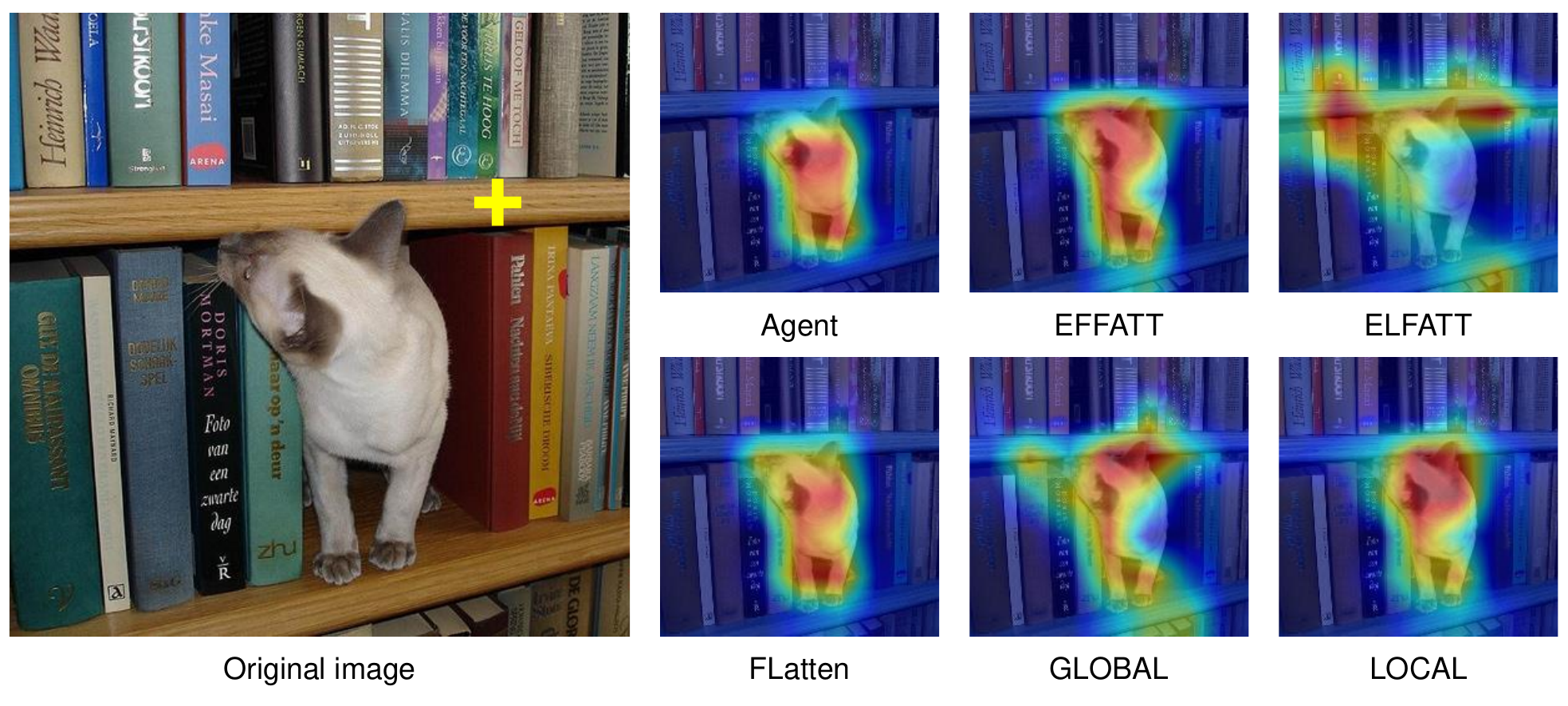}}
    \hfill
    \subfloat[True Label: Snail]{\includegraphics[width=0.5\linewidth]{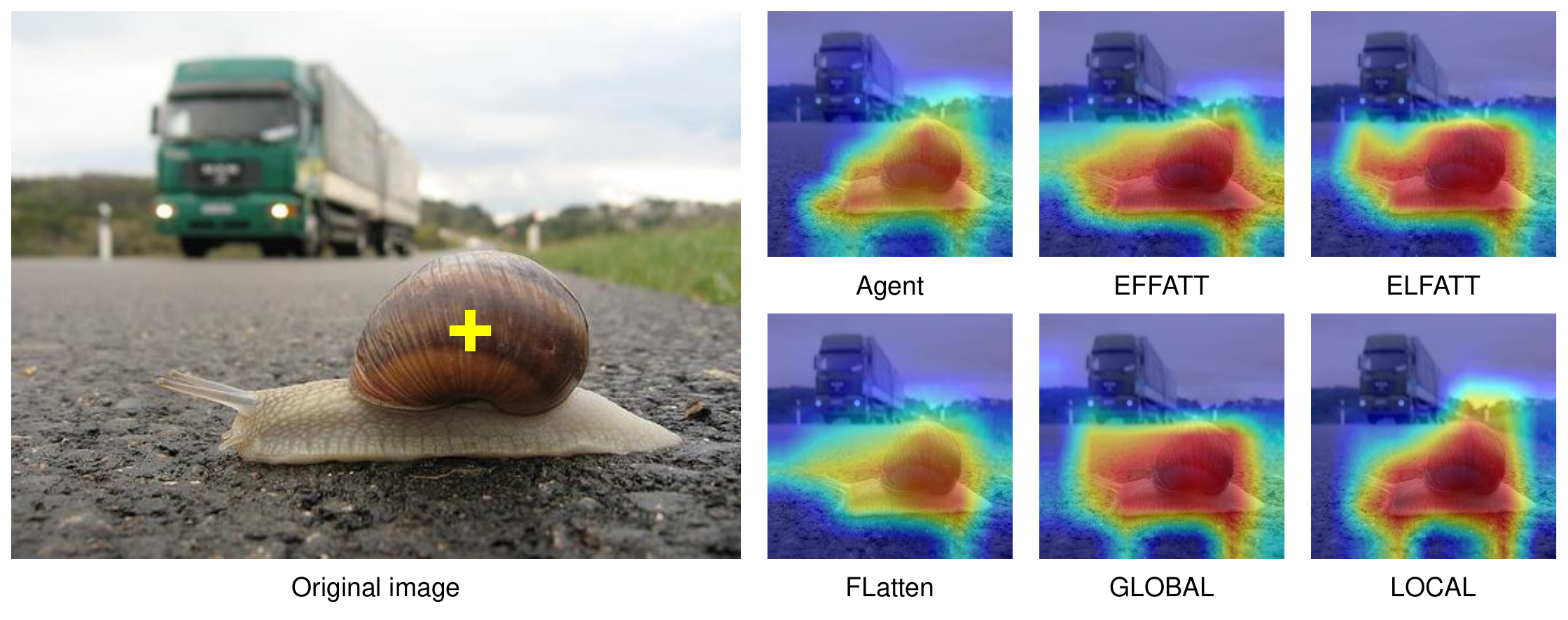}}
    \hfill
    \subfloat[True Label: Black-Footed Ferret]{\includegraphics[width=0.5\linewidth]{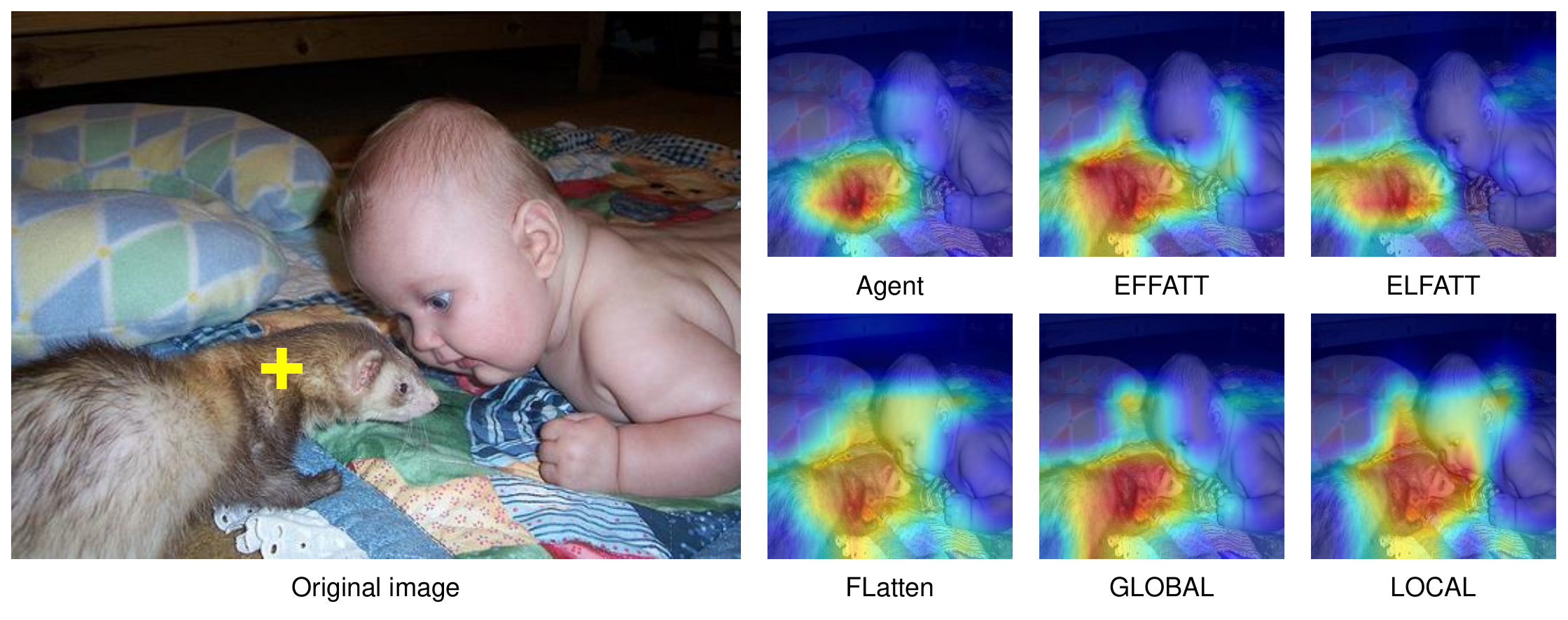}}
    \hfill
    \subfloat[True Label: Scottish Deerhound]{\includegraphics[width=0.5\linewidth]{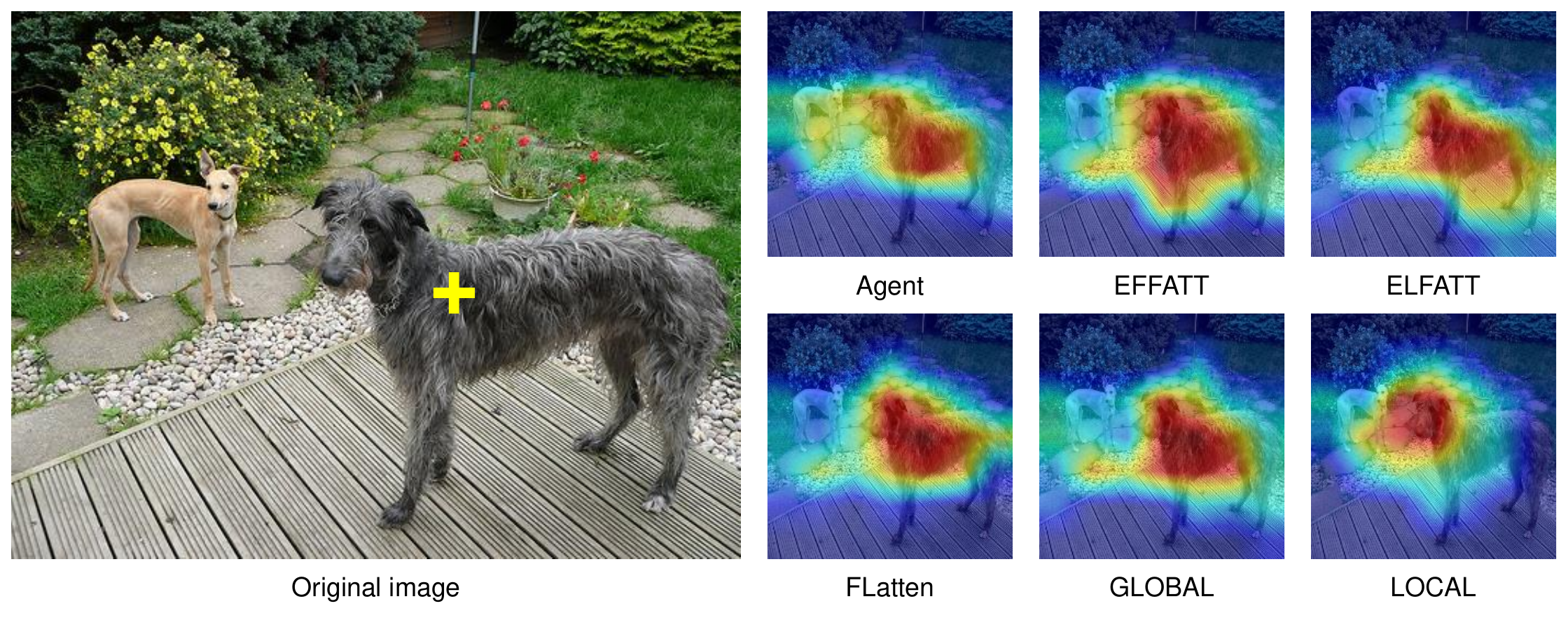}}
    \hfill
    \subfloat[True Label: Pineapple]{\includegraphics[width=0.5\linewidth]{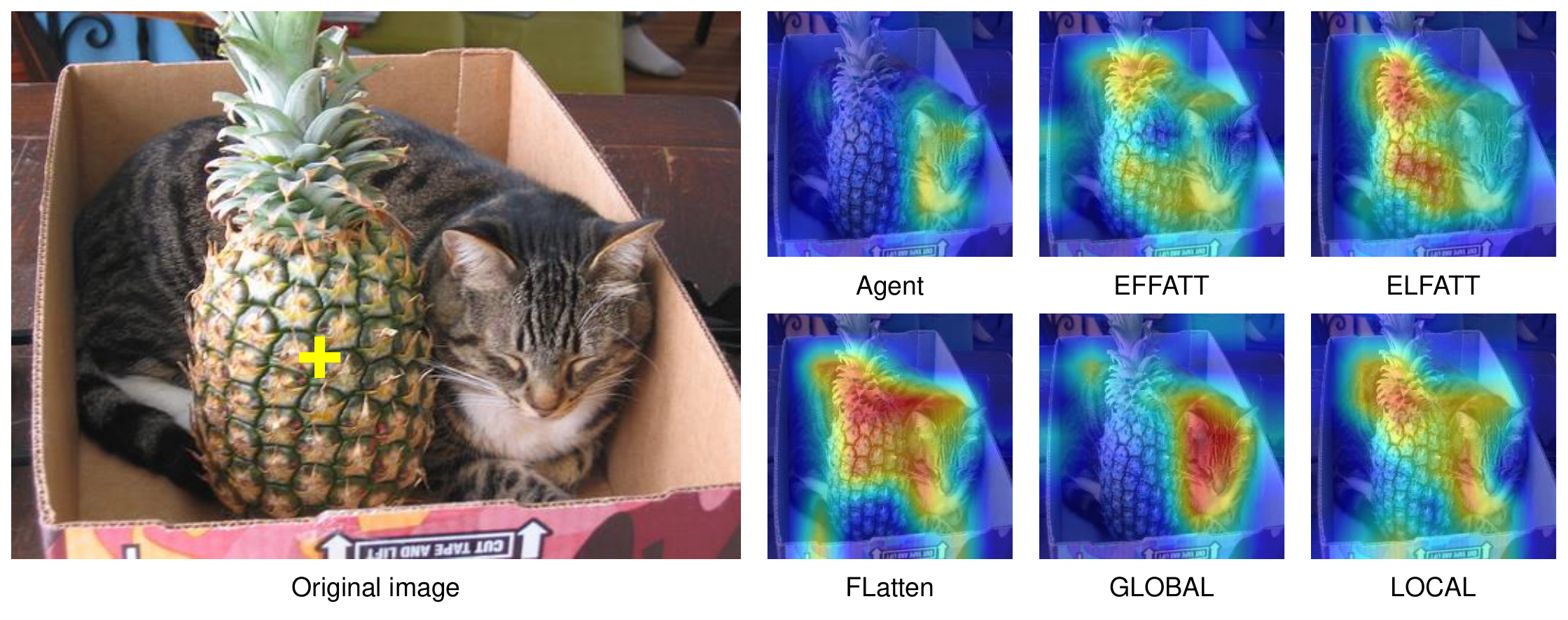}}
    \caption{Visual comparison of class activation map (CAM) based attention results of different attention mechanisms obtained using Score-CAM \cite{SCCAM}. Note: Backbone used is CSWin-T \cite{flatten}.}
    \label{visualcom2}
\end{figure*}

\begin{figure}[htbp]
  \centering
    \subfloat[CSWin]{\includegraphics[width=0.9\linewidth]{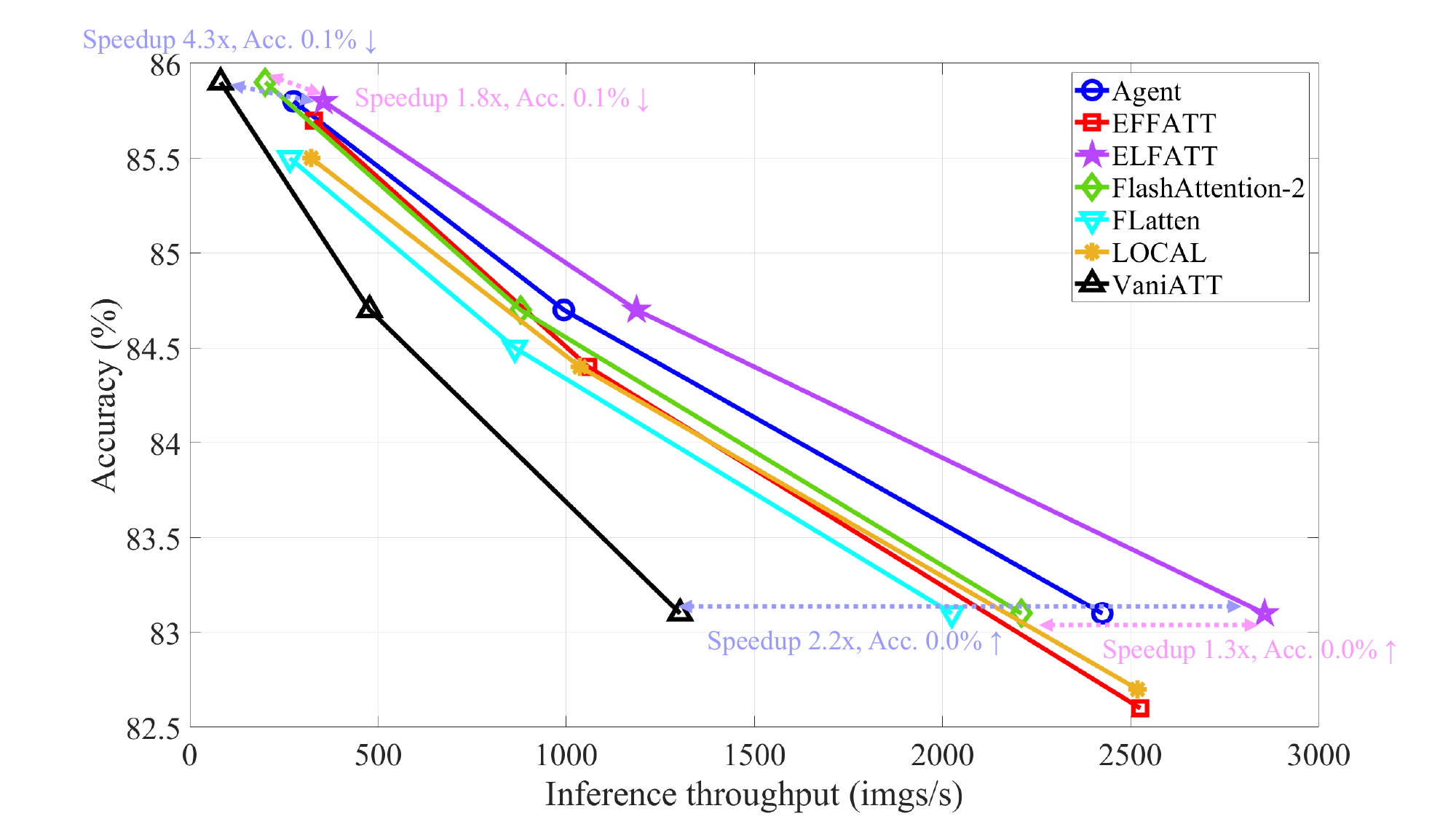}}
    \hfill
    \subfloat[Swin]{\includegraphics[width=0.9\linewidth]{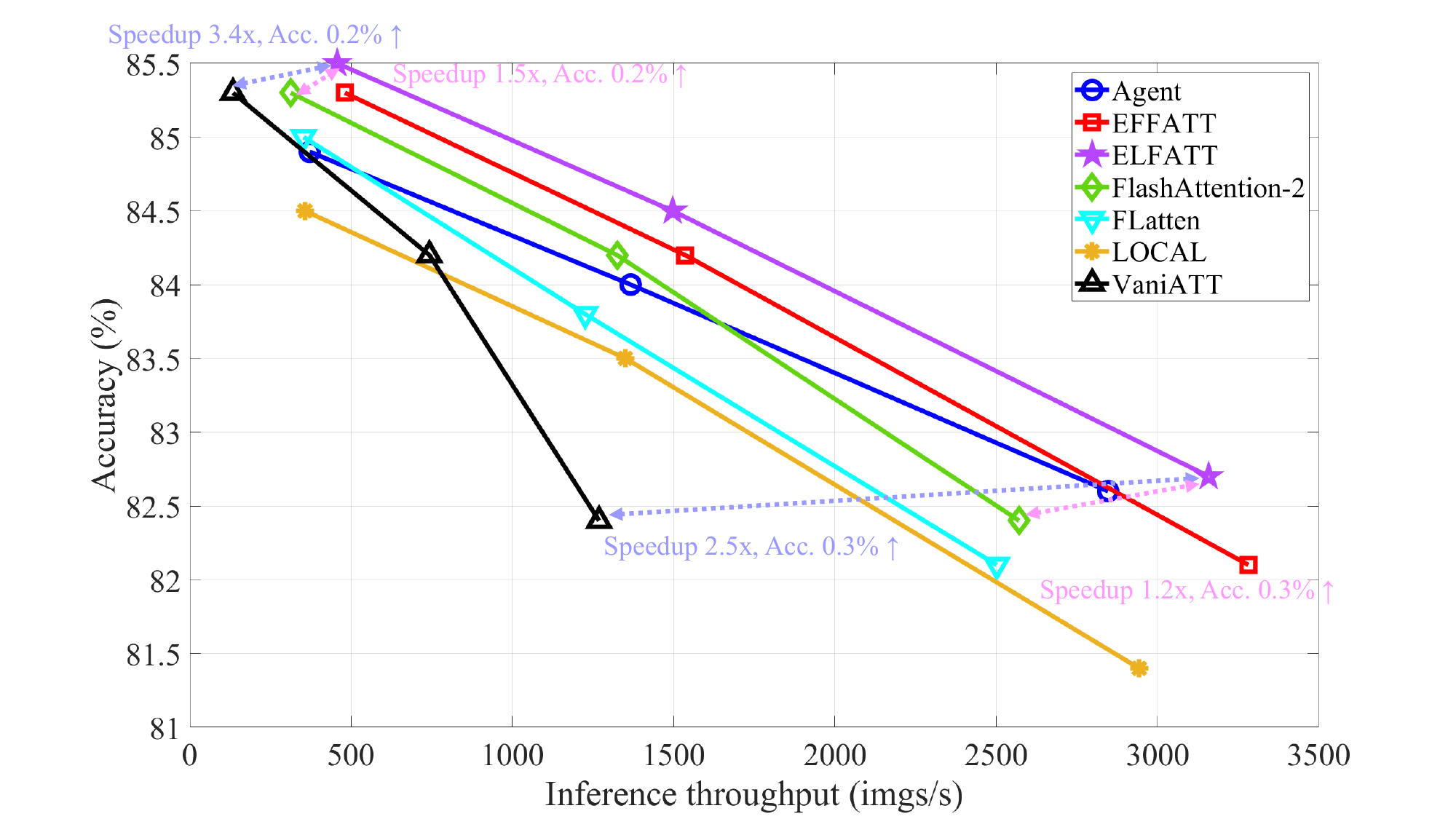}}
    \caption{Accuracy-efficiency curves obtained by different attention mechanisms on ImageNet-1K using different bacbones (CSWin \cite{flatten} and Swin \cite{Swin}). Efficiency is represented by the inference throughput of each method. All methods' inference throughput and accuracy results are from Table 1 of the main body of the paper.}
    \label{acceff}
\end{figure}

\section{Ablation Study of Different Block Sizes}
\label{APAB1}
To validate the effect of the block size on the performance of ELFATT, we performed an ablation study using different block sizes at each level of CSWin-T-ELFATT on ImageNet-1K. As shown in Table \ref{tabab1}, with the increasing number of block sizes, the inference speed of ELFATT will slow down, while the classification accuracy of ELFATT shows an increasing trend. ELFATT using 49-49-196-49 and 196-196-196-49 blocks achieve the best performance in terms of speed and accuracy. It seems that FlashAttention-2 is more efficient for some block sizes (with FlashAttention-2, the block size of 196 is faster than the block size of 49 which may be caused by GPU architectures). Hence, the block sizes for ELFATT can be determined according to the performance and efficiency requirements.

\section{Ablation Study of LePE}
\label{APAB2}
To validate the effect of LePE on performance, we compared the performance of ELFATT and VaniATT using or without using LePE in CSWin-T on ImageNet-1K. As shown in Table \ref{tabab2}, using LePE, the performance of VaniATT improves from 82.9 to 83.1 and the performance of ELFATT also improves from 82.9 to 83.1, respectively, further demonstrating the powerful local positional information enhancement of LePE. More details about LePE can be found in \cite{CSWin}.

\begin{table}[htbp]
  \centering
  \caption{The comparison of top-1 test accuracy (Acc.), inference throughput (FPS), parameter numbers (\#), and number of floating point operations (FLOPs) of ELFATT using different block sizes in each level (L1-L4) of CSWin-T on ImageNet-1K. Note: Inference throughput is obtained using a batch size of 512 with mixed precision on a single NVIDIA H20 (96 GB) GPU. The best values are in bold.}
    \scriptsize
    \begin{tabular}{cccc|rrrrrr}
    \toprule
        \multicolumn{4}{c|}{Block size of each level} & \multicolumn{1}{c}{\multirow{2}[0]{*}{Res.}} & \multicolumn{1}{c}{\multirow{2}[0]{*}{Acc.}} & \multicolumn{1}{c}{\multirow{2}[0]{*}{FPS (nFA/FA)}} & \multicolumn{1}{c}{\multirow{2}[0]{*}{\#}} & \multicolumn{1}{c}{\multirow{2}[0]{*}{FLOPs (nFA/FA)}} \\
    \multicolumn{1}{l}{L1} & \multicolumn{1}{l}{L2} & \multicolumn{1}{l}{L3} & \multicolumn{1}{l|}{L4} &       &       &       &       &  \\
    \midrule
    49  &  49  &  49   &  49                             &$224^2$& 82.5                                              &   2619/2805 imgs/s         &  20M     &  4.37G/4.13G      \\
    49  &  49  &  196  &  49                             &$224^2$& \textbf{83.1}                                     &   2603/2856 imgs/s         &  20M     &  4.44G/4.13G      \\
    49  &  196 &  196  &  49                             &$224^2$& 82.8                                              &   2552/2863 imgs/s         &  20M     &  4.50G/4.13G      \\
    196 &  196 &  196  &  49                             &$224^2$& \textbf{83.1}                                     &   2512/2881 imgs/s         &  20M     &  4.56G/4.13G      \\
    \bottomrule
    \end{tabular}
  \label{tabab1}
\end{table}

\begin{table}[htbp]
  \centering
  \caption{The comparison of top-1 test accuracy (Acc.), inference throughput (FPS), parameter numbers (\#), and number of floating point operations (FLOPs) of ELFATT and VaniATT using or without using LePE in CSWin-T on ImageNet-1K. Note: Inference throughput is obtained using a batch size of 512 with mixed precision on a single NVIDIA H20 (96 GB) GPU. The best values are in bold. ``Res.'' denotes resolution, ``imgs'' denotes images, ``w/o'' denotes ``without'', ``nFA'' denotes non-FlashAttention-2, and ``FA'' denotes FlashAttention-2.}
    \scriptsize
    \begin{tabular}{lrrrrrr}
    \toprule
    Method & \multicolumn{1}{r}{Res.} & \multicolumn{1}{r}{Acc. (\%)}   & \multicolumn{1}{r}{FPS (nFA/FA)} & \multicolumn{1}{r}{\#} & \multicolumn{1}{r} {FLOPs (nFA/FA)}             \\
    \midrule
    \multicolumn{6}{c}{CSWin-T-ELFATT}\\
    \midrule
    LePE                                  &$224^2$& \textbf{83.1}                                     &   2512/2881 imgs/s         &  20M     &  4.56G/4.13G      \\
    w/o-LePE                              &$224^2$& 82.9                                              &   2801/3271 imgs/s         &  20M     &  4.54G/4.12G      \\
    \midrule
    \multicolumn{6}{c}{CSWin-T-GLOBAL}\\
    \midrule
    LePE                                  &$224^2$& \textbf{83.1}                                     &   1303/2210 imgs/s         &  20M     &  7.60G/4.09G      \\
    w/o-LePE                              &$224^2$& 82.9                                              &   1332/2296 imgs/s         &  20M     &  7.58G/4.08G      \\
    \bottomrule
    \end{tabular}
  \label{tabab2}
\end{table}

\begin{table}[htbp]
  \centering
  \caption{The comparison of top-1 test accuracy (Acc.), inference throughput (FPS), parameter numbers (\#), and number of floating point operations (FLOPs) obtained by using the ELFATT modules to replace the VaniATT modules in different levels (L1-L4) of CSWin-T-GLOBAL on ImageNet-1K. Note: Inference throughput is obtained using a batch size of 512 with mixed precision on a single NVIDIA H20 (96 GB) GPU. The best values are in bold. ``Res.'' denotes resolution, ``imgs'' denotes images, ``nFA'' denotes non-FlashAttention-2, and ``FA'' denotes FlashAttention-2. \checkmark denotes the level consists of full ELFATT modules and \ding{52} denotes half of this level is composed of ELFATT modules and the remaining is composed of VaniATT modules. The order of composition is: VaniATT-ELFATT-$...$-VaniATT-ELFATT.}
    \scriptsize
    \begin{tabular}{cccc|rrrrr}
    \toprule
    \multicolumn{4}{c|}{Levels using ELFATT} & \multicolumn{1}{c}{\multirow{2}[0]{*}{Res.}} & \multicolumn{1}{c}{\multirow{2}[0]{*}{Acc.}} & \multicolumn{1}{c}{\multirow{2}[0]{*}{FPS (nFA/FA)}} & \multicolumn{1}{c}{\multirow{2}[0]{*}{\#}} & \multicolumn{1}{c}{\multirow{2}[0]{*}{FLOPs (nFA/FA)}} \\
    \multicolumn{1}{l}{L1} & \multicolumn{1}{l}{L2} & \multicolumn{1}{l}{L3} & \multicolumn{1}{l|}{L4} &       &       &       &       &  \\
    \midrule
        \checkmark  &                 &                     &                  &$224^2$                   & 83.0              &   2257/2865 imgs/s         &  20M     &  5.17G/4.10G       \\ 
        \checkmark  &   \checkmark    &                     &                  &$224^2$                   & 83.0              &   2467/2918 imgs/s         &  20M     &  4.63G/4.12G       \\ 
        \checkmark  &   \checkmark    &   \checkmark        &                  &$224^2$                   & 82.6              &   2555/2842 imgs/s         &  20M     &  4.48G/4.15G       \\ 
        \checkmark  &   \checkmark    &   \ding{52}         &                  &$224^2$                   & \textbf{83.1}     &   2512/2881 imgs/s         &  20M     &  4.56G/4.13G       \\ 
        \checkmark  &   \checkmark    &   \checkmark        &   \checkmark     &$224^2$                   & 82.6              &   2552/2835 imgs/s         &  20M     &  4.48G/4.15G       \\ 
     \midrule    
    \multicolumn{4}{c|}{CSWin-T-GLOBAL}                                   &$224^2$                   & \textbf{83.1}     &   1303/2210 imgs/s         &  20M     &  7.60G/4.09G       \\
    \bottomrule
    \end{tabular}
  \label{tabab3}
\end{table}

\begin{table}[htbp]
  \centering
  \caption{The comparison of top-1 test accuracy (Acc.), inference throughput (FPS), parameter numbers (\#), and number of floating point operations (FLOPs) obtained by CSWin-T-ELFATT using different combinations of $c_1$ and $c_2$ on ImageNet-1K. Note: Inference throughput is obtained using a batch size of 512 with mixed precision on a single NVIDIA H20 (96 GB) GPU. The best values are in bold.}
    \scriptsize
    \begin{tabular}{cc|rrrrr}
    \toprule
    \multicolumn{2}{c|}{Combinations of $c_1$ and $c_2$} & \multicolumn{1}{c}{\multirow{2}[0]{*}{Res.}} & \multicolumn{1}{c}{\multirow{2}[0]{*}{Acc.}} & \multicolumn{1}{c}{\multirow{2}[0]{*}{FPS (nFA/FA)}} & \multicolumn{1}{c}{\multirow{2}[0]{*}{\#}} & \multicolumn{1}{c}{\multirow{2}[0]{*}{FLOPs (nFA/FA)}} \\
    \multicolumn{1}{c}{\quad$c_1$} & \multicolumn{1}{c|}{$c_2$}                 &       &                   &                            &          &                    \\
    \midrule
        \quad$0\times c$    &   $1\times c$                                         &$224^2$& 82.7              &   2331/2880 imgs/s         &  20M     &  4.76G/4.09G       \\ 
        \quad$0.25\times$ c &   $0.75\times c$                                      &$224^2$& 83.0              &   2530/2835 imgs/s         &  20M     &  4.65G/4.10G       \\ 
        \quad$0.5\times$ c  &   $0.5\times c$                                       &$224^2$& \textbf{83.1}     &   2512/2881 imgs/s         &  20M     &  4.56G/4.13G       \\ 
        \quad$0.75\times$ c &   $0.25\times c$                                      &$224^2$& \textbf{83.1}     &   2458/2796 imgs/s         &  20M     &  4.48G/4.18G       \\     
        \quad$1\times$ c    &   $0\times c$                                         &$224^2$& 82.6              &   2394/2526 imgs/s         &  20M     &  4.35G/4.17G       \\
    \bottomrule
    \end{tabular}
  \label{tabab4}
\end{table}

\section{Ablation Study of Number of Levels Using the ELFATT Module}
\label{APAB3}
To validate the effect of the number of ELFATT modules used to replace the VaniATT modules at different levels of vision transformers, we compared the performance of CSWin-T using a different number of ELFATT modules at different levels. As shown in Table \ref{tabab3}, with the number of levels using ELFATT modules increasing, the inference speed increases without using FlashAttention-2. The difference between CSWin-T using 3 levels of ELFATT modules and 4 levels of ELFATT modules is not significant. The fourth level of CSWin-T is composed of only one module and the sequence length is only 49 which is too short to achieve a significant acceleration effect. Another thing that can be found in Table \ref{tabab3} is that with an increase in the number of levels using ELFATT modules, CSWin-T using FlashAttention-2 achieves the fastest inference speed when using two levels of ELFATT modules. Because the sequence lengths of the first two levels are 3136 and 784, respectively, which are significantly longer than the last two levels (the $3^{\rm rd}$ level: 196, and the $4^{\rm th}$ level: 49). That is also the reason why some efficient attention mechanisms, such as Agent \cite{AGENTATT} and FLatten \cite{flatten}, only replace some of levels by their efficient attention modules. We also observed that with an increasing number of levels using ELFATT modules, the performance shows a gradual decline. To address this defect, we introduced a hybrid architecture in the third level which replaces half of the ELFATT modules at this level by VaniATT modules. The order of composition is as follows: VaniATT-ELFATT-$...$-VaniATT-ELFATT. The reason is that the third level of the vision transformer is usually much deeper than other levels and the sequence of this level is also much shorter than the first two levels. VaniATT at this level will not affect the speed too much and can help the model to converge faster. Swin-T-ELFATT, Swin-B-ELFATT, and CSWin-B-ELFATT also use a pure ELFATT architecture in the first two levels and a hybrid architecture in the third level. All variants of CSWin-T using ELFATT modules to replace VaniATT modules are faster than CSWin-T-GLOBAL.

\section{Ablation Study of Different $c_1$ and $c_2$}
\label{APAB4}
Table \ref{tabab4} shows the performance comparison of CSWin-T-ELFATT using different combinations of $c_1$ and $c_2$ on ImageNet-1K. When $c_1 = 1\times c$ and $c_2 = 0\times c$, ELFATT becomes EFFATT, and when $c_1 = 0\times c$ and $c_2 = 1\times c$, ELFATT becomes a local window-based attention mechanism which can be regarded as a simpler version of the local window-based attention mechanism used in Swin. As shown in Table \ref{tabab4}, $c_1 = 0.5\times c$ and $c_2 = 0.5\times c$, and $c_1 = 0.75\times c$ and $c_2 = 0.25\times c$ achieve better performance than other combinations of $c_1$ and $c_2$. Without using FlashAttention-2, the difference between $c_1 = 0.5\times c$ and $c_2 = 0.5\times c$, and $c_1 = 0.75\times c$ and $c_2 = 0.25\times c$ in terms of speed is not significant. Using FlashAttention-2, $c_1 = 0\times c$ and $c_2 = 1\times c$, and $c_1 = 0.5\times c$ and $c_2 = 0.5\times c$ achieve close speed and are significantly faster than other combinations.

\section{Ablation Study of Different Patch Partition/Merging Methods}
\label{APAB5}
Table \ref{tabab5} shows the performance comparison of Swin-T-ELFATT using different patch partition/merging methods on ImageNet-1K. ``concat'' \cite{Swin} denotes the concatenation-based patch partition/merging method of Swin. ``conv'' \cite{CSWin} denotes the convolution-based patch partition/merging method of CSWin. It can be seen that ``conv'' significantly outperforms ``concat'' in terms of speed and accuracy. Hence, in this paper, all ELFATT variants use ``conv'' for patch partition/merging.

\section{Speed Comparison on Edge GPUs: Experiment Settings}
\label{APEGPU}
In addition to high-performance computing applications, model deployment should also be evaluated in emerging edge scenarios such as robotic vision, unmanned aerial vehicles (UAVs), and autonomous driving. These scenarios typically operate under strict power constraints while striving to maximize performance and efficiency. The NVIDIA Jetson series provide a good embedded and edge computing platform to evaluate AI model performance in such resource-constrained environments.

\begin{table}[htbp]
  \centering
  \caption{The comparison of top-1 test accuracy (Acc.), inference throughput (FPS), parameter numbers (\#), and number of floating point operations (FLOPs) of ELFATT using different patch partition/merging methods in Swin-T on ImageNet-1K. Note: Inference throughput is obtained using a batch size of 512 with mixed precision on 1 NVIDIA H20 (96 GB) GPU.}
    \scriptsize
    \begin{tabular}{lrrrrrr}
    \toprule
    Method & \multicolumn{1}{r}{Res.} & \multicolumn{1}{r}{Acc. (\%)}   & \multicolumn{1}{r}{FPS (nFA/FA)} & \multicolumn{1}{r}{\#} & \multicolumn{1}{r} {FLOPs (nFA/FA)}             \\
    \midrule
    concat &$224^2$& 82.3                                     &   2027/2769 imgs/s         &  28M     &  6.81G/4.42G      \\
    conv   &$224^2$& 82.7                                     &   2884/3159 imgs/s         &  30M     &  4.99G/4.67G      \\
    \bottomrule
    \end{tabular}
  \label{tabab5}
\end{table}

\subsection{Experiment Settings}
The experiments were conducted to evaluate the inference speed of the model on NVIDIA Jetson platforms, specifically Jetson Nano and Jetson AGX Orin, across power modes ranging from 5W to 60W. Each model was evaluated in full precision (FP32) and mixed precision modes. The evaluations were performed on the ImageNet-1K dataset with batch sizes of 1 for Jetson Nano and 128 for Jetson AGX Orin. Both devices installed the latest NVIDIA JetPack SDK (The Jetson Nano used JetPack SDK 4.6.6 and the Jetson AGX Orin utilized JetPack SDK 6.1). Each experiment was repeated 100 times and the reported results represent the average values to ensure statistical reliability.

\begin{figure*}[t!]
\centering
\includegraphics[width=0.9\linewidth]{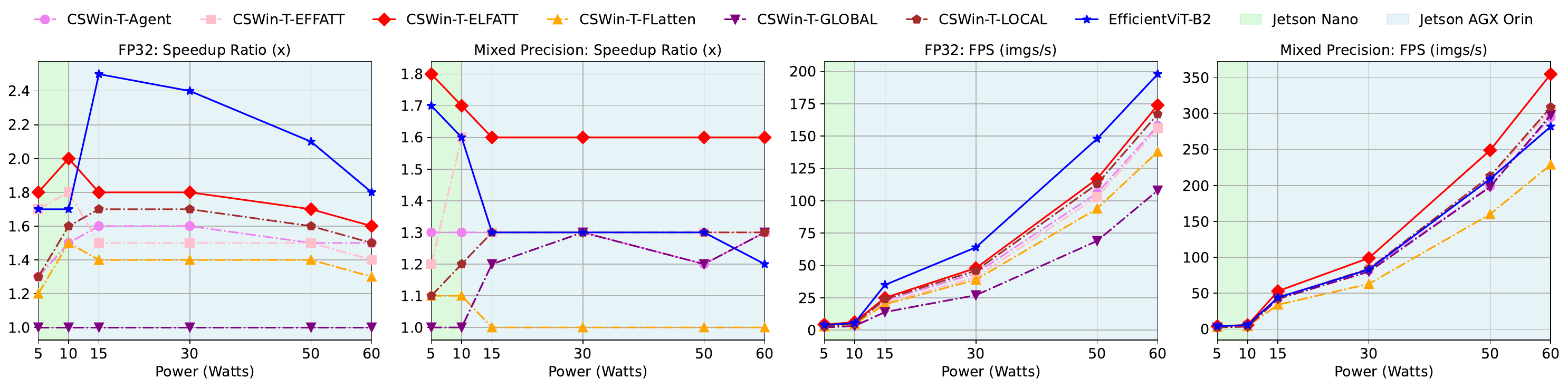}
\caption{The speed comparison of different attention mechanisms across power modes from 5W to 60W on NVIDIA Jetson Nano and AGX Orin under power-constrained inference scenarios.}
\label{figedgespeed}
\end{figure*}

\section{Stable Diffusion Acceleration Using ELFATT: Experiment Settings and Ablation Studies}
\label{APSD}
We also introduced ELFATT to the acceleration of stable diffusion (SD) \cite{SD}. The original EFFATT-based global linear attention head and LePE of ELFATT cannot be used without training. We used Agent \cite{AGENTATT} to replace the original EFFATT-based global linear attention head. The reason is that: (1) Agent is an improved version of EFFATT, and it only introduces a small agent matrix as the auxiliary key or query matrix to be multiplied with the full query or full key matrix to reduce its dimenison before the softmax normalization; (2) Agent can be used to replace the self-attention block of SD without training. Hence, Agent can be introduced to replace the original EFFATT-based global linear attention head of ELFATT. LePE is removed. Finally, in the self-attention block of SD, Eq. (\ref{lepe1}) of the main body of the paper will become as follows,

\begin{small}
\begin{equation}
\begin{split}
{\rm exp}\left(\textit{\textbf{Q}}\textit{\textbf{K}}^{\top}\right)\textit{\textbf{V}} \approx &\left\lbrack{\rm exp}\left(\bar{\textit{\textbf{Q}}}\bar{\mathcal{A}}^{\top}\right){\rm exp}\left(\bar{\mathcal{A}}\bar{\textit{\textbf{K}}}^{\top}\right)\bar{\textit{\textbf{V}}}, g\left({\rm exp}\left(f(\tilde{\textit{\textbf{Q}}}) f(\tilde{\textit{\textbf{K}}})^{\top}\right)f(\tilde{\textit{\textbf{V}}})\right)\right\rbrack,
\end{split}
\label{lepe1sd}
\end{equation}
\end{small}

\noindent Eq. (\ref{lepe2}) of the main body of the paper will become as follows,

\begin{small}
\begin{equation}
\begin{split}
\left\lbrack{\rm exp}\left(\bar{\textit{\textbf{Q}}}\bar{\textit{\textbf{K}}}^{\top}\right)\bar{\textit{\textbf{V}}}, {\rm exp}\left(\tilde{\textit{\textbf{Q}}}\tilde{\textit{\textbf{K}}}^{\top}\right)\tilde{\textit{\textbf{V}}}\right\rbrack\approx&\left\lbrack{\rm exp}\left(\bar{\textit{\textbf{Q}}}\bar{\mathcal{A}}^{\top}\right){\rm exp}\left(\bar{\mathcal{A}}\bar{\textit{\textbf{K}}}^{\top}\right)\bar{\textit{\textbf{V}}},\right.\\
&\left.g\left({\rm exp}\left(f(\tilde{\textit{\textbf{Q}}}) f(\tilde{\textit{\textbf{K}}})^{\top}\right)f(\tilde{\textit{\textbf{V}}})\right)\right\rbrack,
\end{split}
\label{lepe2sd}
\end{equation}
\end{small}

\noindent where $\bar{\mathcal{A}} = {\rm Pooling}\left(\bar{\textit{\textbf{Q}}}\right) \in \mathbb{R}^{m_1 \times c_1}$, ${\rm Pooling}(\cdot)$ denotes the average pooling function, and $m_1 \ll m$.

\begin{table}[htbp]
\scriptsize
  \centering
  \caption{The comparison of stable diffusion (SD) v1.5 (Diffusers v0.32.1) using ELFATT with different block sizes in FP16 precision where merging ratio is 0.5. Memory consumption for each variant indicates how much GPU memory increases when the batch size is raised by 1. All variants are accelerated by FlashAttention-2. The processing time of each variant is obtained using a batch size of 48 on 1 NVIDIA Tesla A100 (40 GB) GPU.}
    \begin{tabular}{rrrrr}
    \toprule
    \multicolumn{1}{r}{Block size} & \multicolumn{1}{r}{FID $\downarrow$}  & \multicolumn{1}{r}{Time (s/img)} & \multicolumn{1}{r}{Memory (GB/img)}  \\
    \midrule
    32&  30.8       &  0.710            &   0.65        \\
    56&  30.7       &  0.709            &   0.65        \\
    81&  30.5       &  0.708            &   0.65        \\
    128& 30.4       &  0.707            &   0.70        \\
    227& 30.5       &  0.705            &   0.70        \\
    512& 30.5       &  0.709            &   0.70        \\
    \bottomrule
    \end{tabular}
  \label{tabSDa1}
\end{table}

\subsection{Experiment Settings}
The experiments were carried out to assess the performance and acceleration effect of ELFATT in a stable diffusion task using 1 NVIDIA Tesla A100 (40 GB) GPU. Following the experiment settings of \cite{AGENTATT} and \cite{ToMe}, we applied ELFATT (Eqs. (\ref{lepe1sd}) and (\ref{lepe2sd})) on stable diffusion v1.5 using the pipeline of Diffusers v0.32.1 and called the proposed stable diffusion acceleration method ELFATT-SD. We compared ELFATT with Agent (Agent-SD) \cite{AGENTATT} and ToMe (ToMe-SD) \cite{ToMe}. We used ImageNet-1K labels as prompts to generate 2000 images with a resolution of $512^2$ (2 images per class). The inference steps (50) \cite{INSTEPS} and the cfg scale (7.5) \cite{CFG} remain the same as Agent and ToMe. Similar to Agent-SD, ELFATT-SD is also developed based on the token merging method ToMe-SD. We also used the hybrid architecture of Agent-SD which means ELFATT was used in early inference steps, and after that vanilla ToMe was applied for the remaining inference steps. The ratio of inference steps used for ELFATT is 20\%. The evaluation metric used is the FID score \cite{FID} which is used to evaluate the difference between 2000 generated images and 50000 images of the ImageNet-1K validation set.

\begin{figure}
\centering
\includegraphics[width=0.9\linewidth]{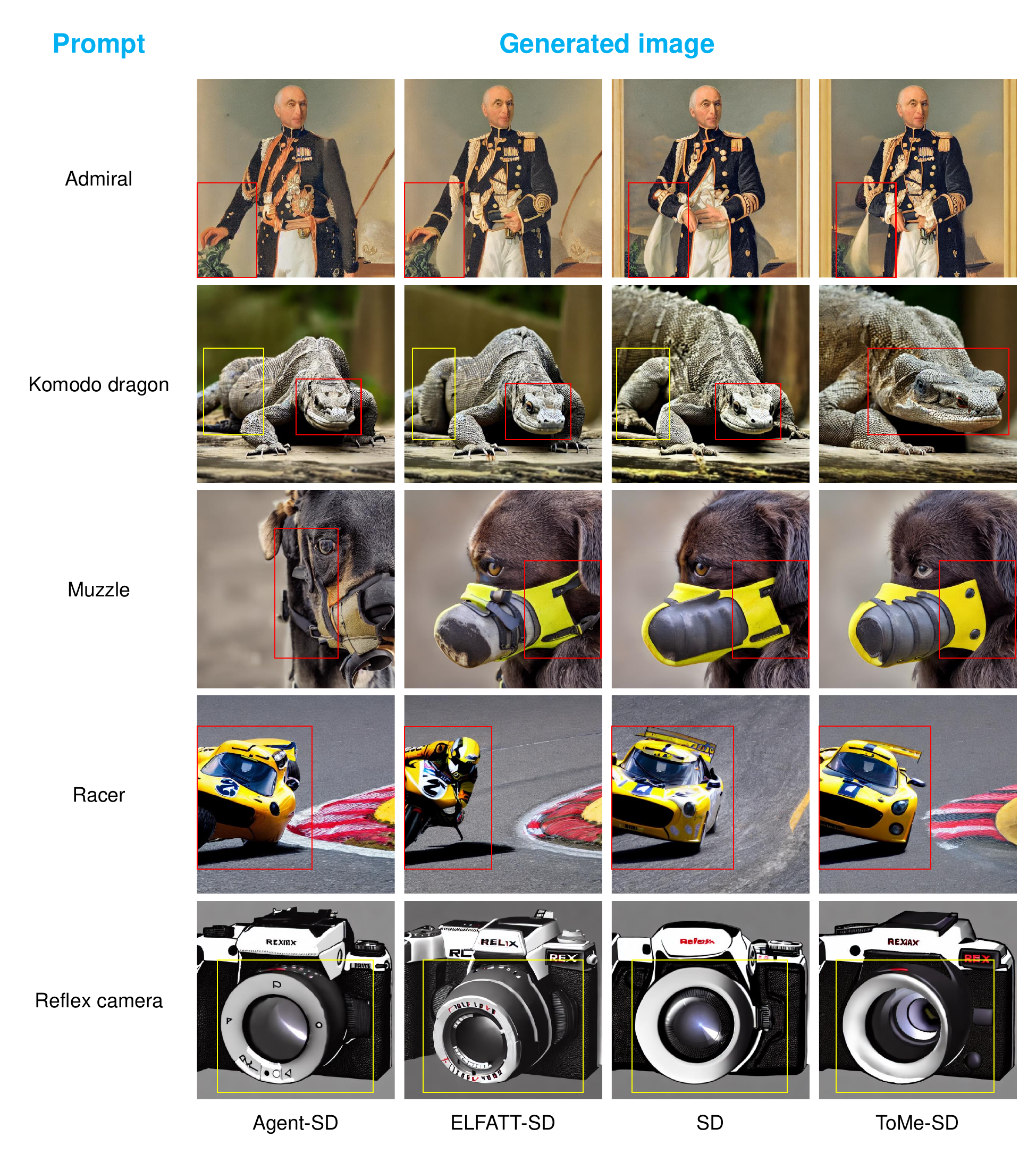}
\caption{The visual comparison of images generated by Agent-SD (merging ratio: 0.5), ELFATT-SD (merging ratio: 0.5), SD, and ToMe-SD (merging ratio: 0.5).}
\label{SDvis}
\end{figure}

\begin{table}[htbp]
\scriptsize
  \centering
  \caption{The comparison of SD using ELFATT in different inference steps with BF16 precision (merging ratio is 0.5).}
    \begin{tabular}{lrrr}
    \toprule
    Steps &  \multicolumn{1}{r}{FID $\downarrow$}  & \multicolumn{1}{r}{Time (s/img)} & \multicolumn{1}{r}{Memory (GB/img)}  \\
    \midrule
    0 (ToMe-SD) &  30.9       &   0.714           &     0.70          \\
    10          &  30.4       &   0.707           &     0.70          \\
    20          &  30.6       &   0.702           &     0.60          \\
    \bottomrule
    \end{tabular}
  \label{tabSDa2}
\end{table}

\subsection{Ablation Study of Different Block Sizes}
Table \ref{tabSDa1} shows the effect of different block sizes used in ELFATT on the performance of the stable diffusion task. When the block size increases, the speed is first improved and then declined. When the block size is 227, the speed is the fastest. The speed of the block size 128 is the second fastest and this block size has the lowest FID score. It seems that when the block size becomes smaller than 227, the inference speed is slowed, which may be caused by FlashAttention-2 and GPU architectures. In summary, the speed difference of the block size range from 32 to 512 is not significant. The reason is that: (1) Even the largest block size 512 is 8x smaller than the original sequence length 4096. (2) Only the first 20\% inference steps were applied to use ELFATT for acceleration. However, our method is still significantly faster than other methods, especially the baseline which shows the effectiveness of ELFATT for acceleration of SD.

\subsection{Ablation Study of Number of Inference Steps Using ELFATT}
Table \ref{tabSDa2} shows the effect of different inference steps using ELFATT on the performance of the stable diffusion task. When 0 inference steps using ELFATT are adopted, ELFATT-SD becomes ToMe-SD. With more inference steps using ELFATT being adopted, the speed becomes faster; however, the FID score is reduced first and then increased. When early 10 (20\%) inference steps using ELFATT are adopted, the FID score is the lowest.

\end{document}